\documentclass[aps,amsmath,amssymb,superscriptaddress,a4paper,onecolumn,accepted=2020-03-18]{quantumarticle}
\pdfoutput=1

\usepackage{graphicx}
\usepackage{dcolumn}
\usepackage{bm}
\usepackage{epsfig,hyperref,amsthm}
\usepackage{mathtools}
\usepackage{color,soul}
\usepackage{colordvi}
\usepackage[dvipsnames]{xcolor}
\usepackage{relsize}
\usepackage{accents}
\usepackage{wasysym}  
\usepackage{xypic}
\usepackage{enumitem}

\newcommand{\ket}[1]{\ensuremath{|#1\rangle}}
\newcommand{\bra}[1]{\ensuremath{\langle #1|}}

\newcommand{\proj}[1]{\ket{#1}\bra{#1}}
\newcommand{\be}{\begin{equation}}
\newcommand{\ee}{\end{equation}}
\newcommand{\ba}{\begin{eqnarray}}
\newcommand{\ea}{\end{eqnarray}}

\newcommand{\tket}[1]{\ket{\overline{#1}}}
\newcommand{\tbra}[1]{\bra{\overline{#1}}}
\newcommand{\tE}[1]{\overline{E}_{#1}}
\newcommand{\tproj}[1]{\ket{\overline{#1}}\bra{\overline{#1}}}

\newcommand{\tqx}[1]{\overline{Q}^{\rm X}_{#1}}

\newcommand{\E}{\mathcal{E}}

\newcommand{\bip}{\mathbb{C}^d\otimes\mathbb{C}^d}

\newcommand{\norm}[1]{\left\|#1\right\|}

\newcommand{\rIso}{\rho_{\rm iso}}
\newcommand{\id}{\mathbb{I}}

\newtheorem{theorem}{Theorem}

\newtheorem{alemma}{Lemma}[section]
\newtheorem{aproposition}[alemma]{Proposition}
\newtheorem{afact}[alemma]{Fact}
\newtheorem{acorollary}[alemma]{Corollary}
\newtheorem{atheorem}[alemma]{Theorem}
\newtheorem{adefinition}[alemma]{Definition}

\newtheorem{question}{Question}

\newtheorem{definition}{Definition}

\newcommand{\supA}{\overline{\sup_{\rm A}}\;}
\newcommand{\wt}[1]{\widetilde{#1}}
\newcommand{\mE}{\mathcal{E}}
\newcommand{\wtL}{\widetilde{\Lambda}}
\newcommand{\ep}{\epsilon}

\definecolor{nred}{rgb}{0.9,0.1,0.1}
\definecolor{nblack}{rgb}{0,0,0}
\definecolor{nblue}{rgb}{0.2,0.2,0.8}
\definecolor{ngreen}{rgb}{0.2,0.6,0.2}
\definecolor{ublue}{rgb}{0,0,0.5}
\definecolor{OliveGreen}{cmyk}{0.64,0,0.95,0.40}
\definecolor{pur}{rgb}{0.75,0,0.75}
\definecolor{nngrn}{rgb}{0,0.5,0.5}

\newcommand{\blu}{\color{nblue}}

\begin{document}
\title{Resource Preservability}

\author{Chung-Yun Hsieh}
\email{chung-yun.hsieh@icfo.eu}
\affiliation{ICFO - Institut de Ci\`encies Fot\`oniques, The Barcelona Institute of Science and Technology, 08860 Castelldefels, Spain}

\date{\today}

\begin{abstract} 
Resource theory is a general, model-independent approach aiming to understand the qualitative notion of resource quantitatively.
In a given resource theory, free operations are physical processes that do not create the resource and are considered zero-cost.
This brings the following natural question: For a given free operation, what is its ability to preserve a resource?
We axiomatically formulate this ability as the {\em resource preservability}, which is constructed as a channel resource theory induced by a state resource theory. 
We provide two general classes of resource preservability monotones: One is based on state resource monotones, and another is based on channel distance measures.
Specifically, the latter gives the robustness monotone, which has been recently found to have an operational interpretation.
As examples, we show that athermality preservability of a Gibbs-preserving channel can be related to the smallest bath size needed to thermalize all its outputs, and it also bounds the capacity of a classical communication scenario under certain thermodynamic constraints.
We further apply our theory to the study of entanglement preserving local thermalization (EPLT) and provide a new family of EPLT which admits arbitrarily small nonzero entanglement preservability and free entanglement preservation at the same time.
Our results give the first systematic and general formulation of the resource preservation character of free operations.
\end{abstract}

\maketitle

\section{Introduction}

An important goal in the study of physics is to understand and identify different {\em resources}: It may be an effect, an object, or a phenomenon, which enables us to achieve something that can never be achieved in its absence.
Before consuming the resource and triggering the advantages, one needs to make sure the given systems {\em have} the resource.
Hence, the first important question is: {\em How to probe it?}
Tremendous efforts have been made in this line of research for various resources.
For instance, the positive partial transpose criterion for entanglement is a representative result for entanglement detection~\cite{Ent-RMP,Peres1996,Horodecki1996}.
Also, various Bell inequalities and steering inequalities provide alternative ways of probing different quantum resources~\cite{Bell,Bell-RMP,Wiseman2007,Jones2007,steering-review,Steering-RMP}.

Knowing merely the existence of the resource is, however, insufficient for all applications.
This is because one may not only need the resource, but also need it to be strong enough: 
To demonstrate quantum advantages in teleportation~\cite{Bennett1993,Horodecki1999-2}, to witness a stronger than classical heat back-flow~\cite{Jennings2010}, or to violate a Bell/steering inequality, strong enough quantum correlations are necessary.
A quantitative understanding of qualitative resources is therefore crucial.
This question can be answered by a generic approach called {\em resource theory}, aiming to provide a general strategy to quantitatively formulate a given resource.

A resource theory can be interpreted as a triplet, consisting of the resource itself (e.g. entanglement), quantities without the resource (e.g. separable states), and physical processes that will not create the resource (e.g. local operation and classical communication channels~\cite{QCI-book}).
A resource theory provides a method to {\em quantify} the resource: With reasonable postulates, a {\em resource monotone} can be introduced, which can be interpreted as a quantifier attributing numbers to the resource content.
This important feature allows us to know more than just whether the resource exists or not: It also enables us to know the amount of resource.
Many results have been established in (but not limited to) the resource theories of entanglement~\cite{Ent-RMP,Vedral1997}, coherence ~\cite{Coherence-RMP,Baumgratz2014}, nonlocality~\cite{Bell-RMP,Wolfe2019}, steering~\cite{steering-review,Skrzypczyk2014,Piani2015,Gallego2015,Steering-RMP}, asymmetry~\cite{Gour2008,Marvian2016}, and athermality~\cite{Brandao2013,Brandao2015,Horodecki2013,Lostaglio2018,Serafini2019,Narasimhachar2019}.
Also, several general features of resource theories have been reported~\cite{Horodecki2013-2,Brandao2015-2,del_Rio2015,Coecke2016,Gour2017,Anshu2018,Regula2018,Bu2018,Liu2017,Lami2018,RT-RMP,Takagi2019,Takagi2019-2,Liu2019,Fang2019}. 
Notably, resource theories of {\em quantum channels} (or simply {\em channels}, which are also known as {\em completely-positive trace-preserving} maps~\cite{QCI-book}) and related topics have drawn much attention recently~\cite{Hsieh2017,Kuo2018,Pirandola2017,Dana2017,Bu2018,Wilde2018,Diaz2018,Zhuang2018,Gour2019-3,Theurer2019,Seddon2019,Rosset2018,LiuWinter2019,LiuYuan2019, Gour2019,Gour2019-2,Bauml2019,Wang2019,Berk2019,Takagi2019-3}.

\begin{center}
\begin{figure*}
\scalebox{0.8}{\includegraphics{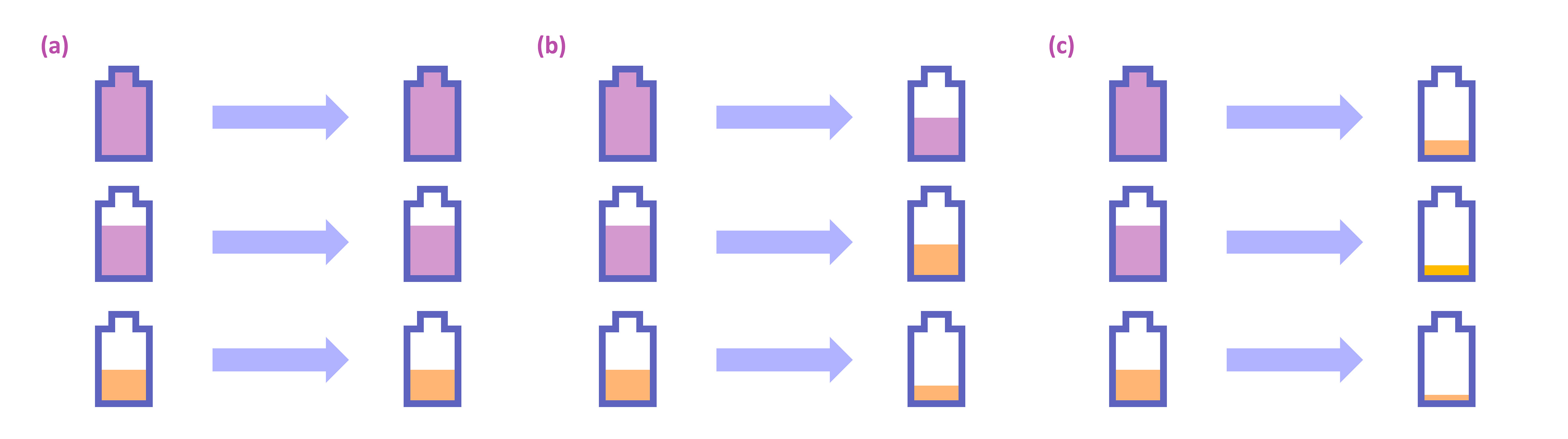} }
\caption{
In this work, we say a channel preserves a resource if it does not completely destroy the resource for every inputs.
To illustrate this, consider the above figure.
Each battery icon represents a state, and the height of the colored column is a resource measure.
The purple color indicates the existence of the resource, and the yellow means the state is free.
Then we have three examples: (a) The channel maintains the resourcefulness for every input.
(b) The channel partially degrades the resource content but can still output the resource for certain inputs. 
(c) The channel totally destroys resource for every input.
In this work, channels in (a) and (b) are considered to have the ability to preserve the resource, but not the one in (c).
}
\label{Fig:R-Preservability} 
\end{figure*}
\end{center}

One important ingredient in a resource theory is the allowed physical processes that will not create the resource, which are called {\em free operations}.
An ultimate goal for a resource theory is to identify under which conditions a quantity can be transformed into another via free operations.
A proper answer can tell us how resourceful the output quantities can be after free operations, giving useful information for both theoretical and practical purposes.
This is conceptually related to channel's ability to preserve a resource, which is a phenomenon lacking a quantitative understanding.
To be precise, we say a channel preserves a resource if it does not completely destroy the resource for every input.
In other words, it can partially degrade the resource, while there must be certain output states that are resourceful (see Fig.~\ref{Fig:R-Preservability}).
This motivates us to ask the following question:
\begin{center}
{\em
Given a free operation, how to quantify its ability to preserve the given resource?
}
\end{center}
In other words, we are asking for a quantitative study of the qualitative behavior (i.e.\;the ability to preserve the resource) of free operations, which can be interpreted as a resource theory inherited from the given resource theory.
With a rigorous answer, one will be able to identify the efficiency of the given free operation to protect the resource, which will clarify the fundamental structure of free operations in a general resource theory.
This question is also motivated by other purposes: For example, a suitable measure of the ability of a given dynamics to preserve entanglement can provide new insights to the study of the interplay between entanglement and thermalization~\cite{Hsieh2019}.
Also, some previous results have addressed similar issues for entanglement~\cite{Moravcíkova2010}, while a general treatment for free operations with arbitrary state resources is still unknown.

In this work, we axiomatically formulate the ability of free operations to preserve a resource of quantum states.
This ability, termed {\em resource preservability}, is formulated as a {\em channel} resource theory induced by the given {\em state} resource theory.
We provide general assumptions of the formulation, discussing the corresponding free operation, and introducing axioms on the resource preservability monotones.

Two classes of resource preservability monotones are provided: One is induced by the resource monotones of the given state resource theory, with the intuition behind as the maintained resource during the process; another is based on the channel distance from the set of free operations that will destroy the resource.
Moreover, the one based on channel distance will induce a robustness-like monotone, with an operational interpretation as the erasure cost of resource preservability due to Ref.~\cite{LiuWinter2019}.

As an example, we consider the resource theory of athermality and show that an resource preservability monotone of a given Gibbs-preserving channel is related to the smallest bath size needed to thermalize all its output states.
We further show that the robustness-like monotone serves as an upper bound of the classical capacity of a classical communication scenario subject to certain thermodynamic constraints.
These connect thermodynamics and classical communication to our current study.
As another application, we apply our theory to the study of {\em entanglement preserving local thermalizations} (EPLTs)~\cite{Hsieh2019}, which are local operation plus shared randomness channels that can locally thermalize subsystems for arbitrary inputs, while keep the global entanglement for certain inputs.
We show that EPLTs can admit arbitrarily small entanglement preservability at finite temperatures and preservation of free entanglement~\cite{Horodecki1998} simultaneously.
This reveals the fact that EPLT is a concept compatible with arbitrarily small ability of entanglement preservation, and can still preserve distillable entanglement at the same time.

This work is structured as follows.
We start with basic notions of a general state resource theory and general setup of resource preservability in Sec.~\ref{Sec:ResourcePreservability}.
After the formal setup, we formulate free super-channel in Sec.~\ref{Sec:FreeSuperChannel}, and in Sec.~\ref{Sec:ResourcePreservabilityMonotone} we axiomatically introduce resource preservability monotones.
In Sec.~\ref{Sec:All-Application}, we consider examples with the resource theory of athermality and apply the theory of resource preservability to EPLT.
Finally, we conclude in Sec.~\ref{Sec:Conclusion}.

\section{Setup and Assumptions}\label{Sec:ResourcePreservability}
A {\em resource theory} of quantum states, or simply a {\em state resource theory}, can be understood as a combination of the following three ingredients: The resource itself (denoted by $R$), states without the resource (the {\em free states}; denote the set of all free states by $\mathcal{F}_R$), and channels that can be applied freely and cannot create the resource (the {\em free operations}; denote the set of all free operations by $\mathcal{O}_R$).
Hence, a state resource theory can be written as the triplet
$
(R,\mathcal{F}_R,\mathcal{O}_R).
$
A {\em channel resource theory} can be defined in a similar way with a state resource theory by replacing states by channels, and the corresponding free operations ($\mathcal{O}_R$) will be super-channels~\cite{Chiribella2008,Chiribella2008-2}.

In this work, the only class of channel resource theories will be the one of resource preservability induced by different state resource theories. 
Hence, for convenience, from now on {\em $R$-theory} means the resource theory of the given resource $R$ of {\em quantum states}.
The corresponding channel resource theory of resource preservability (abbreviated as {\em $R$-preservability}) will be called an {\em $R$-preservability theory}.

To formulate $R$-preservability as a channel resource theory inherited from a given $R$-theory, the first thing is to identify the free channels.
To this end, we consider free operations of the given $R$-theory that cannot preserve resource for {\em every} input:
\begin{align}
\mathcal{O}_R^N\coloneqq\{\mathcal{E}\in \mathcal{O}_R\;|\;\mathcal{E}(\rho)\in \mathcal{F}_R\;\forall\,\rho\}.
\end{align}
Channels of this kind will be called {\em resource-annihilating channels} (abbreviated as {\em $R$-annihilating channels}) which is inspired by the name of entanglement-annihilating channel~\cite{Moravcíkova2010}.
This set gives the free channels of the $R$-preservability theory.
In view of this notion, every element in $\mathcal{O}_R\setminus\mathcal{O}_R^N$ will be understood to have certain ability to preserve the given resource\footnote{We remark that the setting here is consistent while not the same with the channel resource theory introduced in Ref.~\cite{LiuYuan2019}: Since in the study of $R$-preservability the identity channel will be the most resourceful one, some results of Refs.~\cite{LiuYuan2019,LiuWinter2019} cannot apply.
Also, our approach is genuinely different from the resource destroying maps~\cite{Liu2017}, which leave free states invariant and map resourceful states to some free states.} (see also Fig.~\ref{Fig:R-Preservability}).

It remains to specify the corresponding free operations and quantifiers of $R$-preservability, which are the main tasks of this work.
Before that, we need to impose some basic assumptions and constraints on the given $R$-theory in order to have a reasonable study.

At the beginning of the formulation, one may wonder whether we should assume the following property in a bipartite system ${\rm SS'}$:
\begin{center}
{\em
$\Lambda_{\rm S}\otimes\Lambda_{\rm S'}\in\mathcal{O}_R^N$ if $\Lambda_{\rm S},\Lambda_{\rm S'}\in\mathcal{O}_R^N$?
}
\end{center}
This property forbids any possibility to activate the $R$-preservability.
This is, however, not true due to the existence of activation properties of certain resources~\cite{Palazuelos2012,Hsieh2016,Quintino2016,Masanes2008,Liang2012}.
More precisely, in Appendix~\ref{Sec:Activation} we show that in some $R$-theories, one can construct a free operation $\wt{\mathcal{T}}\in\mathcal{O}_R^N$ such that $\wt{\mathcal{T}}^{\otimes k}\notin\mathcal{O}_R^N$ for some integer $k>0$.
This means if we want to formulate $R$-preservability theory in a general way applicable to different $R$-theories, we need to respect certain properties such as the activation of the $R$-preservability.
To impose basic assumptions on $R$-theory, we need the following concept first:
\begin{definition}
{\em (Absolutely Free State)}
A free state $\widetilde{\eta}$ is said to be an {\em absolutely free state} for the given $R$-theory if 
\begin{align}
\widetilde{\eta}\otimes\eta\in\mathcal{F}_R\quad\forall\;\eta\in\mathcal{F}_R.
\end{align}
We denote the set of all absolutely free states by $\widetilde{\mathcal{F}}_R$.
\end{definition}

In other words, absolutely free states are those without hidden resource~\cite{Masanes2008,Liang2012}.
For example, in the $R$-theory of entanglement, all the separable states are absolutely free states.
However, as we have mentioned, there also exist $R$-theories with states that are not absolutely free: This can be seen by the superactivation of nonlocality~\cite{Palazuelos2012} and steering~\cite{Hsieh2016,Quintino2016}.
We remark that $\wt{\mathcal{F}}_R$ is closed under tensor product; that is, $\wt{\eta}_1\otimes\wt{\eta}_2\in\wt{\mathcal{F}}_R$ if $\wt{\eta}_1,\wt{\eta}_2\in\wt{\mathcal{F}}_R$.

With the above notion, we consider $R$-theories with the following properties in this work:
\begin{enumerate}[label = (R\arabic*)]
\item\label{Def:Proper:Nonempty} $\wt{\mathcal{F}}_R\neq\emptyset$ and $\mathcal{F}_R$ is convex.
\item\label{Def:Proper:FreeIdentity} Identity and partial trace are free operations.
\item\label{Def:ProperQR-Tensor} Tensoring with absolutely free states [i.e.\,$(\cdot)\mapsto(\cdot)\otimes\wt{\eta}$ for a given $\wt{\eta}\in\wt{\mathcal{F}}_R$] are free operations.
\item\label{Def:Proper:Tensor} Free operations are closed under tensor products, convex sums, and compositions:
If $\mathcal{E}_1,\mathcal{E}_2\in\mathcal{O}_R$, then $\mathcal{E}_1\otimes\mathcal{E}_2\in\mathcal{O}_R$, $p\mathcal{E}_1 + (1-p)\mathcal{E}_2\in\mathcal{O}_R$ $\forall\;p\in[0,1]$, and $\mathcal{E}_1\circ\mathcal{E}_2\in\mathcal{O}_R$.
\end{enumerate}
Let us briefly comment on the above properties.
We assume property~\ref{Def:Proper:Nonempty} because we aim to study $R$-preservability, which is a comparison of  resourceless states and resourceful states.
Also, we expect genuinely resourceless states exist and convex sums of resourceless states will not be resourceful, which are common features shared by many $R$-theories.
Property~\ref{Def:Proper:FreeIdentity} is assumed because in an $R$-theory, identity map and partial trace can never increase the amount of resource and will usually fulfill other conditions of a free operation: Conceptually, it means ``doing noting'' and ``ignoring part of the system'' are both free and costless.
Property~\ref{Def:ProperQR-Tensor} makes sure the resource content will not increase after an extension with an absolutely free state $\wt{\eta}$.
Property~\ref{Def:Proper:Tensor} is imposed because, for two channels which cannot create the resource, we expect 
their simultaneous applications (tensor product), classical mixture (convex sum), and sequential applications (composition) still will not have the ability to create the resource.
As expected, it is a common property possessed by many choices of free operations in $R$-theories such as the ones of entanglement~\cite{Ent-RMP}, nonlocality~\cite{Rosset2018,Wolfe2019}, and athermality~\cite{Brandao2013,Lostaglio2018} (note that there do exist examples which cannot satisfy this property\footnote{To see a counterexample, consider the $R$-theory of nonlocality with the nonlocality non-generating channels as free operations.
Suppose $\rho_0$ is a local state such that $\rho_0^{\otimes 2}$ is nonlocal~\cite{Palazuelos2012}. 
Then the state preparation channel $\Phi_{\rho_0}:(\cdot)\mapsto\rho_0$ is a nonlocality non-generating channel, while $\Phi_{\rho_0}\otimes\Phi_{\rho_0}$ will always have nonlocal output, thereby being able to generate nonlocality.
This again shows the activation property may lead to unexpected results.}).
This also implies that in this work the set $\mathcal{O}_R^N$ is always convex.

Before the formulation of $R$-preservability, it is important to introduce the following analog concept of absolutely free states for channels.

\begin{definition}
{\em (Absolutely $R$-Annihilating Channel)}
We say $\widetilde{\Lambda}\in\mathcal{O}_R^N$ is an {\em absolutely $R$-annihilating channel} if
\begin{align}
\widetilde{\Lambda}\otimes\Lambda\in\mathcal{O}_R^N\quad\forall\,\Lambda\in\mathcal{O}_R^N.
\end{align}
We denote the set of all such channels by $\widetilde{\mathcal{O}}_R^N$.
\end{definition}
This definition means the $R$-preservability of absolutely $R$-annihilating channels cannot be activated.
As an example of an absolutely $R$-annihilating channel, consider again the $R$-theory of entanglement.
Then every {\em local operation and classical communication} (LOCC) channel that is entanglement-annihilating~\cite{Moravcíkova2010} and entanglement-breaking~\cite{Horodecki2003} will be absolutely $R$-annihilating channels (see Appendix~\ref{App:EntExample} for the detailed explanation).

We also remark the following facts for a given $R$-theory:
\begin{align}\label{Fact:wtL}
&\wtL\circ\mE\in\wt{\mathcal{O}}_R^N\quad\&\quad\mE\circ\wtL\in\wt{\mathcal{O}}_R^N\quad\forall\;\mE\in\mathcal{O}_R;\nonumber\\
&\wtL_{\rm S}\otimes\wtL_{\rm S'}\in\wt{\mathcal{O}}_R^N\quad\forall\;\wtL_{\rm S},\wtL_{\rm S'}\in\wt{\mathcal{O}}_R^N.
\end{align}
According to the first line in Eq.~\eqref{Fact:wtL}, a sequential application of free operations cannot preserve any resource (even with the assistance of ancillary $R$-annihilating channels) if one has already added one absolutely $R$-annihilating channel in the sequence.
Also, since absolutely $R$-annihilating channels forbid activation, the second line in Eq.~\eqref{Fact:wtL} means simultaneous applications of two such channels still do not allow activation.

Before introducing the main results, we specify notations.
In this work we ignore the dependency of system size of the notations $\mathcal{O}_R^N$ and $\mathcal{O}_R$.
To emphasize the contrast between the main systems and ancillary systems, we use subscripts ${\rm S,S'}$ for the main systems and A,B for the ancillary systems.
When only bipartition needs to be addressed, we use the common notations A,B for subsystems.
The meaning of subscripts will be clear from the context.

\section{Free Operation of Resource Preservability}\label{Sec:FreeSuperChannel}
To specify the free operation of $R$-preservability, we need to know first how to map a channel in $\mathcal{O}_R$ into another channel in $\mathcal{O}_R$.
The general structure of such mapping (which maps channels to channels) is shown to take the following form~\cite{Chiribella2008}:
\begin{align}\label{Eq:Super-Channel}
\mathcal{E}\mapsto\mathcal{M}\circ(\mathcal{E}\otimes\mathcal{I}_{\rm A})\circ\mathcal{N},
\end{align}
where ${\rm A}$ stands for the ancillary system, and $\mathcal{M}$, $\mathcal{N}$ are some quantum channels.
Such mappings are called {\em super-channels}~\cite{Chiribella2008,Chiribella2008-2,RT-RMP,LiuWinter2019}.
One potential way to introduce free operations of $R$-preservability, or simply {\em free super-channels}, is to consider all super-channels that will not increase $R$-preservability.
This gives the largest possible set of free super-channels, while it may not always have intuitive and clear physical interpretation (see Ref.~\cite{Takagi2019-3} for an exception).
Also, whether all such mappings can always map elements of $\mathcal{O}_R$ into $\mathcal{O}_R$ is still unclear\footnote{\label{App:Well-defined}We remark that the structure of free supper channels depends on the structure of $\mathcal{O}_R$. For example, in the $R$-theory of entanglement, if we set $\mathcal{O}_R$ to be all LOSR channels (see Appendix~\ref{App:LOSR} for the definition), then LOCC channels are outside our consideration.
This means a super-channel that maps some LOSR channels into LOCC channels will not be a suitable free super-channel in this case.
Hence, the set of all super-channels that will not generate $R$-preservability may not always be a well-defined set of free super-channels of $R$-preservability.}.
Hence, in this work we prefer a different approach: We try to impose conditions on Eq.~\eqref{Eq:Super-Channel} and focus on free super-channels with clear physical meanings.

To this end, we interpret Eq.~\eqref{Eq:Super-Channel} as a three-step process consisting of a pre-processing ($\mathcal{N}$), an ancillary process ($\mathcal{I}_{\rm A}$), and a post-processing ($\mathcal{M}$).
The first condition to be imposed is that free super-channels should be realized freely in the given $R$-theory, since we expect them to be implementable without the assistance of the resource $R$.
This suggests that all steps in Eq.~\eqref{Eq:Super-Channel} should be free operations of the given $R$-theory; that is, $\mathcal{N},\mathcal{M}\in\mathcal{O}_R$.
The second condition to be imposed is that free super-channels cannot create $R$-preservability.
However, since identity map has the best $R$-preservability, this may fail if one uses identity map for the ancillary process in Eq.~\eqref{Eq:Super-Channel}.
This suggests that the ancillary system should perform certain processes to ensure it is impossible to create $R$-preservability.
Concerning the existence of activation properties discussed in Appendix~\ref{Sec:Activation}, we ask the ancillary system to perform only absolutely $R$-annihilating channels.
The above discussions motivate us to consider the following notion as the free operation of an $R$-preservability theory in this work:
\begin{definition} \label{Def:Free-Super-Channel}
{\rm (Free Super-Channel of $R$-Preservability)}
In this work, the free operation of $R$-preservability, or say the {\em free super-channel} $F:\mathcal{O}_R\to\mathcal{O}_R$, is of the form
\begin{align}\label{Eq:Free_Super-Channel}
F_\mathcal{E}\coloneqq\Lambda_+\circ(\mathcal{E}\otimes\widetilde{\Lambda}_{\rm A})\circ\Lambda_-,
\end{align}
where $\Lambda_+,\Lambda_-\in\mathcal{O}_R$ are free operations of the $R$-theory and $\widetilde{\Lambda}_{\rm A}\in\widetilde{\mathcal{O}}_R^N$ is an absolutely $R$-annihilating channel.
\end{definition}
For the generality of the $R$-preservability theory, we allow different input/output dimensions of the free super-channels\footnote{One can also formulate the theory with the fixed system dimension and forbid the ancillary systems, while in our approach we prefer a more general version.
This is similar to the case of channel discrimination: One can use either trace norm or diamond norm.
The trace norm gives an intuitive description of channel discrimination with the focus only on the given system, while the performance can be improved when one switches to the diamond norm.
In this work, we try to capture the spirit of the latter.}, which means the $R$-preservability of the given channel on the main system S may be assisted by channels acting on ancillary systems, while the ancillary channels need to obey the rules: They cannot provide additional $R$-preservability, and they cannot be assisted by the given state resource $R$.

Note that if one simply assumes $\Lambda_+,\Lambda_-$ to possess zero $R$-preservability, then the output will only be $R$-annihilating channels.
Hence, we allow $\Lambda_+,\Lambda_-$ to be arbitrary free operations.
Also, we have $F_\Lambda\in\mathcal{O}_R^N$ if $\Lambda\in\mathcal{O}_R^N$, which is because $\widetilde{\Lambda}_{\rm A}\in\widetilde{\mathcal{O}}_R^N$. 
This ensures that Eq.~\eqref{Eq:Free_Super-Channel} is a suitable free operation even with the activation property of $R$-preservability (Appendix~\ref{Sec:Activation}).

\section{Resource Preservability Monotone}\label{Sec:ResourcePreservabilityMonotone}
An important feature of a resource theory is that it provides a way to quantify the resource~\cite{RT-RMP}. Let $\mathbb{Q}$ be the set of all states or all channels.
Then a {\em resource monotone} of the given resource $R$ is a function \mbox{$Q_R:\mathbb{Q}\to[0,\infty]$} satisfying properties~\ref{P:Positivity} and~\ref{P:Monotonicity}:
\begin{enumerate}[label = (M\arabic*)]
\item\label{P:Positivity} $Q_R(q)\ge0\quad\forall\,q\in\mathbb{Q}$ and $Q_R(q)=0$ if $q\in\mathcal{F}_R$. 
\item\label{P:Monotonicity} $Q_R[\Lambda(q)]\le Q_R(q)\quad\forall\,q\in\mathbb{Q}\quad\&\quad\forall\Lambda\in\mathcal{O}_R$.
\item\label{P:Convexity} $Q_R[pq_1 + (1-p)q_2]\le pQ_R(q_1) + (1-p)Q_R(q_2)$ $\forall\,q_1,q_2\in\mathbb{Q}\quad\&\quad\forall p\in[0,1]$.
\item\label{P:Faithful} $Q_R(q)=0$ if and only if $q\in\mathcal{F}_R$.
\end{enumerate}
It is called {\em convex} if it also satisfies property~\ref{P:Convexity}, and it is called {\em faithful} if it also satisfies property~\ref{P:Faithful}.
To avoid trivial case, we always assume $Q_R(q)>0$ for some $q$ in this work.
With the above notions, we are now in position to introduce the $R$-preservability monotones.
\begin{definition}\label{Def:R-Preservability}
{\rm (Resource Preservability Monotone)}
In an $R$-preservability theory,
an {\em $R$-preservabil-}{\em ity monotone} $P_R$ is a channel resource monotone satisfying the following additional property:
\begin{align}\label{P:Tensor} 
P_R(\mathcal{E}\otimes\mathcal{E}')\ge P_R(\mathcal{E})\quad\forall\;\mathcal{E},\mathcal{E}'\in\mathcal{O}_R,
\end{align}
and the equality holds if $\mathcal{E}'\in\widetilde{\mathcal{O}}_R^N$.
\end{definition}
This additional property illustrates the basic expectation of a good quantifier of $R$-preservability: $R$-preservability will not decrease under tensor product, and it will not increase under tensor product with absolutely $R$-annihilating channels.
Note again that we do not impose the property $P_R(\mathcal{E}_{\rm S}\otimes\Lambda_{\rm S'})\le P_R(\mathcal{E}_{\rm S})\quad\forall\,\Lambda_{\rm S'}\in\mathcal{O}_R^N\;\&\;\mathcal{E}_{\rm S}\in\mathcal{O}_R$ due to the existence of the activation property discussed in Appendix~\ref{Sec:Activation}.
It is still possible for an $R$-preservability monotone to satisfy this property, which simply means that monotone cannot witness activated $R$-preservability.

We introduce two classes of $R$-preservability monotones, whose underlying intuitions are stated as follows:
\begin{itemize}
\item Interpret $R$-preservability as {\em the ability to maintain resource during the operation}.
\item Interpret $R$-preservability as {\em the channel distance from the set of $R$-annihilating channels}.
\end{itemize}
While they originate from different concepts, in the following sections we will show that both of them admit $R$-preservability monotones.

\subsection{Resource Preservability Monotone: The Maintained Resource}\label{Sec:PQ_R}
For a given resourceful state $\rho$ and a given state resource monotone $Q_R$, an intuitive way to quantify the ability of a free operation $\mathcal{E}_{\rm S}$ to preserve the resource $R$ of $\rho$ is to compare the difference between $Q_R(\rho)$ and $Q_R[\mathcal{E}_{\rm S}(\rho)]$; that is, $\frac{Q_R[\mathcal{E}_{\rm S}(\rho)]}{Q_R(\rho)}$.
This proposes the following general candidate induced by $Q_R$: (we use subscript to denote the corresponding subsystems)
\begin{align}\label{Eq:GeneralMonotone}
P_{Q_R}^{(f,g)}(\mathcal{E}_{\rm S})\coloneqq\overline{\sup_{{\rm A}}}\;\frac{(f\circ Q_R)[(\mathcal{E}_{\rm S}\otimes \widetilde{\Lambda}_{\rm A})(\rho_{\rm SA})]}{(g\circ Q_R)(\rho_{\rm SA})},
\end{align}
where $f$ is a finite-valued strictly increasing function with $f(0) = 0$, $g$ is a non-decreasing function satisfying $g^{-1}(\{0\})\subseteq\{0\}$ [this means the only $x$ that may achieve $g(x)=0$ is $x=0$].
Here we use the following abbreviation:
\begin{align}\label{Eq:supA}
\overline{\sup_{\rm A}}\coloneqq\sup_{{\rm A};\widetilde{\Lambda}_{\rm A}\in\widetilde{\mathcal{O}}_R^N;\rho_{\rm SA}},
\end{align}
where the maximization is taken over all possible finite dimensional ancillary systems ${\rm A}$, all absolutely $R$-annihilating channels $\widetilde{\Lambda}_{\rm A}\in\widetilde{\mathcal{O}}_R^N$ on the ancillary system {\rm A}, and all states $\rho_{\rm SA}$ on the composite system ${\rm SA}$.
In the maximization we allow the ancillary system to have zero dimension, corresponding to the original system ${\rm S}$.
We stress that the maximization in Eq.~\eqref{Eq:GeneralMonotone} is restricted to $\rho_{\rm SA}$ achieving non-zero $Q_R$ values.
This makes sure the value is always finite.

The idea behind Eq.~\eqref{Eq:GeneralMonotone} is to consider a general ratio between the input and the output of the given free operation.
By considering particular combinations of $f$ and $g$, we have the following candidates:
\begin{align}
\overline{\sup_{{\rm A}}}\,\frac{Q_R[(\mathcal{E}\otimes\widetilde{\Lambda}_{\rm A})(\rho_{\rm SA})]}{Q_R(\rho_{\rm SA})}\;;\;\overline{\sup_{{\rm A}}}\,Q_R[(\mathcal{E}\otimes\widetilde{\Lambda}_{\rm A})(\rho_{\rm SA})].
\end{align}
The first one can be interpreted as the optimal maintained resource during the process $\mathcal{E}$, and the second one can be understood as the optimal remaining amount of resource in the end of the process $\mathcal{E}$.

Note that we do not use identity map $\mathcal{I}_{\rm A}$ for the ancillary systems in the above definition.
This is because identity channel is the {\em most resourceful} channel, and considering ancillary system with it may create ``artificial $R$-preservability''.
For example, if one uses identity for the ancillary systems in the $R$-theory of entanglement, then one will have non-zero $R$-preservability for entanglement-annihilating channels that are not entanglement-breaking~\cite{Moravcíkova2010}.
Merely using $R$-annihilating channels $\mathcal{O}_R^N$ for the extension is still not enough due to the existence of the activation property (Appendix~\ref{Sec:Activation}).
This explains the need of introducing absolutely $R$-annihilating channels.

We now present the first main result, whose proof is given in Appendix~\ref{App:Proof_Result:Maintain}. 
Recall that $R$-theory represents a state resource theory with resource $R$.
\begin{theorem}\label{Result:Maintain}
Given an $R$-theory and a state resource monotone $Q_R$.
Then $P_{Q_R}^{(f,g)}$ defined by Eq.~\eqref{Eq:GeneralMonotone} is an \mbox{$R$-preservability} monotone. 
Moreover, It is faithful if $Q_R$ is faithful, and it is convex if $f\circ Q_R$ is convex.
\end{theorem}
As a remark, the assumption $\wt{\mathcal{F}}_R\neq\emptyset$ is only used in the proof of Eq.~\eqref{P:Tensor}, and this assumption can be dropped when $g$ is a positive constant.
We state this special case in Corollary~\ref{Coro:g-constant}.
Also, it will be an interesting future research topic to study specific operational interpretations of different combinations of $f,g$ with different $R$-theories.

\subsection{Resource Preservability Monotone: The Channel Distance}
One intuitive way to quantify a resource is to consider the distance away from the set consisting of quantities without the resource.
Here we use the similar way to interpret $R$-preservability.
To this end, we consider a general distance measure on states defined as a function $D:\mathcal{S}\times\mathcal{S}\to[0,\infty]$ satisfying $D(\rho,\sigma)\ge0$ and equality holds if and only if $\rho = \sigma$ ($\mathcal{S}$ is the set of quantum states).
Now, we introduce the following candidates induced by $D$ to quantify $R$-preservability:
\begin{align}\label{Eq:P_D}
&P_{D}(\mathcal{E})\coloneqq\inf_{\Lambda_{\rm S}\in\mathcal{O}_R^N}\supA D\left[(\mathcal{E}\otimes\widetilde{\Lambda}_{\rm A})(\rho_{\rm SA}),(\Lambda_{\rm S}\otimes\widetilde{\Lambda}_{\rm A})(\rho_{\rm SA})\right];\\
\label{Eq:P_Dbar}
&\bar{P}_{D}(\mathcal{E})\coloneqq\inf_{\Lambda_{\rm S}\in\mathcal{O}_R^N}\sup_{{\rm A};\rho_{\rm SA}} D\left[(\mathcal{E}\otimes\mathcal{I}_{\rm A})(\rho_{\rm SA}),(\Lambda_{\rm S}\otimes\mathcal{I}_{\rm A})(\rho_{\rm SA})\right].
\end{align}
Again, we use the abbreviation introduced in Eq.~\eqref{Eq:supA}, and $\sup_{{\rm A};\rho_{\rm SA}}$ means the maximization taken over all the ancillary systems ${\rm A}$ and the states $\rho_{\rm SA}$ on ${\rm SA}$.
Note that unlike the previous section, since now we only compare the {\em distance} between two channels, using identity to extend the system is allowed, and this is the reason why we list two candidates here.
Before introducing the main result, we say a set {\em $\mathcal{A}$ is closed under the distance measure $D$} if for every sequence $\{\Lambda_k\}_{k=1}^\infty\subseteq\mathcal{A}$ satisfying $\lim_{k\to\infty}\sup_\rho D[\mathcal{E}(\rho),\Lambda_k(\rho)] = 0$, we will have $\mathcal{E}\in\mathcal{A}$.
We now provide the following result, whose proof is given in Appendix~\ref{App:Proof_Result:DistanceMonotone}.
\begin{theorem}\label{Result:DistanceMonotone}
Given an $R$-theory and a distance measure $D$ satisfying the property 
\begin{align}\label{Def:ProperDistanceMeasure}
D[\Lambda(\rho),\Lambda(\sigma)]\le D(\rho,\sigma)\quad\forall\;\rho,\sigma\;\&\;\forall\;\Lambda\in\mathcal{O}_R.
\end{align}
Then $P_D$ and $\bar{P}_{D}$ are \mbox{$R$-preservability} monotones. 
Moreover, they are faithful if $\mathcal{O}_R^N$ is closed under $D$.
\end{theorem}
Note that Eq.~\eqref{Def:ProperDistanceMeasure} is a relaxed version of the data-processing inequality.
As a remark, Eq.~\eqref{Def:ProperDistanceMeasure} and condition~\ref{Def:Proper:Tensor} imply the ordering $P_D\le\bar{P}_{D}$.

\subsection{Resource Preservability Monotone: The Robustness}\label{Sec:Robustness}
We will provide a detailed example in this section to illustrate Theorem~\ref{Result:DistanceMonotone}.
In short, with a specific distance measure, a robustness-like monotone can be obtained.
To start with, consider the {\em max-relative entropy} defined by~\cite{Datta2009}:
\begin{align}\label{Eq:max-relative-entropy}
D_{\rm max}(\rho\|\sigma)\coloneqq \log_2\inf\{\lambda\,|\,\rho\le\lambda\sigma\},
\end{align}
where the minimization is taken over all non-negative integer $\lambda$, and in this work we always consider logarithm to the base 2.
$D_{\rm max}$ fulfills~\cite{Datta2009} (1) $D_{\rm max}(\rho\|\sigma)\ge0$ and the equality holds if and only if $\rho=\sigma$, (2) (data-processing inequality) $D_{\rm max}[\mathcal{E}(\rho)\|\mathcal{E(\sigma)}]\le D_{\rm max}(\rho\|\sigma)$ for all channels $\mathcal{E}$ and states $\rho,\sigma$.
Hence, it satisfies Eq.~\eqref{Def:ProperDistanceMeasure}.
Theorem~\ref{Result:DistanceMonotone} means $P_{D_{\rm max}}$ and $\bar{P}_{D_{\rm max}}$ are both $R$-preservability monotone, and they are faithful if $\mathcal{O}_R^N$ is closed under $D_{\rm max}$.

It turns out that this fact implies a direct robustness form and the corresponding operational interpretation based on Ref.~\cite{LiuWinter2019}.
To see this, define the {\em $R$-preservability log-robustness} according to Ref.~\cite{LiuWinter2019}:
\begin{align}\label{Eq:Robustness}
L_R(\mathcal{E}) = -\log_2\sup\{p\in[0,1]\,|\,p\mathcal{E} + (1-p)\mathcal{C}\in\mathcal{O}_R^N\},
\end{align}
where the optimization is taken over all channels $\mathcal{C}$.
This quantity depicts how robust the $R$-preservability of $\mathcal{E}$ is when it is interrupted by another channel.
From Ref.~\cite{LiuWinter2019} we learn that $\bar{P}_{D_{\rm max}} = L_R$.
This means $\bar{P}_{D_{\rm max}}$ may have the same operational interpretation with $L_R$.
To formally illustrate this, we now translate the Definition 9 in Ref.~\cite{LiuWinter2019} into the following version for $R$-preservability (in what follows, the {\em diamond norm} is defined by $\norm{\mE_{\rm S}}_\diamond\coloneqq\sup_{{\rm A};\rho_{\rm SA}}\norm{(\mE_{\rm S}\otimes\mathcal{I}_{\rm A})(\rho_{\rm SA})}_1$, where the maximization is taken over all ancillary systems ${\rm A}$ and states $\rho_{\rm SA}$ on the system ${\rm SA}$, and $\norm{\rho}_1\coloneqq{\rm tr}|\rho| = {\rm tr}\sqrt{\rho^\dagger\rho}$ is the {\em trace norm}):
\begin{definition}\label{Def:PreservabilityDestructionCost}
{\rm($R$-Preservability Destruction Cost)}
For a given channel $\mathcal{E}_{\rm S}\in\mathcal{O}_R$ and $0<\epsilon\le1$, we say a channel $\bar{\Lambda}_{\rm S'}\in\wt{\mathcal{O}}_R^N$ together with an ensemble of reversible unitary free operations $\{\mathcal{U}_i,\mathcal{V}_i,p_i\}_{i=1}^k$ (i.e.\,$\mathcal{U}_i,\mathcal{V}_i\in\mathcal{O}_R$ and also their inverses are in $\mathcal{O}_R$) form an {\em $\epsilon$-destruction process of $R$-preservability} for $\mathcal{E}_{\rm S}$ if for some $\Lambda_{\rm SS'}\in\mathcal{O}_R^N$ we have 
\begin{align}\label{Eq:Destruction}
\frac{1}{2}\norm{\sum_{i=1}^kp_i\mathcal{U}_i\circ(\mathcal{E}_{\rm S}\otimes\bar{\Lambda}_{\rm S'})\circ \mathcal{V}_i - \Lambda_{\rm SS'}}_\diamond\le\epsilon.
\end{align}
The {\em $\epsilon$-destruction cost for $R$-preservability} is defined by $C_R^\epsilon(\mathcal{E}_{\rm S})\coloneqq\log_2\min k$, where the minimization is taken over all $\epsilon$-destruction process of $R$-preservability for $\mathcal{E}_{\rm S}$.
\end{definition}
Definition~\ref{Def:PreservabilityDestructionCost} is slightly different from the Definition 9 in Ref.~\cite{LiuWinter2019}. 
In Ref.~\cite{LiuWinter2019}, $\mathcal{U}_i$ and $\mathcal{V}_i$ are asked to be free channels, which will correspond to $\mathcal{O}_R^N$ in our current study.
While this will always lead to zero $R$-preservability for the output channel, we relax this condition in this work.
Also, we require the ancillary channel $\bar{\Lambda}_{\rm S'}$ to be an absolutely $R$-annihilating channel.

To state the result, we also consider the smooth version of $P_{D_{\rm max}}$ and $\bar{P}_{D_{\rm max}}$~\cite{LiuWinter2019}:
\begin{align}
P_{D_{\rm max}}^\epsilon(\mathcal{E})&\coloneqq\inf_{\frac{1}{2}\norm{\mathcal{E}' - \mathcal{E}}_\diamond\le\epsilon}P_{D_{\rm max}}(\mathcal{E}');\\
\label{Eq:SmoothP_Dmaxbar}
\bar{P}_{D_{\rm max}}^\epsilon(\mathcal{E})&\coloneqq\inf_{\frac{1}{2}\norm{\mathcal{E}' - \mathcal{E}}_\diamond\le\epsilon}\bar{P}_{D_{\rm max}}(\mathcal{E}').
\end{align}
Now we state the following result when the given $R$-theory admits no activation of $R$-preservability.
We note that although we write it as a theorem, conceptually this result is a corollary of Theorem 10 in Ref.~\cite{LiuWinter2019}. 
We give the proof in Appendix~\ref{App:Proof_Coro:OperationalMeaning} for the completeness of this work.

\begin{theorem}\label{Coro:OperationalMeaning}
Given an $R$-theory satisfying the following three conditions:
\begin{enumerate}[label=(\roman*)]
\item $\mathcal{O}_R^N = \wt{\mathcal{O}}_R^N\neq\emptyset$.
\item $\mathcal{O}_R$ is closed in the diamond norm topology for all possible input/output dimensions.
\item In a multipartite case, the pair-wise permutation unitaries between two local systems are in $\mathcal{O}_R$.
\end{enumerate}
Then for a given $\mathcal{E}\in\mathcal{O}_R$ and for any $0<\eta\le\epsilon<1$, we have
\begin{align}\label{Eq:DstructionCostBounds}
\bar{P}_{D_{\rm max}}^{\sqrt{\epsilon(2-\epsilon)}}(\mE)\le C^\epsilon_R(\mathcal{E})\le\bar{P}_{D_{\rm max}}^{\epsilon - \eta}(\mE) + 2\log_2\frac{1}{\eta} - 1.
\end{align}
\end{theorem}

Theorem~\ref{Coro:OperationalMeaning} provides a clear operational meaning of $\bar{P}_{D_{\rm max}}(\mathcal{E})$: It shows how robust the $R$-preservability of the given free operation $\mathcal{E}$ is when it is randomized over reversible free unitary operations together with an ancillary absolutely R-annihilating channel.
This can also be interpreted as the erasure cost of $R$-preservability.
Note that we assume no activation property of $R$-preservability.
When the given $R$-preservability can be activated, the lower bounds in Theorem~\ref{Coro:OperationalMeaning} can still be proved, while it is so far unclear whether the upper bound can also be obtained.

\section{Applications to Thermodynamics}\label{Sec:All-Application}
After introducing the general framework, one may ask for specific examples to illustrate $R$-preservability.
Specially, a natural question is whether there is any application.
These issues will be addressed in the following two sections with the focus on thermodynamics.
We remark that detailed studies of coherence preservability have been reported in Ref.~\cite{Saxena2019} recently.

\subsection{Thermodynamic Implications of Athermality Preservability}\label{Sec:Application}
We will give two examples by considering the $R$-theory of athermality with Gibbs-preserving maps as the free operations.
It turns out that the $R$-preservability monotones of athermality (or simply athermality preservability) given by Eqs.~\eqref{Eq:P_D} and~\eqref{Eq:P_Dbar} with max-relative entropy $D_{\rm max}$ can be directly related to two recently reported results: For a given Gibbs-preserving channel $\mathcal{N}$, $P_{D_{\rm max}}(\mathcal{N})$ [Eq.~\eqref{Eq:P_D}] is operationally related to the smallest bath size needed to thermalize all outputs of $\mathcal{N}$~\cite{Sparaciari2019}, and $\bar{P}_{D_{\rm max}}(\mathcal{N})$ [Eq.~\eqref{Eq:P_Dbar}] is an upper bound of the single-shot classical capacity of $\mathcal{N}$ in a classical communication scenario subject to thermodynamic constraints~\cite{Takagi2019-3}.
These illustrate how $R$-preservability can be related to existing results and provide new physical messages.

To start with, we define the $R$-theory of athermality.
The term ``athermality'' means the status that a system is out of thermal equilibrium.
With a fixed system size, the only state without this resource is the unique thermal state.
Formally, consider a given system ${\rm S}$ with dimension $d$. 
Suppose the system Hamiltonian is $H_{\rm S}$ and a temperature $T$ is also given.
Then the corresponding {\em thermal state} reads 
\begin{align}\label{Eq:ThermalState}
\gamma = \frac{e^{-\beta H_{\rm S}}}{{\rm tr}(e^{-\beta H_{\rm S}})}, 
\end{align}
where $\beta = \frac{1}{k_BT}$ is the inverse temperature and $k_B$ is the Boltzmann constant.
In this $R$-theory, all free states take the form $\{\gamma^{\otimes k}\,|\,k\in\mathbb{N}\}$, where we only consider dimensions of the form $d^k$ with some positive integer $k$.
The free operations will be the {\em Gibbs-preserving maps}, which are channels $\mathcal{E}$ satisfying $\mathcal{E}(\gamma^{\otimes k}) = \gamma^{\otimes l}$ for some $k,l\in\mathbb{N}$ (note that $k,l$ are uniquely determined by the input/output dimensions).
Intuitively, these channels are those which cannot drive thermal equilibrium states out of equilibrium, thereby being unable to create athermality.
Equipped with Gibbs-preserving maps, the corresponding $R$-theory will satisfy properties~\ref{Def:Proper:Nonempty},~\ref{Def:Proper:FreeIdentity},~\ref{Def:ProperQR-Tensor}, and~\ref{Def:Proper:Tensor}.

\subsubsection{Athermality Preservability and Bath Size}\label{Sec:Bath-Size}
As the first example, we will demonstrate that athermality preservability of a Gibbs-preserving channel can be naturally linked to the bath size needed to thermalize a system.
To this end, we use the framework and a thermalization model introduced by Ref.~\cite{Sparaciari2019}.
We will briefly explain the ingredients relevant to this work, and we refer the readers to Ref.~\cite{Sparaciari2019} for further details.

We begin by specifying the system and the bath in the thermalization scenario.
Consider a system ${\rm S}$ with Hilbert space $\mathcal{H}_{\rm S}$ and a bath ${\rm B}$ with Hilbert space $\mathcal{H}_{\rm S}^{\otimes (n-1)}$ ($n\in\mathbb{N}$).
The bath is assumed to possess the temperature $T$ and the Hamiltonian $H_{\rm B} = \sum_{i=1}^{n-1}\id_{1}\otimes...\otimes\id_{i-1}\otimes H_{\rm S}\otimes\id_{i+1}\otimes...\otimes\id_{n-1}$, where $H_{\rm S}$ is the Hamiltonian of the given system ${\rm S}$.
Let $\gamma$ be the thermal state associated with $T$ and $H_{\rm S}$.
Then we assume the bath is initially in the state $\gamma^{\otimes (n-1)}$.
The central question is to study how large the bath needs to be in order to successfully thermalize the system ${\rm S}$ in a given state $\rho_{\rm S}$.

To this end, a global channel $\mathcal{E}_{\rm SB}:{\rm SB}\to{\rm SB}$ is said to {\em $\epsilon$-thermalize} the system state $\rho_{\rm S}$ if~\cite{Sparaciari2019}
\begin{align}\label{Eq:Thermalize-Def}
\norm{\mathcal{E}_{\rm SB}\left[\rho_{\rm S}\otimes\gamma^{\otimes (n-1)}\right] - \gamma^{\otimes n}}_1\le\epsilon.
\end{align}
Here the aim is to study thermalization of a fixed and given input state in the sense that the channel will globally thermalize the system ${\rm SB}$.
To model the system-bath interaction for thermalization, we consider the following master equation introduced by Ref.~\cite{Sparaciari2019}:
\begin{align}\label{Eq:Master-Equation}
\frac{\partial\rho_{\rm SB}(t)}{\partial t} = \sum_k\lambda_k\left[U^{(k)}_{\rm SB}\rho_{\rm SB}(t)U^{(k),\dagger}_{\rm SB} - \rho_{\rm SB}(t)\right],
\end{align}
where $\rho_{\rm SB}(t)$ is the state on the global system ${\rm SB}$ at time $t$, $U^{(k)}_{\rm SB}$ is an energy-preserving unitary on the global system satisfying $[U^{(k)}_{\rm SB}, H_{\rm S} + H_{\rm B}] = 0$, and $\lambda_k$ is the rate for the corresponding unitary channel to happen [one can see this by checking Eqs.~(A2) and~(A3) in the Appendix A of Ref.~\cite{Sparaciari2019}, which imply that each unitary operator $U^{(k)}_{\rm SB}$ occurs according to a Poisson distribution with mean value $\lambda_kt$].
Each unitary $U^{(k)}_{\rm SB}$ models an elastic collision between certain subsystems in ${\rm SB}$.
We refer the reader to Ref.~\cite{Sparaciari2019} for the detailed framework.

Now, let $\mathcal{C}_n$ be the set of all channels acting as ${\rm SB}\to{\rm SB}$ that can be generated by the model Eq.~\eqref{Eq:Master-Equation} with a bath of size $n-1$ and a realization time $t$.
Then consider the following quantity~\cite{Sparaciari2019}:
\begin{align}
n_\epsilon(\rho_{\rm S})\coloneqq\inf\{n\in\mathbb{N}\,|\,\exists\,\mathcal{E}_{\rm SB}\in\mathcal{C}_n\;{\rm s.t.\;Eq.~\eqref{Eq:Thermalize-Def}\;holds}\}.
\end{align}
This quantity can be interpreted as the smallest bath size needed to $\epsilon$-thermalize the given state $\rho_{\rm S}$ with the thermalization model Eq.~\eqref{Eq:Master-Equation}.
It turns out that this concept can be generalized to channels. 
Define 
\begin{align}\label{Eq:BathSizeChannel}
\mathcal{B}^\epsilon(\mathcal{N})\coloneqq\sup_\rho n_\epsilon[\mathcal{N}(\rho)] - 1,
\end{align}
which is the maximization over all the smallest bath sizes among all outputs of $\mathcal{N}$ (note that here we only consider channel $\mathcal{N}$ with the output space ${\rm S}$, which is the main system admitting the given thermal state $\gamma$).
Then this can be interpreted as the smallest bath size needed to $\epsilon$-thermalize {\em all} outputs of $\mathcal{N}$ under the given thermalization model.
One can therefore understand this quantity as the minimal bath size associated with the channel $\mathcal{N}$.

Now we are in the position to provide the following bounds, whose proof is given in Appendix~\ref{App:BathSizeProof}.
This can be regarded as a generalization of the main results of Ref.~\cite{Sparaciari2019} to channels:
\begin{theorem}\label{Coro:BathSize}
Given a Gibbs-preserving map $\mathcal{N}$ and $0\le\epsilon<1$, we have
\begin{align}
\mathcal{B}^\epsilon(\mathcal{N})\le \frac{1}{\epsilon^2}2^{P_{D_{\rm max}}(\mathcal{N})}.
\end{align}
Moreover, if we further assume that $\gamma$ is full-rank, $\mathcal{N}$ is coherence-annihilating, and the system Hamiltonian $H_{\rm S}$ satisfies the energy subspace condition (Definition~\ref{Def:ESC}), then we have
\begin{align}
2^{P_{D_{\rm max}}(\mathcal{N})}\le \mathcal{B}^{\epsilon}(\mathcal{N}) + \frac{2\sqrt{\epsilon}}{p_{\rm min}(\gamma)} + 1.
\end{align}
where $p_{\rm min}(\gamma)$ is the smallest eigenvalue of $\gamma$.
\end{theorem}
We remark that by saying ``coherence-annihilating'' we mean the channel can only output states diagonal in the given energy eigenbasis (i.e. no coherence can survive).
This requirement plus the energy subspace condition (Definition~\ref{Def:ESC}) are necessary for the proof of the lower bound derived in Ref.~\cite{Sparaciari2019} (specifically, it is required by Lemma 17 in the Appendix C of Ref.~\cite{Sparaciari2019}).
An open question in this research line is whether one can derive a similar lower bound without these constraints.

As expected, for a Gibbs-preserving channel $\mathcal{N}$, the quantity $\mathcal{B}^\epsilon(\mathcal{N})$ can be understood as a measure of the robustness of the channel $\mathcal{N}$ against thermalization.
From the upper bound, we learn that the weaker the channel's ability to preserve athermality is, the smaller a heat bath needs to be to thermalize every output of $\mathcal{N}$.
Theorem~\ref{Coro:BathSize} builds a link between the ability to preserve athermality and the resource needed to thermalize all the outputs of a given channel, and it also gives athermality preservability a different thermodynamic interpretation.

\subsubsection{Athermality Preservability and Classical Communication}\label{Sec:CC}
It turns out that the robustness-like monotone $\bar{P}_{D_{\rm max}}$ can be related to a classical communication scenario subject to certain thermodynamic constraints, which is the second example in this section.
To start with, consider the communication scenario in which we want to send classical information (in terms of classical bits, which can be written as an orthonormal basis $\{\ket{m}\}_{m=0}^{M-1}$) via a channel $\mathcal{N}$.
Again, we always assume the channel $\mathcal{N}$ has the output space ${\rm S}$, which is associated with a thermal state $\gamma$ and hence an $R$-theory of athermality.
Here are two constraints for this communication setup: 
\begin{itemize}
\item The whole process (including encoding and decoding) cannot create athermality.
\end{itemize}
In other words, this means the corresponding physical system used to implement the channel and transmit the classical information can only get closer and closer to thermal equilibrium.
\begin{itemize}
\item The dynamics of $\mathcal{N}$ has a time scale much longer than the thermalization time scale in the environment.
\end{itemize}
Theoretically, this motivates us to approximate the dynamics of the surrounding system ${\rm A}$ by the full thermalization $\Phi_{\gamma_{\rm A}}:(\cdot)\mapsto\gamma_{\rm A}$, where $\gamma_{\rm A}$ is the thermal state of the environment (that is, $\gamma_{\rm A} = \gamma^{\otimes k}$ for a positive integer $k$).
Together with the non-signalling constraints discussed in Ref.~\cite{Takagi2019-3} (whose framework is briefly introduced in Appendix~\ref{App:CC}), we have the following theoretical model for this communication scenario:
\begin{align}\label{Eq:Athermality-Constrained}
\mathcal{E}_d\circ(\mathcal{N}\otimes\Phi_{\gamma_{\rm A}})\circ\mathcal{E}_e,
\end{align}
where $\mathcal{E}_e,\mathcal{E}_d,\mathcal{N}$ are all Gibbs-preserving maps ($\mE_e$ and $\mE_d$ can be interpreted as the encoding and decoding maps, respectively), and in the ancillary system ${\rm A}$ there is a full thermalization channel $\Phi_{\gamma_{\rm A}}$.
We call this a {\em thermalized classical communication scenario}, which is identical to the communication scenario given in Ref.~\cite{Takagi2019-3} (see also Appendix~\ref{App:CC}) subject to the thermodynamic constraints given above.
The central goal is to understand how much classical information can be sent within a given error.

Formally, the classical information is indicated by an orthonormal basis $\{\ket{m}\}_{m=0}^{M-1}$, and we are interested in how many basis elements can be recovered in the end of the whole process.
To this end, we use the averaged error $\varepsilon(\mathcal{N},\mathcal{E}_e,\mathcal{E}_d,\gamma_{\rm A})$ of a given combination $(\mathcal{N},\mathcal{E}_e,\mathcal{E}_d,\gamma_{\rm A})$ to evaluate the faithfulness of the output:
\begin{align}\label{Eq:AveError}
\varepsilon(\mathcal{N},\mathcal{E}_e,\mathcal{E}_d,\gamma_{\rm A})\coloneqq 1 - \frac{1}{M}\sum_{m=0}^{M-1}\bra{m}\mathcal{E}_d\circ(\mathcal{N}\otimes\Phi_{\gamma_{\rm A}})\circ\mathcal{E}_e(\proj{m})\ket{m}.
\end{align}
We can now define the following single-shot classical capacity with an error $0<\epsilon<1$:
\begin{align}
C_{\gamma, (1)}^\epsilon(\mathcal{N})\coloneqq\sup_{\mathcal{E}_e,\mathcal{E}_d,{\rm A}}\{\log_2M\,|\,\varepsilon(\mathcal{N},\mathcal{E}_e,\mathcal{E}_d,\gamma_{\rm A})\le\epsilon\},
\end{align}
where the maximization is taken over all the possible Gibbs-preserving channels $\mathcal{E}_e,\mathcal{E}_d$ and ancillary systems ${\rm A}$.
This quantity tells us the optimal performance of the channel $\mathcal{N}$ in a thermalized classical communication scenario: It is the highest amount of classical information that can be transmitted within the given error $\epsilon$.
As a remark, in the realistic situation the size of the environment cannot be as large as we want.
This means in the practical setup one cannot optimize over all the possible ancillary systems, and the above quantity is an upper bound of the realistic capacity in general.

It turns out that this quantity is upper bounded by the athermality preservability monotone $\bar{P}_{D_{\rm max}}$.
Hence, when the given channel $\mathcal{N}$ has a weak ability to maintain athermality, then it can neither have a good performance in a thermalized classical communication scenario.
Formally, we have the following result, which is conceptually a corollary of Theorem 3 in Ref.~\cite{Takagi2019-3}.
The proof can be found in Appendix~\ref{App:ConverseBoundAthermality}.
\begin{theorem}\label{Result:ConverseBoundAthermality}
For a Gibbs-preserving map $\mathcal{N}$ and $0<\epsilon,\delta<1$, we have
\begin{align}
C^\epsilon_{\gamma,(1)}(\mathcal{N})\le \bar{P}^\delta_{D_{\rm max}}(\mathcal{N}) + \log_2\frac{1}{1-\epsilon-\delta}.
\end{align}
\end{theorem}
Theorem~\ref{Result:ConverseBoundAthermality} gives an alternative operational interpretation of Eq.~\eqref{Eq:P_Dbar} in the case of athermality: It is an upper bound of the optimal performance in a thermalized classical communication scenario.
Being consistent with the intuition, this result implies that if the given Gibbs-preserving channel has highly thermalized output states, then it can hardly keep the encoded classical messages through a scenario that cannot drive the system away from thermal equilibrium.

\subsection{Application to Entanglement Preserving Local Thermalization}\label{Sec:EPLT}
As another application, we apply the theory of $R$-preservability to the study of {\em entanglement preserving local thermalization} (EPLT)~\cite{Hsieh2019}, which is a topic aiming to understand the interplay between globally distributed quantum correlation and locally performed thermalizations. 
The central question of EPLT is: Can entanglement survive subsystem thermalizations?
To formulate the question, suppose an unknown input state is distributed to two local agents A and B.
We assume that the agents can neither use quantum resources (e.g. sharing a maximally entangled state), nor can they communicate with each other.
Both of them possess an individual local heat bath (and hence a given local thermal state $\gamma_{\rm X}$; ${\rm X = A,B}$), and we allow classical correlation between the local heat baths.
When both local agents let their local systems interact with the local heat baths and thermalize, the question is whether we can have global entanglement after thermalizations are achieved locally, at least for certain input states.

The above question can be formulated information-theoretically as follows.
Formally, a {\em local operations plus shared randomness} (LOSR; see Appendix~\ref{App:LOSR} for the definition) channel $\mE$ is called a {\em local thermalization} to a pair of single party thermal states $(\gamma_{\rm A},\gamma_{\rm B})$ if 
\begin{align}
{\rm tr}_{\rm A}\circ\mE(\cdot) = \gamma_{\rm B};\;{\rm tr}_{\rm B}\circ\mE(\cdot) = \gamma_{\rm A}.
\end{align} 
In other words, it is a full thermalization channel to $\gamma_{\rm X}$ in the local system ${\rm X}$ [which is different from the state-dependent definition given by Eq.~\eqref{Eq:Thermalize-Def}].
An EPLT is defined to be a local thermalization that can preserve entanglement for certain inputs; that is, it is a local thermalization with non-zero entanglement preservability.
A physical message from the existence of EPLT is when local agents couple to a global heat bath that only admits classical correlations within, it is still possible for global entanglement to survive after subsystem thermalizations: Classical correlations in the bath are enough to protect entanglement from being destroyed by locally performed thermalizations.

Recently, the existence of EPLT has been proved for all nonzero local temperatures and finite-energy local Hamiltonians in a bipartite setup with the help of shared randomness\footnote{The shared randomness is necessary for the existence of EPLT. This is because the full thermalization channel $\Phi_\gamma:(\cdot)\mapsto\gamma$ is an entanglement-breaking channel, and every product local thermalization to $(\gamma_{\rm A},\gamma_{\rm B})$ will take be form $(\cdot)\mapsto\gamma_{\rm A}\otimes\gamma_{\rm B}$~\cite{Hsieh2019}. We also remark that the EPLT constructed in Ref.~\cite{Hsieh2019} can be interpreted as a combination of twirling operation plus a generalized depolarizing channel in finite temperatures. See the beginning of Appendix~\ref{App:Proof_Result:LocalTher} for the detail.}~\cite{Hsieh2019}.
This means there is no temperature and energy thresholds for the existence of EPLT; that is, EPLT has no ``thermodynamic threshold'' to exist. 
From this a natural question is to ask whether there is any ``correlation threshold''; in other words, is it true that EPLT can exist only when its ability to preserve entanglement is strong enough?
With the formulation of $R$-preservability in hand, we can now answer this question.
As proved in Appendix~\ref{App:closetoB}, the result we found suggests that EPLT is a phenomenon generic for different values of the entanglement preservability [note that the $R$-theory of entanglement with LOSR channels as free operations will satisfy properties~\ref{Def:Proper:Nonempty},~\ref{Def:Proper:FreeIdentity},~\ref{Def:ProperQR-Tensor}, and~\ref{Def:Proper:Tensor}]. 

\begin{theorem}\label{Result:NoCorrelationThreshold}
For every full-rank $\gamma_{\rm A},\gamma_{\rm B}$ and every $\delta>0$, there exists an entanglement preserving local thermalization $\mE$ to $(\gamma_{\rm A},\gamma_{\rm B})$ such that
\begin{align}
\bar{P}_{\norm{\cdot}_1}(\mathcal{E})<\delta.
\end{align}
\end{theorem}
Hence, the existence of EPLT is generic both in thermodynamic measures (temperature and energy) and correlation measure (entanglement preservability).
Moreover, this physical message can actually be generalized to the preservability of free entanglement~\cite{Horodecki1998}.
Before stating the main result, we specify terminologies. 
In what follows, the {\em normalized temperature} of the given local system ${\rm X}$ is defined by $\tau_{\rm X}\coloneqq\frac{k_BT_{\rm X}}{\norm{H_{\rm X}}_\infty}$, where $T_{\rm X}$ is the local temperature, $k_B$ is the Boltzmann constant, and $\norm{H_{\rm X}}_\infty$ is the highest local energy (here the {\em sup norm} is defined by $\norm{\rho}_\infty\coloneqq\sup_{\ket{\psi}}|\bra{\psi}\rho\ket{\psi}|$). 
$\mathcal{O}_{\rm FE}^N$ is the set of all LOSR channels that cannot preserve free entanglement~\cite{Horodecki1998}.
Using a new family of EPLT constructed in Appendix~\ref{App:AlternativeEPLT} [Eq.~\eqref{Eq:Small-epsilon-EPLT}], we prove the following result in Appendix~\ref{App:Proof_Result:LocalTher} ($d$ is the common local dimensions of both subsystems):
\begin{theorem}\label{Result:LocalTher}
For every pair $(\gamma_{\rm A},\gamma_{\rm B})$ there exists a local thermalization $\mathcal{E}_+$ such that
\begin{align}\label{Thm:distance}
\inf_{\Lambda\in\mathcal{O}_{\rm FE}^N}\norm{\mE_+ - \Lambda}_\diamond\geq (3d-1) p_{\rm min} -2,
\end{align}
where $p_{\rm min}$ is the smallest eigenvalue among $\gamma_{\rm A}$ and $\gamma_{\rm B}$.

For every $\delta>0$, there exists a finite value $\tau_\delta>0$ such that for every pair $(\gamma_{\rm A},\gamma_{\rm B})$ with $\min_{\rm X}\tau_{\rm X} >\tau_\delta$, there exists an entanglement preserving local thermalization $\mathcal{E}_-$ to $(\gamma_{\rm A},\gamma_{\rm B})$ such that
\begin{align}\label{Result:LocalTherEq}
\bar{P}_{\norm{\cdot}_1}(\mathcal{E}_-)<\delta\quad\&\quad\mathcal{E}_-\notin\mathcal{O}_{\rm FE}^N.
\end{align}
That is, $\mathcal{E}_-$ can preserve free entanglement.
\end{theorem}	
Hence, for arbitrarily small entanglement preservability, there always exists a finite temperature EPLT that can also preserve free entanglement.
In other words, while they preserve arbitrarily little entanglement, many copies of some output can be distilled back to a maximally entangled state by LOCC channels. 
Hence, the conclusion that EPLT exists without a correlation threshold is the same even when we use the preservability of free entanglement as the measure.

We also remark that since 
\begin{align}
\bar{P}_{\norm{\cdot}_1}(\mathcal{E}_+)\coloneqq\inf_{\Lambda\in\mathcal{O}_{\rm E}^N}\norm{\mE_+ - \Lambda}_\diamond\ge\inf_{\Lambda\in\mathcal{O}_{\rm FE}^N}\norm{\mE_+ - \Lambda}_\diamond,
\end{align}
Eq.~\eqref{Thm:distance} automatically implies a lower bound of the entanglement preservability.
For high normalized temperatures, we have $p_{\rm min} \rightarrow \frac{1}{d}$ and the bound in Eq.~\eqref{Thm:distance} becomes arbitrarily close to $1-\frac{1}{d}$, as expected since $\inf_{\Lambda\in\mathcal{O}_{\rm E}^N}\norm{\mathcal{T} - \Lambda}_\diamond \geq 1-\frac{1}{d}$ [see Eq.~\eqref{Eq:App_T_lowerbound}] and $\mathcal{T}$ is an EPLT at infinite normalized temperature~\cite{Hsieh2019}.

\section{Conclusions}\label{Sec:Conclusion}
In a given resource theory of quantum states, we quantify the ability of free operations to preserve the resource.
To this end, we formulate this ability, termed {\em resource preservability}, as a channel resource induced by the given state resource.
Two classes of resource preservability monotones are proved: One is induced by state resource quantifiers, and another is based on channel distance measures.
The latter also induces a robustness-like measure with operational interpretation as the erasure cost of resource preservability~\cite{LiuWinter2019}.

To illustrate the connection between resource preservability and other research directions, we consider the resource theory of athermality.
We provide physical interpretations of two ahtermality preservability monotones induced by max-relative entropy. 
One has a thermodynamic interpretation directly related to the bath size needed for thermalization: The ability of a Gibbs-preserving channel to preserve athermality is physically connected to the minimal bath size needed to thermalize all its outputs.
Another monotone is shown to bound the capacity of a classical communication scenario under certain thermodynamic constraints.
By adding thermodynamic conditions to a general classical communication setup, the ability for the given Gibbs-preserving channel to preserve athermality tells us the highest possible amount of transmissible classical messages.

As another application, we study the entanglement preservability of {\em entanglement preserving local thermalizations} (EPLTs)~\cite{Hsieh2019}, which is a family of local operation plus shared randomness channels that locally behave as thermalization for arbitrary inputs, while globally have the ability to preserve certain amounts of entanglement.
In this work, we show that EPLT can exist with arbitrarily small entanglement preservability for every positive temperatures.
Hence, EPLT's existence is independent of both temperature constraints and the ability to preserve entanglement.
We further provide a new family of EPLTs that has the ability to preserve free entanglement, even though its entanglement preservability can be arbitrarily small at finite temperatures.
This suggests the existence of EPLT is generic in various values of free entanglement preservability.

Several open questions remain.
From the operational perspective, it will be interesting to know whether there is any operational interpretation of $R$-preservability monotones induced by state resource monotones introduced in Sec.~\ref{Sec:PQ_R}.
Also, the robustness-like measure introduced in Sec.~\ref{Sec:Robustness} is shown to have an operational interpretation~\cite{LiuWinter2019} when the given $R$-preservability theory has no activation property, while it is unknown whether this operational interpretation can still hold when the given $R$-preservability allows activation.
Regarding the structure of channel resource theory, it is so far unknown how to characterize the largest set of free super-channels of $R$-preservability (since it may not always be the set of all super-channels that cannot generate $R$-preservability as discussed in footnote~\ref{App:Well-defined}).
Finally, it is also an open question whether one can drop the temperature dependency of entanglement preserving local thermalizations in Theorem~\ref{Result:LocalTher}.
We hope this work can initiate the interest in the study of resource preservation properties in different state resource theories.

{\em Note added.} Recently, we became aware of the related work Ref.~\cite{Saxena2019} which consider the preservation of coherence as a channel resource.

\section*{Acknowledgements}
We thank (in alphabetical order) Antonio Ac\'in, Stefan B$\ddot{\rm a}$uml, Daniel Cavalcanti, Marcus Huber, Yeong-Cherng Liang, Matteo Lostaglio, Ryuji Takagi, Marco T\'ulio Quintino, Gabriel Senno,  Andreas Winter for fruitful discussions and comments.
This project is part of the ICFOstepstone - PhD Programme for Early-Stage Researchers in Photonics, funded by the Marie Sk\l odowska-Curie Co-funding of regional, national and international programmes (GA665884) of the European Commission, as well as by the ‘Severo Ochoa 2016-2019' program at ICFO (SEV-2015-0522), funded by the Spanish Ministry of Economy, Industry, and Competitiveness (MINECO).
We also acknowledge support from the Spanish MINECO (Severo Ochoa SEV-2015-0522), Fundaci\'o Cellex and Mir-Puig, Generalitat de Catalunya (SGR1381 and CERCA Programme).

\newpage\appendix

\section{Remark on the Activation Property of Resource Preservability}\label{Sec:Activation}
In this section, we provide an example of activation property of $R$-preservability.
Consider the $R$-theory of nonlocality~\cite{Bell,Bell-RMP} (and we write $R = {\rm NL}$) on a bipartite system ${\rm SS'}$ with equal finite local dimension $D$, and {\em local operations plus shared randomness} (LOSR) channels as the free operations~\cite{Rosset2018,Wolfe2019} (in Appendix~\ref{App:LOSR} we briefly explain the reason).
First, we recall a phenomenon called {\em superactivation}, which is proved for nonlocality~\cite{Palazuelos2012} and generalized to quantum steering~\cite{Hsieh2016,Quintino2016} (and we also mention other activation properties of nonlocality in Refs.~\cite{Masanes2008,Liang2012}).
Formally, a local state $\rho$ (with local dimension $D=d$) is said to admit superactivation of nonlocality if there exists a finite $k\in\mathbb{N}$ such that $\rho^{\otimes k}$ is nonlocal (in the bipartition ${\rm SS'}$ and local dimension $D=d^k$).
We refer the readers to Appendix~\ref{App:LOSR} for the definition of local/nonlocal states.
In ${\rm SS'}$ with $D=d$, it is shown that a state can demonstrate superactivation of nonlocality if its {\em fully entangled fraction} (FEF) is higher than $\frac{1}{d}$~\cite{Cavalcanti2013}, where for the given bipartite system the FEF is defined by~\cite{Horodecki1999-2,Albeverio2002}:
\begin{align}\label{Eq:FEF}
\mathcal{F}(\rho_{\rm SS'})\coloneqq\sup_{\ket{\Phi_d}}\bra{\Phi_d}\rho_{\rm SS'}\ket{\Phi_d}.
\end{align}
The maximization is taken over all maximally entangled states $\ket{\Phi_d}$ on the given bipartite system ${\rm SS'}$.
FEF is well-known for its capacity to characterize various quantum properties~\cite{Horodecki1999-2,Albeverio2002,Zhao2010,Cavalcanti2013,Bell-RMP,Quintino2016,Hsieh2016,Hsieh2018E,Ent-RMP,Liang2019}.

To construct the example, we make use of the {\em $(U\otimes U^*)$-twirling operation} on ${\rm SS'}$ defined by~\cite{Horodecki1999,Bennett1996}
\begin{align}\label{Eq:Twirling}
\mathcal{T}(\cdot)\coloneqq\int_{U(d)}(U\otimes U^*)(\cdot)(U\otimes U^*)^\dagger dU,
\end{align}
where the integration is taken over the group of $d\times d$ unitary operators $U(d)$ with the Haar measure $dU$.
The twirling operation $\mathcal{T}$ is by definition an LOSR channel, thereby being a free operation.
It has the property to preserve entanglement:
\begin{align}
\bra{\Psi_d^+}\mathcal{T}(\rho_{\rm SS'})\ket{\Psi_d^+} = \bra{\Psi_d^+}\rho_{\rm SS'}\ket{\Psi_d^+}.
\end{align}
Also, the output of $\mathcal{T}$ will always be an {\em isotropic state}~\cite{Horodecki1999}:
\begin{align}\label{Eq:rIso}
\rIso(p)\coloneqq p\proj{\Psi_d^+} + (1-p)\frac{\id_{\rm SS'}}{d^2},
\end{align}
where $\ket{\Psi_d^+}\coloneqq\frac{1}{\sqrt{d}}\sum_{n=0}^{d-1}\ket{n}\otimes\ket{n}$ is a maximally entangled state, and $p\in\left[-\frac{1}{d^2-1},1\right]$ due to the positivity of quantum states.
Now we consider the following channel:
\begin{align}
\wt{\mathcal{T}}(\cdot)\coloneqq \wt{p}\mathcal{T}(\cdot) + (1-\wt{p})\frac{\id_{\rm SS'}}{d^2},
\end{align}
and we choose $\wt{p}$ such that the output state cannot have FEF larger than the threshold for nonlocality of isotropic states~\cite{Bell-RMP}, while can still have FEF larger than $\frac{1}{d}$ for certain entangled inputs.
More precisely, we choose~\cite{Bell-RMP,Almeida2007}
\begin{align}
\frac{1}{d+1}<\wt{p}<\frac{(d-1)^{(d-1)}(3d-1)}{(d+1)d^d},
\end{align}
which will guarantee the above claim.	
Being an LOSR channel, this means $\widetilde{\mathcal{T}}\in\mathcal{O}_{\rm NL}^N$.
Also, when the input state is $\ket{\Psi_d^+}$, $\wt{\mathcal{T}}(\proj{\Psi_d^+})$ will be an entangled isotropic state, thereby having FEF $>\frac{1}{d}$ and hence admitting superactivation of nonlocality.
Hence, when one consider $\widetilde{\mathcal{T}}^{\otimes k}$ with a large enough $k$, it is possible to output nonlocal states (on the given bipartition ${\rm SS'}$ with local dimension $D = d^k$), which means $\widetilde{\mathcal{T}}^{\otimes k}\notin\mathcal{O}_{\rm NL}^N$.
This illustrates the existence of superactivation property of nonlocality preservability, which also teaches us that for a general formulation, the assumption {\em $\Lambda_{\rm S}\otimes\Lambda_{\rm S'}\in\mathcal{O}_R^N$ if $\Lambda_{\rm S},\Lambda_{\rm S'}\in\mathcal{O}_R^N$} cannot be imposed.

As a remark, we note that there do exist examples without activation property.
For instance, if we use Gibbs-preserving map as the free operation in the $R$-theory of athermality, then the only $R$-annihilating channel is the state preparation channel of the given thermal state [Eq.~\eqref{Eq:ThermalState}].
Because product local thermalization cannot preserve any correlation~\cite{Hsieh2019}, we learn that it is impossible to activate resource preservability in this case.

\subsection{Local Operations Plus Shared Randomness Channels}\label{App:LOSR}
In this section, we briefly explain why LOSR channels can be free operations of nonlocality.
It suffices to consider a bipartite system ${\rm AB}$.
Formally, an LOSR channel is defined to take the following form:
\begin{align}
\mE = \int (\mE_\lambda^{\rm A}\otimes\mE_\lambda^{\rm B})p_\lambda d\lambda,
\end{align}
where the integration is taken over the variable $\lambda$ and $\mE_\lambda^{\rm A},\mE_\lambda^{\rm B}$ are local channels.
In what follows we will write $\{E_{a|x}\}$ as a set of local {\em positive operator-valued measures} (POVMs)~\cite{QCI-book}; that is, for each input value $x$, $E_{a|x}$'s form an POVM: $E_{a|x}\ge0\;\forall\,a$ and $\sum_aE_{a|x} = \id_{\rm A}\;\forall\,x$.
We use the notation $\{E_{b|y}\}$ for the POVMs in the subsystem ${\rm B}$.

With the above setting, a quantum state $\rho_{\rm AB}$ is said to be {\em local} if for every local sets of POVMs $\{E_{a|x}\},\{E_{b|y}\}$ one can write~\cite{Bell, Bell-RMP}
\begin{align}\label{Eq:LHV}
{\rm tr}\left[(E_{a|x}\otimes E_{b|y})\rho_{\rm AB}\right] = \int_{\lambda\in\Lambda_{\rm LHV}}P(a|x,\lambda)P(b|y,\lambda)p_\lambda d\lambda
\end{align}
for some variable $\lambda$ in a set $\Lambda_{\rm LHV}$ and some probability distributions $P(a|x,\lambda),P(b|y,\lambda),p_\lambda$.
In other words, a state is local if all the possible combinations of local POVMs cannot distinguish it with a {\em local hidden-variable model}, as depicted by $\Lambda_{\rm LHV}$.
Any state that is not local is said to be {\em nonlocal}.

Now we explain that LOSR channel will map local states to local states.
To see this, we note that for a given LOSR channel $\mE$, we have
\begin{align}
{\rm tr}\left[(E_{a|x}\otimes E_{b|y})\mE(\rho_{\rm AB})\right] &= \int {\rm tr}\left[(E_{a|x}\otimes E_{b|y})(\mE_\lambda^{\rm A}\otimes\mE_\lambda^{\rm A})(\rho_{\rm AB})\right]p_\lambda d\lambda\nonumber\\
&=\int {\rm tr}\left\{\left[\mE_\lambda^{{\rm A},\dagger}(E_{a|x})\otimes \mE_\lambda^{{\rm B},\dagger}(E_{b|y})\right](\rho_{\rm AB})\right\}p_\lambda d\lambda,
\end{align}
where for ${\rm X = A,B}$, $\mE_\lambda^{{\rm X},\dagger}$'s are completely-positive unital map since $\mE_\lambda^{\rm X}$'s are completely-positive trace-preserving map.
This means $\mE_\lambda^{{\rm A},\dagger}(E_{a|x})$ and $\mE_\lambda^{{\rm B},\dagger}(E_{b|y})$ again form local sets of POVMs.
Since $\rho_{\rm AB}$ is local, the quantity ${\rm tr}\left\{\left[\mE_\lambda^{{\rm A},\dagger}(E_{a|x})\otimes \mE_\lambda^{{\rm B},\dagger}(E_{b|y})\right](\rho_{\rm AB})\right\}$ must take the form of Eq.~\eqref{Eq:LHV}.
This shows that LOSR channels map local states to local states, and hence form a suitable candidate of free operations for nonlocality.

\section{Example of Absolutely Resource Annihilating Channels}\label{App:EntExample}
Using the $R$-theory of entanglement, we will show that every channel that is entanglement-annihilating and entanglement-breaking will be an absolutely $R$-annihilating channel (and we also say it is absolutely entanglement-annihilating).
\begin{afact}
If a bipartite channel $\mathcal{E}$ is entanglement-annihilating and entanglement-breaking, then it is absolutely entanglement-annihilating.
\end{afact}
\begin{proof}
We rewrite this channel as $\mathcal{E}_{\rm A_1B_1}$, where the subscript means that it is in the bipartite system ${\rm A_1B_1}$.
Then it suffices to show that there will be no entanglement in the ${\rm AB}$ bipartition after the product channel $\mathcal{E}_{\rm A_1B_1}\otimes\Lambda_{\rm A_2B_2}$, where $\Lambda_{\rm A_2B_2}$ is an entanglement-annihilating channel in the bipartite system ${\rm A_2B_2}$ and it annihilates entanglement in the ${\rm AB}$ bipartition.

Because $\mathcal{E}_{\rm A_1B_1}$ is entanglement-breaking, this means $\mathcal{E}_{\rm A_1B_1}\otimes\mathcal{I}_{\rm A_2B_2}$ is entanglement-annihilating in the $12$ bipartition.
In other words, there will be no entanglement in the $12$ bipartition after $\mathcal{E}_{\rm A_1B_1}\otimes\Lambda_{\rm A_2B_2}$.
Hence, the remaining possibility for the preserved entanglement are in the bipartite systems ${\rm A_1B_1}$ and ${\rm A_2B_2}$.
But since both $\mathcal{E}_{\rm A_1B_1}$ and $\Lambda_{\rm A_2B_2}$ are entanglement-annihilating in the ${\rm AB}$ bipartition, we conclude that no entanglement exists in the bipartite systems ${\rm A_1B_1}$ and ${\rm A_2B_2}$.
This shows that the output states cannot be entangled in the ${\rm AB}$ bipartition.
\end{proof}

\section{Proof of Theorem~\ref{Result:Maintain}}\label{App:Proof_Result:Maintain}
\begin{proof}
To show property~\ref{P:Positivity}, note that for a given $\Lambda_{\rm S}\in\mathcal{O}_R^N$ we have $\Lambda_{\rm S}\otimes\widetilde{\Lambda}_{\rm A}\in\mathcal{O}_R^N$ for all $\widetilde{\Lambda}_{\rm A}\in\widetilde{\mathcal{O}}_R^N$.
This means $(f\circ Q_R)[(\Lambda_{\rm S}\otimes\widetilde{\Lambda}_{\rm A})(\rho_{\rm SA})] = 0$ for all $\widetilde{\Lambda}_{\rm A}\in\widetilde{\mathcal{O}}_R^N$ and for all $\rho_{\rm SA}$.
Hence, property~\ref{P:Positivity} is proved.

To show property~\ref{P:Monotonicity}, we recall from Definition~\ref{Def:Free-Super-Channel} that for a given free super-channel $F_\mathcal{E}$ acting on free operations $\mathcal{E}\in\mathcal{O}_R$, there exist an ancillary system ${\rm B}$, two free operations $\Lambda_+,\Lambda_-\in\mathcal{O}_R$, and an absolutely $R$-annihilating channel $\widetilde{\Lambda}_{\rm B}\in\widetilde{\mathcal{O}}_R^N$ such that $F_\mathcal{E} = \Lambda_+\circ(\mathcal{E}\otimes\widetilde{\Lambda}_{\rm B})\circ\Lambda_-$.
In what follows, because the input/output dimensions of $\Lambda_-$ do not need to be the same, we write ${\rm S'}$ as the input space and ${\rm SB}$ as the output space of $\Lambda_-$; namely, we have $\Lambda_-:{\rm S'}\to{\rm SB}$.
Then we have [note that the maximization is taken over $\rho_{\rm SA}$ satisfying $Q_R(\rho_{\rm SA})>0$ according to our definition; see explanations below Eq.~\eqref{Eq:supA}]
\begin{align}
P_{Q_R}^{(f,g)}(F_\mathcal{E})&= \supA \frac{(f\circ Q_R)\left\{[(\Lambda_+\otimes\mathcal{I}_{\rm A})\circ(\mathcal{E}\otimes\widetilde{\Lambda}_{\rm B}\otimes\widetilde{\Lambda}_{\rm A})\circ(\Lambda_-\otimes\mathcal{I}_{\rm A})](\rho_{\rm S'A})\right\}}{(g\circ Q_R)(\rho_{\rm S'A})}\nonumber\\
&\le \supA \frac{(f\circ Q_R)\left\{[(\mathcal{E}\otimes\widetilde{\Lambda}_{\rm B}\otimes\widetilde{\Lambda}_{\rm A})\circ(\Lambda_-\otimes\mathcal{I}_{\rm A})](\rho_{\rm S'A})\right\}}{(g\circ Q_R)(\rho_{\rm S'A})}\nonumber\\
&\le \supA \frac{(f\circ Q_R)\left\{[(\mathcal{E}\otimes\widetilde{\Lambda}_{\rm B}\otimes\widetilde{\Lambda}_{\rm A})\circ(\Lambda_-\otimes\mathcal{I}_{\rm A})](\rho_{\rm S'A})\right\}}{(g\circ Q_R)[(\Lambda_-\otimes\mathcal{I}_{\rm A})(\rho_{\rm S'A})]}\nonumber\\
&\le \supA \frac{(f\circ Q_R)\left[(\mathcal{E}\otimes\widetilde{\Lambda}_{\rm B}\otimes\widetilde{\Lambda}_{\rm A})(\rho_{\rm SBA})\right]}{(g\circ Q_R)(\rho_{\rm SBA})}\nonumber\\
&\le \supA \frac{(f\circ Q_R)\left[(\mathcal{E}\otimes\widetilde{\Lambda}_{\rm A})(\rho_{\rm SA})\right]}{(g\circ Q_R)(\rho_{\rm SA})}\nonumber\\
&=P_{Q_R}^{(f,g)}(\mathcal{E}).
\end{align}
The second line is because $Q_R$ is non-increasing under free operation $(\Lambda_+\otimes\mathcal{I}_{\rm A})$, which is due to the properties~\ref{Def:Proper:FreeIdentity}, ~\ref{Def:Proper:Tensor},~\ref{P:Monotonicity}, and the fact that $f$ is strictly increasing.
The same reasons imply the third line (while with some subtleties explained below).
The fourth line is because maximizing over all states of the form $(\Lambda_-\otimes\mathcal{I}_{\rm A})(\rho_{\rm SA})$ is sub-optimal than the range of all states on the system ${\rm SBA}$.
The fifth line is because $\widetilde{\Lambda}_{\rm B}\otimes\widetilde{\Lambda}_{\rm A}$ gives a range that is sub-optimal than all the possible $\widetilde{\Lambda}_{\rm A}$ when one maximizes over all the ancillary systems ${\rm A}$ [recall from Eq.~\eqref{Fact:wtL} that the set of absolutely $R$-annihilating channels for an $R$-theory satisfying properties~\ref{Def:Proper:Nonempty},~\ref{Def:Proper:FreeIdentity},~\ref{Def:ProperQR-Tensor}, and~\ref{Def:Proper:Tensor} is closed under tensor product].

Here we note that the ranges of optimization in the second line and the third line are different.
In the second line, the optimization is taken over $\rho_{\rm S'A}$ with $Q_R(\rho_{\rm S'A})>0$, which implies two different cases.
The first case is when the optimization over this range is zero [$\overline{\sup}_{\rm A}(...)=0$ in the second line].
Then in this case the desired inequality holds.
This means we can assume the second case without loss of generality; that is, we can assume the optimization in the second line over $Q_R(\rho_{\rm S'A})>0$ gives nonzero value.
Hence, the range for the second line can be rewritten as $\rho_{\rm S'A}$ with $Q_R(\rho_{\rm S'A})>0$ and $Q_R[(\Lambda_-\otimes\mathcal{I}_{\rm A})(\rho_{\rm S'A})]>0$, since the latter inequality is necessary for a nonzero numerator (note that actually the latter inequality implies the former one, while we still write them both explicitly for understanding).
Then one can proceed to the third line with this condition.
This proves property~\ref{P:Monotonicity}.

To prove the property given by Eq.~\eqref{P:Tensor}, we first note the following: (the maximization is again taken over states with non-zero $Q_R$ values)
\begin{align}\label{Eq:assumptionEQ}
P_{Q_R}^{(f,g)}(\mE_{\rm S}\otimes\mE_{\rm S'}) &= \supA\frac{(f\circ Q_R)\left[(\mE_{\rm S}\otimes\mE_{\rm S'}\otimes\wtL_{\rm A})(\rho_{\rm SS'A})\right]}{(g\circ Q_R)(\rho_{\rm SS'A})}\nonumber\\
&\ge\supA\frac{(f\circ Q_R)\left[(\mE_{\rm S}\otimes\mE_{\rm S'}\otimes\wtL_{\rm A})(\rho_{\rm SA}\otimes\wt{\eta}_{\rm S'})\right]}{(g\circ Q_R)(\rho_{\rm SA}\otimes\wt{\eta}_{\rm S'})}\nonumber\\
&\ge\supA\frac{(f\circ Q_R)\left[(\mE_{\rm S}\otimes\wtL_{\rm A})(\rho_{\rm SA})\right]}{(g\circ Q_R)(\rho_{\rm SA}\otimes\wt{\eta}_{\rm S'})}\nonumber\\
&\ge\supA\frac{(f\circ Q_R)\left[(\mE_{\rm S}\otimes\wtL_{\rm A})(\rho_{\rm SA})\right]}{(g\circ Q_R)(\rho_{\rm SA})}\nonumber\\
&=P_{Q_R}^{(f,g)}(\mE_{\rm S}).
\end{align}
Note that $\overline{\sup}_{\rm A}$ in the first line is maximizing over the system ${\rm SS'A}$.
The second line is because fixing an absolutely free state $\wt{\eta}_{\rm S'}\in\wt{\mathcal{F}}_R$ [here we use the assumption $\wt{\mathcal{F}}_R\neq\emptyset$ in property~\ref{Def:Proper:Nonempty}] will make the maximization sub-optimal than the original one, and we note that since this line $\overline{\sup}_{\rm A}$ is maximizing over ${\rm SA}$ [with $Q_R(\rho_{\rm SA})>0$].
The third line is because $f$ is strictly increasing and $Q_R$ is a resource monotone [property~\ref{Def:Proper:FreeIdentity}].
The fourth line is because $g$ is non-decreasing and $Q_R$ is a resource monotone [property~\ref{Def:ProperQR-Tensor}].
This proves the inequality in Eq.~\eqref{P:Tensor} for general $\mE_{\rm S}$ and $\mE_{\rm S'}$.

In the case that $\mE_{\rm S'} = \wtL_{\rm S'}\in\wt{\mathcal{O}}_R^N$, we have
\begin{align}
P_{Q_R}^{(f,g)}(\mE_{\rm S}\otimes\wtL_{\rm S'}) &= \supA\frac{(f\circ Q_R)\left[(\mE_{\rm S}\otimes\wtL_{\rm S'}\otimes\wtL_{\rm A})(\rho_{\rm SS'A})\right]}{(g\circ Q_R)(\rho_{\rm SS'A})}\nonumber\\
&\le\supA\frac{(f\circ Q_R)\left[(\mE_{\rm S}\otimes\wtL_{\rm A})(\rho_{\rm SA})\right]}{(g\circ Q_R)(\rho_{\rm SA})}\nonumber\\
&=P_{Q_R}^{(f,g)}(\mE_{\rm S}),
\end{align}
where the second line is because the range $\wtL_{\rm S'}\otimes\wtL_{\rm A}$ with the fixed $\wtL_{\rm S'}$ is sub-optimal than all the possible absolutely $R$-annihilating channel $\wtL_{\rm A}$ when one maximizes over all the ancillary systems ${\rm A}$ [recall again from Eq.~\eqref{Fact:wtL} that the set of absolutely $R$-annihilating channels will be closed under tensor product in the current case].
This shows the equality in Eq.~\eqref{P:Tensor}.

Finally, when $f\circ Q_R$ is convex, $P_{Q_R}^{(f,g)}$ is by definition convex.
This proves property~\ref{P:Convexity}.
To address property~\ref{P:Faithful}, for a given $\mE_{\rm S}\in\mathcal{O}_R$ we note that $P_{Q_R}^{(f,g)}(\mathcal{E}_{\rm S})=0$ implies $Q_R\left[(\mathcal{E}_{\rm S}\otimes\wtL_{\rm A})(\rho_{\rm SA})\right] = 0$ for all $\rho_{\rm SA}$, all $\wtL_{\rm A}\in\wt{\mathcal{O}}_R^N$, and all ancillary systems ${\rm A}$.
By considering the ancillary system as the trivial one (i.e.\,with zero dimension), we have $Q_R\left[\mathcal{E}_{\rm S}(\rho_{\rm S})\right] = 0$ for all $\rho_{\rm S}$.
when $Q_R$ is faithful, this means $\mathcal{E}_{\rm S}(\rho_{\rm S})\in\mathcal{F}_R$ for all $\rho_{\rm S}$, thereby implying $\mathcal{E}_{\rm S}\in\mathcal{O}_R^N$.
This shows property~\ref{P:Faithful} and also completes the whole proof.
\end{proof}
We remark that the assumption $\wt{\mathcal{F}}_R\neq\emptyset$ is only used in the proof of Eq.~\eqref{Eq:assumptionEQ}.
In other words, this assumption can be dropped if $g$ maps every input to a positive constant.
Write $g_c(\cdot) = c$, this means the following corollary for an $R$-theory satisfying the rest of properties~\ref{Def:Proper:Nonempty},~\ref{Def:Proper:FreeIdentity},~\ref{Def:ProperQR-Tensor}, and~\ref{Def:Proper:Tensor}:
\begin{acorollary}\label{Coro:g-constant}
Given an $R$-theory and a state resource monotone $Q_R$.
$f$ is a finite-valued strictly increasing function with $f(0) = 0$ and $c>0$ is a positive constant.
Then $P_{Q_R}^{(f,g_c)}$is an \mbox{$R$-preservability} monotone. 
Moreover, It is faithful if $Q_R$ is faithful, and it is convex if $f\circ Q_R$ is convex.
\end{acorollary}

\newpage
\section{Proof of Theorem~\ref{Result:DistanceMonotone}}\label{App:Proof_Result:DistanceMonotone}
\begin{proof}
Property~\ref{P:Positivity} holds automatically according to the definition.
To prove property~\ref{P:Monotonicity}, for a given free super-channel $F_{\mE_{\rm S}} = \Lambda_+\circ(\mE_{\rm S}\otimes\wtL_{\rm B})\circ\Lambda_-$ with $\Lambda_+,\Lambda_-\in\mathcal{O}_R$ and $\wtL_{\rm B}\in\wt{\mathcal{O}}_R^N$, the direct computation shows (we again adapt the notation $\Lambda_-:{\rm S'\to SB}$)
\begin{align}
&P_D(F_{\mE_{\rm S}})\nonumber\\
&= \inf_{\Lambda_{\rm S'}\in\mathcal{O}_R^N}\supA D\left[(F_{\mE_{\rm S}}\otimes\wtL_{\rm A})(\rho_{\rm S'A}),(\Lambda_{\rm S'}\otimes\wtL_{\rm A})(\rho_{\rm S'A})\right]\nonumber\\
&\le \inf_{\Lambda_{\rm SB}\in\mathcal{O}_R^N}\supA D\left\{(F_{\mE_{\rm S}}\otimes\wtL_{\rm A})(\rho_{\rm S'A}),\left[(\Lambda_+\circ\Lambda_{\rm SB}\circ\Lambda_-)\otimes\wtL_{\rm A}\right](\rho_{\rm S'A})\right\}\nonumber\\
&\le \inf_{\Lambda_{\rm SB}\in\mathcal{O}_R^N}\supA D\left\{\left[(\mE_{\rm S}\otimes\wtL_{\rm B}\otimes\wtL_{\rm A})\circ(\Lambda_-\otimes\mathcal{I}_{\rm A})\right](\rho_{\rm S'A}),\left[(\Lambda_{\rm SB}\otimes\wtL_{\rm A})\circ(\Lambda_-\otimes\mathcal{I}_{\rm A})\right](\rho_{\rm S'A})\right\}\nonumber\\
&\le \inf_{\Lambda_{\rm SB}\in\mathcal{O}_R^N}\supA D\left[(\mE_{\rm S}\otimes\wtL_{\rm B}\otimes\wtL_{\rm A})(\rho_{\rm SBA}),(\Lambda_{\rm SB}\otimes\wtL_{\rm A})(\rho_{\rm SBA})\right]\nonumber\\
&\le \inf_{\Lambda_{\rm S}\in\mathcal{O}_R^N}\supA D\left[(\mE_{\rm S}\otimes\wtL_{\rm B}\otimes\wtL_{\rm A})(\rho_{\rm SBA}),(\Lambda_{\rm S}\otimes\wtL_{\rm B}\otimes\wtL_{\rm A})(\rho_{\rm SBA})\right]\nonumber\\
&\le \inf_{\Lambda_{\rm S}\in\mathcal{O}_R^N}\supA D\left[(\mE_{\rm S}\otimes\wtL_{\rm A})(\rho_{\rm SA}),(\Lambda_{\rm S}\otimes\wtL_{\rm A})(\rho_{\rm SA})\right]\nonumber\\
&=P_D(\mE_{\rm S}).
\end{align}
The second line is because $\Lambda_+\circ\Lambda_{\rm SB}\circ\Lambda_-\in\mathcal{O}_R^N$ [which is true because of the assumptions that we made for $R$-theories in this work] forms a sub-optimal range compared with $\Lambda_{\rm S'}\in\mathcal{O}_R^N$.
The third line is because of the properties~\ref{Def:Proper:FreeIdentity} and~\ref{Def:Proper:Tensor}, plus the fact that $D$ satisfies Eq.~\eqref{Def:ProperDistanceMeasure}.
The fourth line is because $(\Lambda_-\otimes\mathcal{I}_{\rm A})(\rho_{\rm S'A})$ forms a sub-optimal range for the maximization $\overline{\sup}_{\rm A}$.
The fifth line is because $\Lambda_{\rm S}\otimes\wtL_{\rm B}\in\mathcal{O}_R^N$ (this is true due to the definition of the absolutely $R$-annihilating channels) with the fixed map $\wtL_{\rm B}\in\wt{\mathcal{O}}_R^N$ and the variable $\Lambda_{\rm S}$ forms a sub-optimal range for the minimization $\inf_{\Lambda_{\rm SB}\in\mathcal{O}_R^N}$.
The sixth line is because $\wtL_{\rm B}\otimes\wtL_{\rm A}$ forms a sub-optimal range for the maximization $\overline{\sup}_{\rm A}$ [recall Eq.~\eqref{Fact:wtL}].
This proves property~\ref{P:Monotonicity}.

To prove Eq.~\eqref{P:Tensor}, we first compute the following
\begin{align}
P_D(\mE_{\rm S}\otimes\mE_{\rm S'}) &= \inf_{\Lambda_{\rm SS'}\in\mathcal{O}_R^N}\supA D\left[(\mE_{\rm S}\otimes\mE_{\rm S'}\otimes\wtL_{\rm A})(\rho_{\rm SS'A}),(\Lambda_{\rm SS'}\otimes\wtL_{\rm A})(\rho_{\rm SS'A})\right]\nonumber\\
&\ge\inf_{\Lambda_{\rm SS'}\in\mathcal{O}_R^N}\supA D\left[(\mE_{\rm S}\otimes\mE_{\rm S'}\otimes\wtL_{\rm A})(\rho_{\rm SA}\otimes\wt{\eta}_{\rm S'}),(\Lambda_{\rm SS'}\otimes\wtL_{\rm A})(\rho_{\rm SA}\otimes\wt{\eta}_{\rm S'})\right]\nonumber\\
&\ge\inf_{\Lambda_{\rm SS'}\in\mathcal{O}_R^N}\supA D\left\{(\mE_{\rm S}\otimes\wtL_{\rm A})(\rho_{\rm SA}),{\rm tr}_{\rm S'}\left[(\Lambda_{\rm SS'}\otimes\wtL_{\rm A})(\rho_{\rm SA}\otimes\wt{\eta}_{\rm S'})\right]\right\}\nonumber\\
&\ge\inf_{\Lambda_{\rm S}\in\mathcal{O}_R^N}\supA D\left[(\mE_{\rm S}\otimes\wtL_{\rm A})(\rho_{\rm SA}),(\Lambda_{\rm S}\otimes\wtL_{\rm A})(\rho_{\rm SA})\right]\nonumber\\
&=P_D(\mE_{\rm S}).
\end{align}
In the second line we pick a fixed absolutely free state $\wt{\eta}_{\rm S'}\in\wt{\mathcal{F}}_R$, which is possible due to the property~\ref{Def:Proper:Nonempty}.
Then the second line follows from the fact that $\rho_{\rm SA}\otimes\wt{\eta}_{\rm S'}$ forms a sub-optimal range for the maximization $\overline{\sup}_{\rm A}$.
The third line is because of Eq.~\eqref{Def:ProperDistanceMeasure}.
The fourth line is because the mapping ${\rm tr}_{\rm S'}\left\{\Lambda_{\rm SS'}[(\cdot)\otimes\wt{\eta}_{\rm S'}]\right\}$ will be an $R$-annihilating channel [properties~\ref{Def:Proper:FreeIdentity},~\ref{Def:ProperQR-Tensor}, and~\ref{Def:Proper:Tensor}].
This consequently implies a sup-optimal range for the minimization compared with $\inf_{\Lambda_{\rm S}\in\mathcal{O}_R^N}$.
Then the inequality in Eq.~\eqref{P:Tensor} is proved.

To show the equality, we compute the following for a given $\wtL_{\rm S'}\in\wt{\mathcal{O}}_R^N$:
\begin{align}
P_D(\mE_{\rm S}\otimes\wtL_{\rm S'}) &= \inf_{\Lambda_{\rm SS'}\in\mathcal{O}_R^N}\supA D\left[(\mE_{\rm S}\otimes\wtL_{\rm S'}\otimes\wtL_{\rm A})(\rho_{\rm SS'A}),(\Lambda_{\rm SS'}\otimes\wtL_{\rm A})(\rho_{\rm SS'A})\right]\nonumber\\
&\le\inf_{\Lambda_{\rm S}\in\mathcal{O}_R^N}\supA D\left[(\mE_{\rm S}\otimes\wtL_{\rm S'}\otimes\wtL_{\rm A})(\rho_{\rm SS'A}),(\Lambda_{\rm S}\otimes\wtL_{\rm S'}\otimes\wtL_{\rm A})(\rho_{\rm SS'A})\right]\nonumber\\
&\le\inf_{\Lambda_{\rm S}\in\mathcal{O}_R^N}\supA D\left[(\mE_{\rm S}\otimes\wtL_{\rm A})(\rho_{\rm SA}),(\Lambda_{\rm S}\otimes\wtL_{\rm A})(\rho_{\rm SA})\right]\nonumber\\
&=P_D(\mE_{\rm S}).
\end{align}
The second line is because $\Lambda_{\rm S}\otimes\wtL_{\rm S'}$ with the fixed $\wtL_{\rm S'}$ forms a sub-optimal range for the minimization compared with $\inf_{\Lambda_{\rm SS'}\in\mathcal{O}_R^N}$.
The third line is because $\wtL_{\rm S'}\otimes\wtL_{\rm A}$ forms a sub-optimal range for the maximization of $\overline{\sup}_{\rm A}$ [Eq.~\eqref{Fact:wtL}].
This proves the equality and Eq.~\eqref{P:Tensor}.

Finally, suppose for a given channel $\mathcal{E}$ we have $P_D(\mathcal{E}) = 0$.
By definition this implies
\begin{align}
\inf_{\Lambda_{\rm S}\in\mathcal{O}_R^N}\sup_\rho D[\mathcal{E}(\rho),\Lambda_{\rm S}(\rho)] = 0
\end{align}
since we are allowed to consider the zero-dimensional ancillary system.
Hence, there exists a sequence $\{\Lambda_k\}_{k=1}^\infty\subseteq\mathcal{O}_R^N$ such that $\lim_{k\to\infty}\sup_\rho D[\mathcal{E}(\rho),\Lambda_k(\rho)]=0$.
If $\mathcal{O}_R^N$ is closed under $D$, this implies $\mathcal{E}\in\mathcal{O}_R^N$.
This proves property~\ref{P:Faithful}, and the proof for $P_D$ is completed.

The case for $\bar{P}_D$ is almost the same: One simply needs to replace $\wtL_{\rm A}$ and $\overline{\sup}_{\rm A}$ by $\mathcal{I}_{\rm A}$ and $\sup_{{\rm A};\rho_{\rm SA}}$, respectively.
Also we remark that the proof of property~\ref{P:Monotonicity} for $\bar{P}_D$ is a direct application of Theorem 1 in Ref.~\cite{LiuYuan2019}.
This also means $\bar{P}_D$ can be a monotone if we consider the largest set of possible free super-channels, whenever this set is well-defined (see Ref.~\cite{App:Well-defined} for the explanation).
\end{proof}

\section{Proof of Theorem~\ref{Coro:OperationalMeaning}}\label{App:Proof_Coro:OperationalMeaning}
To sketch the proof, we note that Theorem 10 in Ref.~\cite{LiuWinter2019} is true even without assumptions 3 in their paper, which is crucial for $R$-preservability theories since the identity channel can never be a free channel.
Using all the listed assumptions in Theorem~\ref{Coro:OperationalMeaning}, one can prove the upper bound by the same strategy in Ref.~\cite{LiuWinter2019}.
Also, the small difference between Definition~\ref{Def:PreservabilityDestructionCost} in this work and Definition 9 in Ref.~\cite{LiuWinter2019} will not change the proof of the lower bound.

For the completeness of this work, we still state the detailed proof in this section.
Before the proof, we recall the Generalized Convex-Split Lemma for completely-positive maps~\cite{LiuWinter2019}:

\begin{alemma}\label{Lemma:CSL}
{\em (Generalized Convex-Split Lemma)~\cite{LiuWinter2019}}
Let $\alpha,\beta$ be completely-positive maps with $\norm{\alpha}_\diamond = \norm{\beta}_\diamond = 1$.
Suppose there exists a completely-positive map $\alpha'$ with $\norm{\alpha'}_\diamond\le1$ and $p\in(0,1]$ such that $\beta = p\alpha + (1-p)\alpha'$.
Then the validity of the inequality $\log_2{n}\ge\log_2{\frac{1}{p}} + 2\log_2{\frac{1}{\delta}}$ implies the following estimate
\begin{align}
\norm{\sum_{i=1}^n\frac{1}{n}\beta^{\otimes (i-1)}\otimes\alpha\otimes\beta^{\otimes (n-i)} - \beta^{\otimes n}}_\diamond\le\delta.
\end{align}
\end{alemma}

Before the main proof, we note the following two facts. 
The first one has $\bar{P}_{D_{\rm max}} = L_R$ [Eq.~\eqref{Eq:Robustness}] as a direct consequence.
\begin{afact}\label{Fact:AlternativExpression}
Given two channels $\mathcal{E}$ and $\Lambda$.
Then we have
\begin{align}
\sup_{{\rm A};\rho_{\rm SA}}\inf\left\{\lambda\,|\,0\le[(\lambda\Lambda-\mE)\otimes\mathcal{I}_{\rm A}](\rho_{\rm SA})\right\} = \inf\left\{\lambda\,|\,0\le[(\lambda\Lambda-\mE)\otimes\mathcal{I}_{\rm A}](\rho_{\rm SA})\;\forall{\rm A}\,\&\,\rho_{\rm SA}\right\}.
\end{align}
\end{afact}
\begin{proof}
Define the set $\mathcal{L}_{\bf A}\coloneqq\left\{\lambda\,|\,0\le[(\lambda\Lambda - \mE)\otimes\mathcal{I}_{\rm A}](\rho_{\rm SA})\right\}$ with ${\bf A}$ denotes a particular combination of an ancillary system ${\rm A}$ and a state $\rho_{\rm SA}$ on the system ${\rm SA}$.
Then the left-hand-side can be written as $\sup_{\rm\bf A}\inf\{\lambda\,|\,\lambda\in\mathcal{L}_{\bf A}\}$, and the right-hand-side can be written as $\inf\left\{\lambda\,|\,\lambda\in\bigcap_{\bf A}\mathcal{L}_{\bf A}\right\}$.
With the above notations, the inequality ``$\le$'' follows by the fact that $\bigcap_{\bf A}\mathcal{L}_{\bf A}\subseteq\mathcal{L}_{{\bf A}'}$ for all ${{\bf A}'}$.
On the other hand, consider a given $k\in\mathbb{N}$.
Then there exists an ${\bf A}_k$ such that $\inf\left\{\lambda\,|\,\lambda\in\mathcal{L}_{{\bf A}_k}\right\} + \frac{1}{k}>\sup_{\rm\bf A}\inf\{\lambda\,|\,\lambda\in\mathcal{L}_{\bf A}\}\ge\inf\left\{\lambda\,|\,\lambda\in\mathcal{L}_{{\bf A}_k}\right\}$.
Also, there exists $\lambda_k\in\mathcal{L}_{{\bf A}_k}$ such that $\lambda_k - \frac{1}{k}<\inf\left\{\lambda\,|\,\lambda\in\mathcal{L}_{{\bf A}_k}\right\}\le\lambda_k$.
This means $\lambda_k + \frac{1}{k} > \inf\{\lambda\,|\,\lambda\in\mathcal{L}_{\bf A}\}$ for all ${\bf A}$.
In other words, this means 
$\lambda_k + \frac{1}{k}\in\bigcap_{\bf A}\mathcal{L}_{\bf A}$, which also implies 
\begin{align}
\inf\left\{\lambda\,|\,\lambda\in\bigcap_{\bf A}\mathcal{L}_{\bf A}\right\}&\le\lambda_k + \frac{1}{k}\nonumber\\
&\le\inf\left\{\lambda\,|\,\lambda\in\mathcal{L}_{{\bf A}_k}\right\}+\frac{2}{k}\nonumber\\
&\le\sup_{\rm\bf A}\inf\{\lambda\,|\,\lambda\in\mathcal{L}_{\bf A}\}+\frac{2}{k}.
\end{align}
Since this is true for all $k\in\mathbb{N}$, the result follows.
\end{proof}

The second fact is a property similar to Eq.~\eqref{P:Tensor} for the smooth version of $\bar{P}_D$ defined similarly to Eq.~\eqref{Eq:SmoothP_Dmaxbar}.
\begin{afact}\label{fact:Smooth-Tensor}
Given an $R$-theory and a distance measure $D$ satisfying Eq.~\eqref{Def:ProperDistanceMeasure}.
Then for every $\delta\ge0$ and channels $\mE_{\rm S},\mE_{\rm S'}\in\mathcal{O}_R$, we have
\begin{align}
\bar{P}_D^\delta(\mE_{\rm S}\otimes\mE_{\rm S'})\ge\bar{P}_D^\delta(\mE_{\rm S}).
\end{align}
\end{afact}
\begin{proof}
First, we have the following definition similar to Eq.~\eqref{Eq:SmoothP_Dmaxbar}:
\begin{align}
\bar{P}_D^\delta(\mE_{\rm S}\otimes\mE_{\rm S'}) \coloneqq \inf_{\frac{1}{2}\norm{\mathcal{C}_{\rm SS'} - \mE_{\rm S}\otimes\mE_{\rm S'}}_\diamond\le\delta}\bar{P}_D(\mathcal{C}_{\rm SS'}).
\end{align}
A direct computation shows
\begin{align}
\bar{P}_D(\mathcal{C}_{\rm SS'}) &= \inf_{\Lambda_{\rm SS'}\in\mathcal{O}_R^N}\sup_{\rm A;\rho_{\rm SS'A}}D\left[(\mathcal{C}_{\rm SS'}\otimes\mathcal{I}_{\rm A})(\rho_{\rm SS'A}),(\Lambda_{\rm SS'}\otimes\mathcal{I}_{\rm A})(\rho_{\rm SS'A})\right]\nonumber\\
&\ge\inf_{\Lambda_{\rm SS'}\in\mathcal{O}_R^N}\sup_{\rm A;\rho_{\rm SA}}D\left[(\mathcal{C}_{\rm SS'}\otimes\mathcal{I}_{\rm A})(\rho_{\rm SA}\otimes\wt{\eta}_{\rm S'}),(\Lambda_{\rm SS'}\otimes\mathcal{I}_{\rm A})(\rho_{\rm SA}\otimes\wt{\eta}_{\rm S'})\right]\nonumber\\
&\ge\inf_{\Lambda_{\rm SS'}\in\mathcal{O}_R^N}\sup_{\rm A;\rho_{\rm SA}}D\left[(\mathcal{C}_{\rm S}'\otimes\mathcal{I}_{\rm A})(\rho_{\rm SA}),(\Lambda_{\rm S}'\otimes\mathcal{I}_{\rm A})(\rho_{\rm SA})\right]\nonumber\\
&\ge\inf_{\Lambda_{\rm S}\in\mathcal{O}_R^N}\sup_{\rm A;\rho_{\rm SA}}D\left[(\mathcal{C}_{\rm S}'\otimes\mathcal{I}_{\rm A})(\rho_{\rm SA}),(\Lambda_{\rm S}\otimes\mathcal{I}_{\rm A})(\rho_{\rm SA})\right]\nonumber\\
&=\bar{P}_D(\mathcal{C}_{\rm S}').
\end{align}
The second line is because fixing an absolutely free state $\wt{\eta}_{\rm S'}$ [which is possible due to property~\ref{Def:Proper:Nonempty}] forms a sub-optimal range for the maximization.
The third line is a consequence of the data-processing inequality under partial trace [namely, Eq.~\eqref{Def:ProperDistanceMeasure} and property~\ref{Def:Proper:FreeIdentity}], where we define $\mathcal{C}_{\rm S}'(\cdot)\coloneqq{\rm tr}_{\rm S'}\circ\mathcal{C}_{\rm SS'}[(\cdot)\otimes\wt{\eta}_{\rm S'}]$ and $\Lambda_{\rm S}'(\cdot)\coloneqq{\rm tr}_{\rm S'}\circ\Lambda_{\rm SS'}[(\cdot)\otimes\wt{\eta}_{\rm S'}]$.
Using properties~\ref{Def:Proper:FreeIdentity},~\ref{Def:ProperQR-Tensor}, and~\ref{Def:Proper:Tensor}, we learn that $\Lambda_{\rm S}'\in\mathcal{O}_R^N$ and hence all possible such channels form a sub-optimal range for the minimization compared with all elements in $\mathcal{O}_R^N$.
This explains the fourth line.

Now we note that
\begin{align}\label{Eq:Sub-optimal-range01}
\norm{\mathcal{C}_{\rm SS'} - \mE_{\rm S}\otimes\mE_{\rm S'}}_\diamond&\coloneqq\sup_{\rm A;\rho_{\rm SS'A}}\norm{(\mathcal{C}_{\rm SS'}\otimes\mathcal{I}_{\rm A})(\rho_{\rm SS'A}) - ( \mE_{\rm S}\otimes\mE_{\rm S'}\otimes\mathcal{I}_{\rm A})(\rho_{\rm SS'A})}_1\nonumber\\
&\ge\sup_{\rm A;\rho_{\rm SA}}\norm{(\mathcal{C}_{\rm SS'}\otimes\mathcal{I}_{\rm A})(\rho_{\rm SA}\otimes\wt{\eta}_{\rm S'}) - ( \mE_{\rm S}\otimes\mE_{\rm S'}\otimes\mathcal{I}_{\rm A})(\rho_{\rm SA}\otimes\wt{\eta}_{\rm S'})}_1\nonumber\\
&\ge\sup_{\rm A;\rho_{\rm SA}}\norm{(\mathcal{C}_{\rm S}'\otimes\mathcal{I}_{\rm A})(\rho_{\rm SA}) - ( \mE_{\rm S}\otimes\mathcal{I}_{\rm A})(\rho_{\rm SA})}_1\nonumber\\
& = \norm{\mathcal{C}_{\rm S}' - \mE_{\rm S}}_\diamond.
\end{align}
Again, the second line is due to the sub-optimal range for the maximization when we fix an absolutely free state $\wt{\eta}_{\rm S'}$.
Also, the third line follows from the data-processing inequality (or equivalently, the contractivity) of trace norm under quantum channels.

Finally, we have
\begin{align}
\bar{P}_D^\delta(\E_{\rm S}\otimes\mE_{\rm S'})&\ge \inf_{\frac{1}{2}\norm{\mathcal{C}_{\rm SS'} - \mE_{\rm S}\otimes\mE_{\rm S'}}_\diamond\le\delta}\bar{P}_D(\mathcal{C}_{\rm S}')\nonumber\\
&\ge\inf_{\frac{1}{2}\norm{\mathcal{C}_{\rm S}' - \mE_{\rm S}}_\diamond\le\delta}\bar{P}_D(\mathcal{C}_{\rm S}')\nonumber\\
&\ge\inf_{\frac{1}{2}\norm{\mE' - \mE_{\rm S}}_\diamond\le\delta}\bar{P}_D(\mE')\nonumber\\
&=\bar{P}_D^\delta(\mE_{\rm S}).
\end{align}
The second line follows from Eq.~\eqref{Eq:Sub-optimal-range01}, which implies that all $\mathcal{C}_{\rm SS'}$ satisfying $\frac{1}{2}\norm{\mathcal{C}_{\rm SS'} - \mE_{\rm S}\otimes\mE_{\rm S'}}_\diamond\le\delta$ form a subset of the set of all $\mathcal{C}_{\rm SS'}$ satisfying $\frac{1}{2}\norm{\mathcal{C}_{\rm S}' - \mE_{\rm S}}_\diamond\le\delta$.
Since all possible $\mathcal{C}_{\rm S}'$ form a subset of all possible channels $\mE'$ satisfying $\frac{1}{2}\norm{\mE' - \mE_{\rm S}}_\diamond\le\delta$, we have the third line.
The proof is completed.
\end{proof}

Now we start the proof of Theorem~\ref{Coro:OperationalMeaning}:

\begin{proof}
We follow the same strategy in the proof of Theorem 10 in Ref.~\cite{LiuWinter2019}.
We will show the upper bound at first.

{\em Proof of the upper bound.}--
At the very beginning, consider an arbitrarily given positive integer $l\in\mathbb{N}$.
By definition, there exists a channel $\mE_l$ such that $\norm{\mE_l - \mE}_\diamond\le2(\epsilon-\eta)$ and 
\begin{align}
\bar{P}_{D_{\rm max}}^{\epsilon - \eta}(\mE) \le \bar{P}_{D_{\rm max}}(\mE_l)\le \bar{P}_{D_{\rm max}}^{\epsilon - \eta}(\mE) + \frac{1}{l}.
\end{align}
Also, because we have $\mathcal{O}_R^N = \wt{\mathcal{O}}_R^N$, there exists a channel $\bar{\Lambda}_l\in\wt{\mathcal{O}}_R^N$ such that
\begin{align}\label{Eq:P_DmaxEl-bar}
\bar{P}_{D_{\rm max}}(\mE_l)&\le\sup_{{\rm A};\rho_{\rm SA}} D_{\rm max}\left[(\mE_l\otimes\mathcal{I}_{\rm A})(\rho_{\rm SA})\|(\bar{\Lambda}_l\otimes\mathcal{I}_{\rm A})(\rho_{\rm SA})\right]\nonumber\\
&\coloneqq \sup_{{\rm A};\rho_{\rm SA}}\log_2\inf\left\{\lambda\,|\,0\le[(\lambda\bar{\Lambda}_l - \mE_l)\otimes\mathcal{I}_{\rm A}](\rho_{\rm SA})\right\}\nonumber\\
&=\log_2\inf\left\{\lambda\,|\,0\le[(\lambda\bar{\Lambda}_l - \mE_l)\otimes\mathcal{I}_{\rm A}](\rho_{\rm SA})\;\forall{\rm A}\,\&\,\rho_{\rm SA}\right\}\nonumber\\
&=-\log_2\sup\left\{q\in[0,1]\,|\,0 \le [(\bar{\Lambda}_l - q\mE_l)\otimes\mathcal{I}_{\rm A}](\rho_{\rm SA})\;\forall{\rm A}\,\&\,\rho_{\rm SA}\right\}\nonumber\\
&\le \bar{P}_{D_{\rm max}}(\mE_l) + \frac{1}{l},
\end{align}
where in the third line we use Fact~\ref{Fact:AlternativExpression} and the fourth line is due to the fact that $\lambda<1$ is forbidden in the minimization range (otherwise there exist quantum states $\sigma$ and $\sigma'$ such that $\lambda\sigma - \sigma'\ge0$ for some $\lambda<1$, which is impossible since this implies $0\le{\rm tr}(\lambda\sigma - \sigma') = \lambda - 1<0$).
Let $\mathcal{U}_i$ be the pair-wise permutation unitary channel between the first and the $i$th subsystems (that is, the swap unitary between the two subsystems).
Then we consider the destruction process with $\bar{\Lambda}_l^{\otimes (n-1)}$ and $\left\{\mathcal{U}_i,\mathcal{U}_i,p_i = \frac{1}{n}\right\}_{i=1}^{n}$ \footnote{Note that in general the two pair-wise permutation channels are different because they may act on different spaces. Here we simply use the same notation to stress the fact that both of them are permutation unitary channels between the first and the $i$th subsystems.
More precisely, if $\bar{\Lambda}_l,\mE_l:{\rm S}\to{\rm S'}$, then we have ${\rm S}^{\otimes n}\to{\rm S}^{\otimes n}$ for the pre-processing permutations and ${\rm S'}^{\otimes n}\to{\rm S'}^{\otimes n}$ for the post-processing permutations.}, which gives the following:
\begin{align}\label{Eq:eta-destruction-process}
\sum_{i=1}^n\frac{1}{n}\mathcal{U}_i\circ\left(\mE_l\otimes\bar{\Lambda}_l^{\otimes (n-1)}\right)\circ \mathcal{U}_i = \frac{1}{n}\sum_{i=1}^n\bar{\Lambda}_l^{\otimes (i-1)}\otimes\mE_l\otimes\bar{\Lambda}_l^{\otimes (n-i)}.
\end{align}
From Eq.~\eqref{Eq:P_DmaxEl-bar} we note that when $\log_2{n} > \bar{P}_{D_{\rm max}}(\mathcal{E}_l) + \frac{1}{l} + 2\log_2\frac{1}{2\eta}$ [which automatically implies $\bar{P}_{D_{\rm max}}(\mathcal{E}_l)<\infty$], there always exists an $p_l\in(0,1)$ such that
\begin{itemize}
\item $\log_2{n} > \log_2\frac{1}{p_l} + 2\log_2\frac{1}{2\eta}$.
\item $\bar{\Lambda}_l - p_l\mathcal{E}_l$ is completely-positive.
\end{itemize}
By defining $\alpha'_l\coloneqq\frac{1}{1-p_l}(\bar{\Lambda}_l - p_l\mathcal{E}_l)$, one can see that $\alpha'_l$ is completely-positive and trace-preserving [since both $\bar{\Lambda}_l$ and $\mathcal{E}_l$ are trace-preserving and we have $\bar{\Lambda}_l=p_l\mathcal{E}_l+(1-p_l)\alpha'_l$; note that $p_l<1$].
This means $\alpha'_l$ is also a quantum channel (i.e. a completely-positive trace-preserving map), thereby having $\norm{\alpha'_l}_\diamond = 1$.
Then Lemma~\ref{Lemma:CSL} (with $\alpha=\mathcal{E}_l$ and $\beta=\bar{\Lambda}_l$) implies that when $\log_2{n} > \bar{P}_{D_{\rm max}}(\mathcal{E}_l) + \frac{1}{l} + 2\log_2\frac{1}{2\eta}$ holds, then we have
\begin{align}\label{Eq:approx}
\norm{\frac{1}{n}\sum_{i=1}^n\bar{\Lambda}_l^{\otimes (i-1)}\otimes\mE_l\otimes\bar{\Lambda}_l^{\otimes (n-i)} - \bar{\Lambda}_l^{\otimes n}}_\diamond\le 2\eta;
\end{align}
in other words, Eq.~\eqref{Eq:eta-destruction-process} forms an $\eta$-destruction process for $\mathcal{E}_l$.
(Note that $\bar{\Lambda}_l^{\otimes n}\in\wt{\mathcal{O}}_R^N$ since we assume no activation property).
This also implies the existence of an $\epsilon$-destruction process for $\mE$ since
\begin{align}
\norm{\frac{1}{n}\sum_{i=1}^n\bar{\Lambda}_l^{\otimes (i-1)}\otimes\mE\otimes\bar{\Lambda}_l^{\otimes (n-i)} - \bar{\Lambda}_l^{\otimes n}}_\diamond&\le\norm{\mE - \mE_l}_\diamond + \norm{\frac{1}{n}\sum_{i=1}^n\bar{\Lambda}_l^{\otimes (i-1)}\otimes\mE_l\otimes\bar{\Lambda}_l^{\otimes (n-i)} - \bar{\Lambda}_l^{\otimes n}}_\diamond\nonumber\\
&\le2(\epsilon - \eta) + 2\eta = 2\epsilon,
\end{align}
where we use the relation $\norm{\mE - \mE_l}_\diamond\le2(\epsilon - \eta)$, data-processing inequality, and triangle inequality.

Finally, let $n' = \min\left\{n\in\mathbb{N}\,|\,\log_2{n}>\bar{P}_{D_{\rm max}}(\mE_l) + \frac{1}{l} + 2\log_2{\frac{1}{2\eta}}\right\}$.
Since $C_R^\epsilon(\mE)\coloneqq\min\log_2{n}$ and the minimization is taken over all $\epsilon$-destruction processes, we conclude the following
\begin{align}
C_R^\epsilon(\mE)&\le\log_2{n'}\nonumber\\
&\le \bar{P}_{D_{\rm max}}(\mE_l) + \frac{1}{l} + 2\log_2{\frac{1}{2\eta}} + \max_{x\in\mathbb{N}}\left[\log_2(x+1) - \log_2{x}\right]\nonumber\\
&\le \bar{P}_{D_{\rm max}}(\mE_l) + \frac{1}{l} + 2\log_2{\frac{1}{2\eta}} + 1\nonumber\\
&\le \bar{P}_{D_{\rm max}}^{\epsilon - \eta}(\mE) + \frac{2}{l}+ 2\log_2{\frac{1}{\eta}} - 1,
\end{align}
and the proof of the upper bound is completed since the above estimate works for all positive integer $l$.

{\em Proof of the lower bound.}--
The proof is completely the same with the proof of the lower bound of Theorem 10 in Ref.~\cite{LiuWinter2019}, and we briefly sketch it.
Consider a given $\mE_{\rm S}\in\mathcal{O}_R$.
Then for a given $\epsilon$-destruction process of $R$-preservability consisting of $\bar{\Lambda}_{\rm S'}\in\wt{\mathcal{O}}_R^N$ and $\{\mathcal{U}_i,\mathcal{V}_i,p_i\}_{i=1}^{K}$, we have 
\begin{align}
\norm{\sum_{i=1}^Kp_i\mathcal{N}_i - \Lambda_{\rm SS'}}_\diamond\le2\epsilon,
\end{align}
where $\Lambda_{\rm SS'}\in\mathcal{O}_R^N$ and $\mathcal{N}_i\coloneqq \mathcal{U}_i\circ(\mE_{\rm S}\otimes\bar{\Lambda}_{\rm S'})\circ \mathcal{V}_i$.
This $\epsilon$-destruction process of $R$-preservability can also be interpreted as an $\epsilon$-destruction process defined in Ref.~\cite{LiuWinter2019} by identifying $\mathcal{O}_R$ as free channels in their framework (that is, in the proof of the lower bound of Theorem 10 in Ref.~\cite{LiuWinter2019} we set $\mathcal{N} = \mathcal{E}_{\rm S}\in\mathcal{O}_R$, $\mathcal{F} = \bar{\Lambda}_{\rm S'}\in\wt{\mathcal{O}}_R^N$, $\mathcal{M} = \Lambda_{\rm SS'}\in\mathcal{O}_R^N$, and consider the channel resource theory with free channels as $\mathcal{O}_R$.).
The same proof applies until we reach the following inequality, which is the last inequality in the bottom of page 15 in Ref.~\cite{LiuWinter2019} (the assumptions made in Theorem~\ref{Coro:OperationalMeaning} plus the properties~\ref{Def:Proper:Nonempty},~\ref{Def:Proper:FreeIdentity},~\ref{Def:ProperQR-Tensor},~\ref{Def:Proper:Tensor} ensure the applicability of the proof of Theorem 10 in Ref.~\cite{LiuWinter2019} when we identify $\mathcal{O}_R$ as free channels in their setting):
\begin{align}
\norm{\mathcal{E}_{\rm S}\otimes\bar{\Lambda}_{\rm S'} - \sum_{i=1}^Kp_i\mathcal{U}_i^\dagger\circ\mathcal{M}_i\circ \mathcal{V}_i^\dagger}_\diamond\le\sqrt{\epsilon(2-\epsilon)},
\end{align}
where $\mathcal{M}_i$'s are completely-positive maps satisfying $\sum_{i=1}^Kp_i\mathcal{M}_i = \Lambda_{\rm SS'}$, which means $p_i\mathcal{M}_i \le \Lambda_{\rm SS'}$ for all $i$ (by writing $\mathcal{E}\le\mathcal{E}'$ for two channel $\mE$ and $\mE'$ we means $\mE'-\mE$ is completely-positive).
Hence, we have
\begin{align}
\sum_{i=1}^Kp_i\mathcal{U}_i^\dagger\circ\mathcal{M}_i\circ \mathcal{V}_i^\dagger\le\sum_{i=1}^K\mathcal{U}_i^\dagger\circ\Lambda_{\rm SS'}\circ \mathcal{V}_i^\dagger.
\end{align}
Note that the left-hand-side is a channel.
Because for any $R$-theory considered in this work the set $\mathcal{O}_R^N$ is by definition convex, we have $\frac{1}{K}\sum_{i=1}^K\mathcal{U}_i^\dagger\circ\Lambda_{\rm SS'}\circ \mathcal{V}_i^\dagger\in\mathcal{O}_R^N$.
Hence, we conclude
\begin{align}
\bar{P}_{D_{\rm max}}^{\sqrt{\epsilon(2-\epsilon)}}(\mE_{\rm S}\otimes\bar{\Lambda}_{\rm S'})\le\bar{P}_{D_{\rm max}}\left(\sum_{i=1}^Kp_i\mathcal{U}_i^\dagger\circ\mathcal{M}_i\circ \mathcal{V}_i^\dagger\right)\le \log_2{K}.
\end{align}
Using Fact~\ref{fact:Smooth-Tensor}, we conclude that $\bar{P}_{D_{\rm max}}^{\sqrt{\epsilon(2-\epsilon)}}(\mE_{\rm S})\le\log_2K$ for all possible $K$.
This completes the proof.
\end{proof}

\section{Proof of Theorem~\ref{Coro:BathSize}}\label{App:BathSizeProof}
Before the proof, we need to recap certain key ingredients in Ref.~\cite{Sparaciari2019}.
The first one is a central assumption called {\em energy subspace condition} (we use the notation ${\bf m} = (m_1,m_2,...,m_d)$ to denote a vector in $\mathbb{N}^d$):
\begin{adefinition}\label{Def:ESC}
{\rm (Energy Subspace Condition)~\cite{Sparaciari2019}}
A given Hamiltonian $H$ with energy levels $\{E_i\}_{i=1}^{d}$ is said to fulfill the {\em energy subspace condition} if for any positive integer $M$ and two different vectors $\{{\bf m}\neq{\bf m}'\}\subset\mathbb{N}^d$ satisfying $\sum_{i=1}^{d}m_i = \sum_{i=1}^{d}m'_i = M$, we have
\begin{align}
\sum_{i=1}^dm_iE_i\neq \sum_{i=1}^dm'_iE_i.
\end{align}
\end{adefinition}
Roughly speaking, Definition~\ref{Def:ESC} means energy levels cannot be integer multiples of each other.
This condition also forbids the possibility of degeneracy (otherwise one can simply switch the coefficients of a vector ${\bf m}$ in a subspace with degeneracy to construct a counterexample).

Before mentioning the main results in Ref.~\cite{Sparaciari2019}, we define the smooth max-relative entropy as [also recall Eq.~\eqref{Eq:max-relative-entropy}]
\begin{align}
D_{\rm max}^\epsilon(\rho||\sigma)\coloneqq\inf_{\frac{1}{2}\norm{\rho' - \rho}_1\le\epsilon}D_{\rm max}(\rho'||\sigma).
\end{align}
Note that there is a difference of $\frac{1}{2}$ factor compared with Eq.~(11) in Ref.~\cite{Sparaciari2019}.
Then we have~\cite{Sparaciari2019} ($\gamma$ is the thermal state associated with the given bath temperature $T$ and system Hamiltonian $H_{\rm S}$):
\begin{atheorem}\label{Thm:Carlo}
{\rm\cite{Sparaciari2019}} For a given state $\rho_{\rm S}$, we have
\begin{align}
n_\epsilon(\rho_{\rm S}) \le \frac{1}{\epsilon^2}2^{D_{\rm max}(\rho_{\rm S}||\gamma)} + 1.
\end{align}
Moreover, if the system Hamiltonian $H_{\rm S}$ satisfies the energy subspace condition and $\rho_{\rm S}$ is diagonal in the energy eigenbasis of $H_{\rm S}$, then we also have
\begin{align}\label{Eq:Carlo}
D_{\rm max}^{\sqrt{\epsilon}}(\rho_{\rm S}||\gamma)\le \log_2n_\epsilon(\rho_{\rm S}).
\end{align}
\end{atheorem}

Now the idea is to use the above theorem to prove Theorem~\ref{Coro:BathSize}.
But before the proof, we still need to establish the following lemma regarding the continuity of the max-relative entropy (in a finite dimensional case, we say a quantum state is {\em full-rank} if it has only positive eigenvalues; in other words, its support coincides with the whole Hilbert space):
\begin{alemma}\label{Lemma:Continuity-D_max}
Given three states $\rho,\rho',\sigma$ and $\sigma$ is full-rank.
Then we have
\begin{align}
\left|2^{D_{\rm max}(\rho'||\sigma)} - 2^{D_{\rm max}(\rho||\sigma)}\right| \le \frac{\norm{\rho - \rho'}_1}{p_{\rm min}(\sigma)},
\end{align}
where $p_{\rm min}(\sigma)$ is the smallest eigenvalue of $\sigma$.
\end{alemma}
\begin{proof}
Define the set $\mathcal{L}(\rho||\sigma)\coloneqq\{\lambda\ge0\;|\;\rho\le\lambda\sigma\}$.
Then one can rewrite the definition of the max-relative entropy as 
$
D_{\rm max}(\rho||\sigma) = \inf_{\lambda\in\mathcal{L}(\rho||\sigma)}\log_2\lambda.
$
Now we note that $\lambda\in\mathcal{L}(\rho||\sigma)$ if and only if $\lambda\sigma - \rho\ge0$, which is true if and only if 
\begin{align}\label{Eq:IffCondition}
\inf_{\ket{\phi}}\bra{\phi}(\lambda\sigma - \rho)\ket{\phi}\ge0.
\end{align}
This implies the following estimate for any $\lambda\in\mathcal{L}(\rho||\sigma)$:
\begin{align}
\inf_{\ket{\phi}}\bra{\phi}(\lambda\sigma - \rho)\ket{\phi}& = \inf_{\ket{\phi}}\left[\bra{\phi}(\lambda\sigma - \rho')\ket{\phi} + \bra{\phi}(\rho'-\rho)\ket{\phi}\right]\nonumber\\
&\ge \inf_{\ket{\phi}}\bra{\phi}(\lambda\sigma - \rho')\ket{\phi} + \inf_{\ket{\phi}}\bra{\phi}(\rho'-\rho)\ket{\phi}\nonumber\\
&\ge \inf_{\ket{\phi}}\bra{\phi}(\lambda\sigma - \rho')\ket{\phi} - \norm{\rho - \rho'}_1,
\end{align}
where in the last line we use the relation $\inf_{\ket{\phi}}\bra{\phi}(\rho'-\rho)\ket{\phi} = -\sup_{\ket{\phi}}\bra{\phi}(\rho-\rho')\ket{\phi}\ge-\norm{\rho - \rho'}_\infty\ge-\norm{\rho - \rho'}_1$ (recall that $\norm{\cdot}_\infty\coloneqq\sup_{\ket{\psi}}|\bra{\psi}\cdot\ket{\psi}|$ and $\norm{\cdot}_\infty\le\norm{\cdot}_1$).
Since the argument also works when we exchange the roles of $\rho$ and $\rho'$, we conclude the following bound:
\begin{align}\label{Eq:EstimateRhoRho}
\left|\inf_{\ket{\phi}}\bra{\phi}(\lambda\sigma - \rho')\ket{\phi} - \inf_{\ket{\phi}}\bra{\phi}(\lambda\sigma - \rho)\ket{\phi}\right| \le \norm{\rho - \rho'}_1.
\end{align} 
With the help of the above bound, we have the following computation for a given $\lambda\in\mathcal{L}(\rho||\sigma)$:
\begin{align}
0\le\inf_{\ket{\phi}}\bra{\phi}(\lambda\sigma - \rho)\ket{\phi}&\le\inf_{\ket{\phi}}\bra{\phi}(\lambda\sigma - \rho')\ket{\phi} + \norm{\rho - \rho'}_1\nonumber\\
&=\inf_{\ket{\phi}}\bra{\phi}\left[\left(\lambda + \frac{\norm{\rho - \rho'}_1}{\bra{\phi}\sigma\ket{\phi}}\right)\sigma - \rho'\right]\ket{\phi}\nonumber\\
&\le\inf_{\ket{\phi}}\bra{\phi}\left[\left(\lambda + \frac{\norm{\rho - \rho'}_1}{p_{\rm min}(\sigma)}\right)\sigma - \rho'\right]\ket{\phi},
\end{align}
where in the second line we have $\bra{\phi}\sigma\ket{\phi}>0$ since $\sigma$ is full-rank and hence has no zero eigenvalue.
Also, $\bra{\phi}\sigma\ket{\phi}$ is lower bounded by $p_{\rm min}(\sigma)=\inf_{\ket{\psi}}\bra{\psi}\sigma\ket{\psi}$ for all $\ket{\phi}$.
This computation implies $\lambda + \frac{\norm{\rho - \rho'}_1}{p_{\rm min}(\sigma)}\in\mathcal{L}(\rho'||\sigma)$ whenever $\lambda\in\mathcal{L}(\rho||\sigma)$ [recall Eq.~\eqref{Eq:IffCondition}].
As a consequence, we have
\begin{align}
2^{D_{\rm max}(\rho'||\sigma)}& = \inf_{\lambda\in\mathcal{L}(\rho'||\sigma)}\lambda\nonumber\\
&\le\inf_{\lambda\in\mathcal{L}(\rho||\sigma)}\left(\lambda + \frac{\norm{\rho - \rho'}_1}{p_{\rm min}(\sigma)}\right)\nonumber\\
& = 2^{D_{\rm max}(\rho||\sigma)} + \frac{\norm{\rho - \rho'}_1}{p_{\rm min}(\sigma)}. 
\end{align}
Then the desired bound can be proved by exchanging the roles of $\rho$ and $\rho'$ and apply the same arguement again.
\end{proof}
As a remark, we note that the above lemma implies the Lipschitz continuity of the function $2^{D_{\rm max}(\cdot||\sigma)}$ when $\sigma$ is full-rank.
With Theorem~\ref{Thm:Carlo} and Lemma~\ref{Lemma:Continuity-D_max} in hand, we now start the proof of Theorem~\ref{Coro:BathSize}.
\begin{proof}
When we consider the $R$-theory of athermality equipped with Gibbs-preserving maps, the only absolutely $R$-annihilating channel (with output space ${\rm A}$ and output dimension $d^k$) is the full thermalization channel $\Phi_{\gamma_{\rm A}}:(\cdot)\mapsto\gamma_{\rm A}$, where $\gamma_{\rm A} = \gamma^{\otimes k}$ and $\gamma$ is the given thermal state in the system ${\rm S}$.
Then direct computation shows [recall Eqs.~\eqref{Eq:supA} and~\eqref{Eq:P_D} for notations]
\begin{align}
P_{D_{\rm max}}(\mathcal{N})&\coloneqq\inf_{\Lambda_{\rm S}\in\mathcal{O}_R^N}\overline{\sup_{\rm A}}D_{\rm max}[(\mathcal{N}\otimes\widetilde{\Lambda}_{\rm A})(\rho_{\rm SA})||(\Lambda_{\rm S}\otimes\widetilde{\Lambda}_{\rm A})(\rho_{\rm SA})]\nonumber\\
& = \sup_{{\rm A};\rho_{\rm SA}}D_{\rm max}[(\mathcal{N}\otimes\Phi_{\gamma_{\rm A}})(\rho_{\rm SA})||\Phi_{\gamma\otimes\gamma_{\rm A}}(\rho_{\rm SA})]\nonumber\\
& =\sup_{{\rm A};\rho_{\rm S}}D_{\rm max}[\mathcal{N}(\rho_{\rm S})\otimes\gamma_{\rm A}||\gamma\otimes\gamma_{\rm A}]\nonumber\\
&=\sup_{\rho} D_{\rm max}[\mathcal{N}(\rho)||\gamma],
\end{align}
where the last equality follows from the fact that for an operator $A$ and a positive operator $E$ we have $A\ge0$ if and only if $A\otimes E\ge0$, which implies the relation $D_{\rm max}(\rho\otimes\eta||\sigma\otimes\eta) = D_{\rm max}(\rho||\sigma)$ for all quantum states $\rho,\sigma,\eta$.
Together with Theorem~\ref{Thm:Carlo} and the quantity defined in Eq.~\eqref{Eq:BathSizeChannel}, we conclude that
\begin{align}
\mathcal{B}^\epsilon(\mathcal{N})&\coloneqq\sup_\rho n_\epsilon[\mathcal{N}(\rho)] - 1\nonumber\\
&\le\sup_\rho\frac{1}{\epsilon^2}2^{D_{\rm max}[\mathcal{N}(\rho)||\gamma]}\nonumber\\
&=\frac{1}{\epsilon^2}2^{P_{D_{\rm max}}(\mathcal{N})},
\end{align}
and the upper bound is proved.

To see the lower bound, first we note that being coherence-annihilating for the given Gibbs-preserving channel $\mathcal{N}$ means $\mathcal{N}(\rho)$ is diagonal in the energy eigenbasis for all inputs $\rho$.
Applying Theorem~\ref{Thm:Carlo} and Lemma~\ref{Lemma:Continuity-D_max}, one can conclude that
\begin{align}
1 + \mathcal{B}^{\epsilon}(\mathcal{N})&\coloneqq\sup_\rho n_{\epsilon}[\mathcal{N}(\rho)]\nonumber\\
&\ge\sup_\rho 2^{D_{\rm max}^{\sqrt{\epsilon}}[\mathcal{N}(\rho)||\gamma]}\nonumber\\
&=\sup_\rho\inf_{\frac{1}{2}\norm{\rho' - \mathcal{N}(\rho)}_1\le\sqrt{\epsilon}}2^{D_{\rm max}(\rho'||\gamma)}\nonumber\\
&\ge\sup_{\rho}\left(2^{D_{\rm max}[\mathcal{N}(\rho)|\gamma]} - \frac{2\sqrt{\epsilon}}{p_{\rm min}(\gamma)}\right)\nonumber\\
&=2^{P_{D_{\rm max}}(\mathcal{N})} - \frac{2\sqrt{\epsilon}}{p_{\rm min(\gamma)}},
\end{align}
and the proof is completed.
\end{proof}

\section{Non-Signalling Assisted Classical Communication Scenario}\label{App:CC}
Recently, the application of a channel resource theory to classical communication scenario has been addressed~\cite{Takagi2019-3}.
The central question is how much classical information (in terms of classical bits) can be transmitted noiselessly (or up to certain error) via a channel of the following form:
\begin{align}\label{Eq:Non-signaling}
\mathcal{E}_d\circ(\mathcal{N}\otimes\mathcal{I}_{\rm A})\circ\mathcal{E}_e,
\end{align}
where $\mathcal{N}$ is the noisy channel which is sending the information, and $\mathcal{E}_e,\mathcal{E}_d$ are the encoding and the decoding channels, respectively.
Note that we have $\mathcal{E}_e:{\rm C\to SA}$ and $\mathcal{E}_d:{\rm SA\to C}$, where ${\rm C}$ is the space for the classical bits, which are represented by the orthonormal basis $\{\ket{m}\}_{m=0}^{M-1}$.
As explained in Ref.~\cite{Takagi2019-3}, this structure gives the non-signalling assisted classical communication scenario: The transmission of the classical information is assisted by any possible non-signalling structure, which has the general form given by Eq.~\eqref{Eq:Non-signaling}.
In this case, both $\mathcal{E}_e,\mathcal{E}_d$ can be arbitrary channels.

To quantify how much classical information is transmitted successfully within a given error $\epsilon\in(0,1)$, we consider the following averaged error associated to a given combination of encoding channel $\mathcal{E}_e$, decoding channel $\mathcal{E}_d$, and transmitting channel $\mathcal{N}$~\cite{Takagi2019-3}:
\begin{align}
\varepsilon(\mathcal{E}_e,\mathcal{E}_d,\mathcal{N})\coloneqq 1 - \frac{1}{M}\sum_{m=0}^{M-1}\bra{m}\mathcal{E}_d\circ(\mathcal{N}\otimes\mathcal{I}_{\rm A})\circ\mathcal{E}_e(\proj{m})\ket{m}.
\end{align}
Then one can define the corresponding single-shot classical capacity with error $\epsilon$ as~\cite{Takagi2019-3}:
\begin{align}
C^\epsilon_{\rm NS, (1)}(\mathcal{N})\coloneqq\sup_{\mathcal{E}_e,\mathcal{E}_d}\{\log_2M\,|\,\varepsilon(\mathcal{E}_e,\mathcal{E}_d,\mathcal{N})\le\epsilon\}.
\end{align}
Note that in the non-signalling assisted scenario, the above optimization is taken over all possible channels $\mathcal{E}_e,\mathcal{E}_d$.
More precisely, the encoding map $\mE_e$ can be understood as an effectively ``classical-quantum'' channel because we only consider inputs of the form $\proj{m}$.
Similarly, the decoding map $\mE_d$ can be interpreted as an effectively ``quantum-classical'' channel. 
The above classical capacity indicates the optimal amount of classical bits that can be transmitted and recovered within the given error $\epsilon$ when the only constraint is the non-signalling condition.
From here one can observe that the setup in Sec.~\ref{Sec:CC} is equivalent to the above setup plus the two imposed thermodynamic constraints.

\section{Proof of Theorem~\ref{Result:ConverseBoundAthermality}}\label{App:ConverseBoundAthermality}
The strategy is to follow the spirit of the proof of Theorem 3 in Ref.~\cite{Takagi2019-3}. 
Before the main proof, we need to establish two facts.
First, recall from Eq.~\eqref{Eq:AveError} the following expression for a given combination $(\mathcal{N},\mathcal{E}_e,\mathcal{E}_d,\gamma_{\rm A})$:
\begin{align}
1 - \varepsilon(\mathcal{N},\mathcal{E}_e,\mathcal{E}_d,\gamma_{\rm A}) = \frac{1}{M}\sum_{m=0}^{M-1}\bra{m}\mathcal{E}_d\circ(\mathcal{N}\otimes\Phi_{\gamma_{\rm A}})\circ\mathcal{E}_e(\proj{m})\ket{m}
\end{align}
Here we note that $\gamma_{\rm A} = \gamma^{\otimes k}$ for a positive integer $k$, where $\gamma$ is the thermal state associated with the given $R$-theory of athermality in the main system ${\rm S}$.
Again, $\Phi_{\gamma_{\rm A}}:(\cdot)\mapsto\gamma_{\rm A}$ is the full thermalization (or, equivalently, the state preparation channel) with the target thermal state $\gamma_{\rm A}$.
We remark that throughout this section we assume the channels $\mathcal{N},\mathcal{N}'$ have the same output space ${\rm S}$.
Now, we note the following result, which is similar to Lemma 4 in Ref.~\cite{Takagi2019-3}:
\begin{afact}\label{Fact:0}
Given Gibbs-preserving channels $\mathcal{N},\mathcal{N}',\mathcal{E}_e,\mathcal{E}_d$, and a thermal state $\gamma_{\rm A}$.
Then we have
\begin{align}
\left|\sup_{\mathcal{E}\in\mathcal{O}_R}[1 - \varepsilon(\mathcal{N},\mathcal{E}_e,\mathcal{E},\gamma_{\rm A})] - \sup_{\mathcal{E}\in\mathcal{O}_R}[1 - \varepsilon(\mathcal{N}',\mathcal{E}_e,\mathcal{E},\gamma_{\rm A})]\right| \le \frac{1}{2}\norm{\mathcal{N} - \mathcal{N}'}_\diamond.
\end{align}
\end{afact}
\begin{proof}
We assume $\sup_{\mathcal{E}\in\mathcal{O}_R}[1 - \varepsilon(\mathcal{N},\mathcal{E}_e,\mathcal{E},\gamma_{\rm A})] \ge \sup_{\mathcal{E}\in\mathcal{O}_R}[1 - \varepsilon(\mathcal{N}',\mathcal{E}_e,\mathcal{E},\gamma_{\rm A})]$ without loss of generality.
For every positive integer $k\in\mathbb{N}$, let $\mathcal{E}^{(k)}$ be the Gibbs-preserving map satisfying $\sup_{\mathcal{E}\in\mathcal{O}_R}[1 - \varepsilon(\mathcal{N},\mathcal{E}_e,\mathcal{E},\gamma_{\rm A})]\le 1 - \varepsilon(\mathcal{N},\mathcal{E}_e,\mathcal{E}^{(k)},\gamma_{\rm A}) + \frac{1}{k}$.
Then we have
\begin{align}
&\left|\sup_{\mathcal{E}\in\mathcal{O}_R}[1 - \varepsilon(\mathcal{N},\mathcal{E}_e,\mathcal{E},\gamma_{\rm A})] - \sup_{\mathcal{E}\in\mathcal{O}_R}[1 - \varepsilon(\mathcal{N}',\mathcal{E}_e,\mathcal{E},\gamma_{\rm A})]\right|\nonumber\\
&\le \frac{1}{k} + 1 - \varepsilon(\mathcal{N},\mathcal{E}_e,\mathcal{E}^{(k)},\gamma_{\rm A}) - [1 - \varepsilon(\mathcal{N}',\mathcal{E}_e,\mathcal{E}^{(k)},\gamma_{\rm A})]\nonumber\\
& = \frac{1}{k} + \frac{1}{M}\sum_{m=0}^{M-1}{\rm tr}[E^{(k)}_m ((\mathcal{N} - \mathcal{N}')\otimes\Phi_{\gamma_{\rm A}})\circ\mathcal{E}_e(\proj{m})],
\end{align}
where $\{E^{(k)}_m\}_{m=0}^{M-1}$ is a POVM defined by $E^{(k)}_m\coloneqq\mathcal{E}^{(k),\dagger}(\proj{m})$ since each $\mathcal{E}^{(k),\dagger}$ is a completely-positive unital map.
Following the proof of Lemma 4 in Ref.~\cite{Takagi2019-3}, we note the following estimate:
\begin{align}
\norm{\mathcal{N} - \mathcal{N}'}_\diamond \ge \norm{(\mathcal{N} - \mathcal{N}')(\rho)}_1 = \sup_{0\le E\le\id}2{\rm tr}[E(\mathcal{N} - \mathcal{N}')(\rho)].
\end{align}
Then we conclude that
\begin{align}
\left|\sup_{\mathcal{E}\in\mathcal{O}_R}[1 - \varepsilon(\mathcal{N},\mathcal{E}_e,\mathcal{E},\gamma_{\rm A})] - \sup_{\mathcal{E}\in\mathcal{O}_R}[1 - \varepsilon(\mathcal{N}',\mathcal{E}_e,\mathcal{E},\gamma_{\rm A})]\right|&\le\frac{1}{k} + \frac{1}{2M}\sum_{m=0}^{M-1}\norm{(\mathcal{N} - \mathcal{N}')\otimes\Phi_{\gamma_{\rm A}}}_\diamond\nonumber\\
&\le\frac{1}{k} + \frac{1}{2}\norm{\mathcal{N} - \mathcal{N}'}_\diamond,
\end{align}
where the last inequality follows from the data-processing inequality of the trace norm (or, equivalently, the contractivity under quantum channels).
Since this argument works for all positive integer $k$, the desired upper bound is proved.
\end{proof}
We still need to show another fact, which is similar to Theorem 5 in Ref.~\cite{Takagi2019}:
\begin{afact}\label{Fact:1}
\begin{align}
1 - \varepsilon(\mathcal{N},\mathcal{E}_e,\mathcal{E}_d,\gamma_{\rm A}) \le \frac{1}{M}\times2^{\bar{P}_{D_{\rm max}}(\mathcal{N})}.
\end{align}
\end{afact}
\begin{proof}
Follow the proof of Theorem 5 in Ref.~\cite{Takagi2019}, we first recall from Eq.~\eqref{Eq:Robustness} that with a $R$-theory of athermality we have (note that the output space of the channel $\mathcal{N}$ is ${\rm S}$, and $\gamma$ is the corresponding thermal state)
\begin{align}
\bar{P}_{D_{\rm max}}(\mathcal{N}) = -\log_2\sup\{p\in[0,1]\;|\;p\mathcal{N} + (1-p)\mathcal{C} = \Phi_{\gamma}\},
\end{align}
where the maximization is taken over all channels $\mathcal{C}$.
Then for every $k\in\mathbb{N}$, there exists a channel $\mathcal{C}_k$ and a value $q_k$ such that 
\begin{align}
|\bar{P}_{D_{\rm max}}(\mathcal{N}) - \log_2{\frac{1}{q_k}}|\le\frac{1}{k}\quad\&\quad q_k\mathcal{N} + (1 - q_k)\mathcal{C}_k = \Phi_{\gamma}.
\end{align}
Then we have (note that $\mathcal{N} = \frac{1}{q_k}\Phi_{\gamma} - \frac{1 - q_k}{q_k}\mathcal{C}_k$)
\begin{align}
1 - \varepsilon(\mathcal{N},\mathcal{E}_e,\mathcal{E}_d,\gamma_{\rm A})&=\frac{1}{M}\sum_{m=0}^{M-1}\bra{m}\mathcal{E}_d\circ(\mathcal{N}\otimes\Phi_{\gamma_{\rm A}})\circ\mathcal{E}_e(\proj{m})\ket{m}\nonumber\\
&\le\frac{1}{Mq_k}\sum_{m=0}^{M-1}\bra{m}\mathcal{E}_d\circ\Phi_{\gamma\otimes\gamma_{\rm A}}\circ\mathcal{E}_e(\proj{m})\ket{m}\nonumber\\
&=\frac{1}{Mq_k}\sum_{m=0}^{M-1}\bra{m}\mathcal{E}_d(\gamma\otimes\gamma_{\rm A})\ket{m}\nonumber\\
&\le\frac{1}{M}\times2^{\left[\frac{1}{k} + \bar{P}_{D_{\rm max}}(\mathcal{N})\right]},
\end{align}
which is true for all $k\in\mathbb{N}$.
This means the desired upper bound.
\end{proof}

Now we are ready to prove Theorem~\ref{Result:ConverseBoundAthermality}.
\begin{proof}
For a channel $\mathcal{N}'$ satisfying $\norm{\mathcal{N}' - \mathcal{N}}_\diamond\le2\delta$ with $0<\delta<1$, we have the estimate $\norm{\mathcal{N}'\otimes\Phi_{\gamma_{\rm A}} - \mathcal{N}\otimes\Phi_{\gamma_{\rm A}}}_\diamond\le2\delta$ due to the data-processing inequality of the trace norm.
Suppose the combination $(\mathcal{N},\mathcal{E}_e,\mathcal{E}_d,\gamma_{\rm A})$ satisfies $\varepsilon(\mathcal{N},\mathcal{E}_e,\mathcal{E}_d,\gamma_{\rm A})\le\epsilon$ for a given $0<\epsilon<1$.
Together with Facts~\ref{Fact:0} and~\ref{Fact:1}, we have
\begin{align}
1-\epsilon\le1 - \varepsilon(\mathcal{N},\mathcal{E}_e,\mathcal{E}_d,\gamma_{\rm A}) &\le\sup_{\mathcal{E}\in\mathcal{O}_R}[1 - \varepsilon(\mathcal{N},\mathcal{E}_e,\mathcal{E},\gamma_{\rm A})]\nonumber\\
&\le \sup_{\mathcal{E}\in\mathcal{O}_R}[1 - \varepsilon(\mathcal{N}',\mathcal{E}_e,\mathcal{E},\gamma_{\rm A})] + \delta\nonumber\\
&\le\frac{1}{M}\times2^{\bar{P}_{D_{\rm max}}(\mathcal{N}')} + \delta.
\end{align}
This implies
$
\log_2M\le\bar{P}_{D_{\rm max}}(\mathcal{N}') + \log_2{\frac{1}{1 - \epsilon - \delta}}.
$
Since the argument works for all possible $M$ and $\mathcal{N}'$, this means the desired bound.
\end{proof}

\newpage
\section{Proof of Theorem~\ref{Result:NoCorrelationThreshold}}\label{App:closetoB}
In this appendix, we will show a proposition which has Theorem~\ref{Result:NoCorrelationThreshold} as a direct corollary.
Before the main proof, we first prove the following fact for the $R$-theory of entanglement (and we write $R={\rm E}$).
\begin{afact}\label{Fact}
With given input/output dimensions, $\mathcal{O}_{\rm E}^N$ is convex and compact in the topology induced by the diamond norm $\norm{\cdot}_\diamond$.
\end{afact}
\begin{proof}
By definition, $\mathcal{O}_{\rm E}^N$ is convex.
Because we only consider finite dimensional Hilbert spaces, being compact is equivalent to being bounded and closed.
Since $\norm{\Lambda - \mathcal{I}}_\diamond\le\norm{\Lambda}_\diamond + \norm{\mathcal{I}}_\diamond=2$, we learn that $\mathcal{O}_{\rm E}^N$ is bounded under the diamond norm.
Hence, it suffices to show that it is a closed set.

To prove $\mathcal{O}_{\rm E}^N$ is closed, let us suppose it was not.
Then there exists a map $\Lambda\in\overline{\mathcal{O}_{\rm E}^N}\backslash\mathcal{O}_{\rm E}^N$, where $\bar{A}$ is the closure of the set $A$.
This means there exists a sequence $\{\Lambda_k\}_{k=1}^\infty\subset\mathcal{O}_{\rm E}^N$ such that $\norm{\Lambda_k - \Lambda}_\diamond\to0$ when $k\to\infty$, and there exists an input state $\rho_0$ such that $\Lambda(\rho_0)$ is entangled.
In particular, this means
$
\norm{\Lambda_k(\rho_0) - \Lambda(\rho_0)}_1\to0
$
when $k\to\infty$; in other words, we can use the sequence $\{\Lambda_k(\rho_0)\}_{k=1}^\infty$ consisting of only separable states to approach $\Lambda(\rho_0)$ in the trace norm $\norm{\cdot}_1$.
Because the set of separable states is closed in $\norm{\cdot}_1$, we conclude that $\Lambda(\rho_0)$ is separable, which is a contradiction.
Hence, $\mathcal{O}_{\rm E}^N$ is closed in $\norm{\cdot}_\diamond$, and the proof is completed.
\end{proof}

Now we state the following result:

\begin{aproposition}
For a given pair of thermal states $(\gamma_{\rm A},\gamma_{\rm B})$, if there exists an entanglement preserving local thermalization to $(\gamma_{\rm A},\gamma_{\rm B})$, then for every $\delta>0$ there exists another entanglement preserving local thermalization $\mathcal{E}$ to $(\gamma_{\rm A},\gamma_{\rm B})$ such that
\begin{align}
\bar{P}_{\norm{\cdot}_1}(\mathcal{E})<\delta.
\end{align}
\end{aproposition}
\begin{proof}
Let $\mathcal{L}_0$ be an EPLT to $(\gamma_{\rm A},\gamma_{\rm B})$, and again let $\Phi_{\rho}:(\cdot)\mapsto\rho$ be the constant map with the output state $\rho$.
Then consider the following convex mixture
\begin{align}
\mathcal{L}(p)\coloneqq p\mathcal{L}_0 + (1-p)\Phi_{\gamma_{\rm A}}\otimes\Phi_{\gamma_{\rm B}},
\end{align}
where $p\in[0,1]$.
This map is by definition a local thermalization.
Then one can see that $\mathcal{L}(p)$ is continuous on $p$ with the diamond norm because of
\begin{align}\label{Eq:L-Estimate}
\norm{\mathcal{L}(p) - \mathcal{L}(q)}_\diamond = |p-q|\norm{\mathcal{L}_0 - \Phi_{\gamma_{\rm A}}\otimes\Phi_{\gamma_{\rm B}}}_\diamond,
\end{align}
where $\norm{\mathcal{L}_0 - \Phi_{\gamma_{\rm A}}\otimes\Phi_{\gamma_{\rm B}}}_\diamond$ is a finite positive constant independent of $p$.
Because $[0,1]$ is compact, we learn that $\mathcal{L}([0,1])$ is also compact.
Also, $\mathcal{L}([0,1])$ is by definition convex.
This means $\mathcal{L}([0,1])\cap\mathcal{O}_{\rm E}^N$ is convex and compact since $\mathcal{O}_{\rm E}^N$ is convex and compact (Fact.~\ref{Fact}).
Now we note that Eq.~\eqref{Eq:L-Estimate} also means $\mathcal{L}^{-1}$ exists and is continuous on $\mathcal{L}([0,1])$.
We therefore conclude that $\mathcal{L}^{-1}\left(\mathcal{L}([0,1])\cap\mathcal{O}_{\rm E}^N\right)$ is a connected closed sub-interval contained in $[0,1]$ and containing $0$.
This means there exists $p_0\in[0,1)$ such that
\begin{align}
\mathcal{L}^{-1}\left(\mathcal{L}([0,1])\cap\mathcal{O}_{\rm E}^N\right) = [0,p_0].
\end{align}
Note that $p_0<1$ because $\mathcal{L}_0$ is not in $\mathcal{O}_{\rm E}^N$.
Now we write 
\begin{align}
\inf_{\Lambda\in\mathcal{O}_{\rm E}^N}\norm{\mathcal{L}(p) - \Lambda}_\diamond\le\norm{\mathcal{L}(p) - \mathcal{L}(p_0)}_\diamond=|p-p_0|\norm{\mathcal{L}_0 - \Phi_{\gamma_{\rm A}}\otimes\Phi_{\gamma_{\rm B}}}_\diamond.
\end{align}
For a given $\delta>0$, by choosing 
\begin{align}
p_0<p< p_0 + \frac{\delta}{\norm{\mathcal{L}_0 - \Phi_{\gamma_{\rm A}}\otimes\Phi_{\gamma_{\rm B}}}_\diamond},
\end{align}
the corresponding $\mathcal{L}(p)$ will be an EPLT to $(\gamma_{\rm A},\gamma_{\rm B})$ due to the fact that this channel will not be in $\mathcal{O}_{\rm E}^N$, and satisfies the desired property 
\begin{align}
\bar{P}_{\norm{\cdot}_1}[\mathcal{L}(p)] = \inf_{\Lambda\in\mathcal{O}_{\rm E}^N}\norm{\mathcal{L}(p) - \Lambda}_\diamond<\delta.
\end{align}
\end{proof}
In Ref.~\cite{Hsieh2019} it was shown that bipartite EPLTs exist for every positive local temperature and finite-energy local Hamiltonian (i.e. for every pair of full-rank local thermal states).
Hence, we directly conclude that:
\begin{acorollary}
For every full-rank $\gamma_{\rm A},\gamma_{\rm B}$ and every $\delta>0$, there exists an entanglement preserving local thermalization $\mE$ to $(\gamma_{\rm A},\gamma_{\rm B})$ such that $\bar{P}_{\norm{\cdot}_1}(\mathcal{E})<\delta$.
\end{acorollary}

\section{Alternative Entanglement Preserving Local Thermalization}\label{App:AlternativeEPLT}
In this section, we provide a new family of EPLTs, which can be further proved to admit arbitrarily small entanglement preservability and preservation of free entanglement simultaneously at the finite temperatures (Theorem~\ref{Result:LocalTher}).

We construct this new family of EPLTs in the bipartite system ${\rm AB}$ with equal finite local dimensions indicated as $d$.
Given a positive value $\delta_i\in[0,1]$ with integer $i\in [0,d-2]$, we define the following map on the local system ${\rm X}$:
\begin{align}\label{Eq:element}
\wt{\mathcal{E}}_{\delta_i^{\rm X}}(\cdot) = (1-\delta_i^{\rm X})\tproj{i}(\cdot)\tproj{i} + \delta_i^{\rm X}\tket{i+1}\tbra{i}(\cdot)\tket{i}\tbra{i+1} +\sum_{j\neq i} \tproj{j}(\cdot)\tproj{j},
\end{align}
where we introduced the notation $\tket{n}\coloneqq\ket{d-1-n}$ and $\tE{n}^{\rm X}\coloneqq E_{d-1-n}^{\rm X}$, and the local Hamiltonians are given by $H_{\rm X} = \sum_{i=0}^{d-1} E^{\rm X}_i\proj{i}$ for ${\rm X=A,B}$.
Now we define the following family of channels (dependent of $\delta_i^{\rm X}$) acting on a local system: 
\begin{align}\label{Eq:LocalMap}
\wt{\mathcal{E}}_{\rm X}(\cdot)\coloneqq \wt{\mathcal{E}}_{\delta_{d-2}^{\rm X}}\circ\wt{\mathcal{E}}_{\delta_{d-3}^{\rm X}}\circ ...\circ\wt{\mathcal{E}}_{\delta_2^{\rm X}}\circ\wt{\mathcal{E}}_{\delta_1^{\rm X}}\circ\wt{\mathcal{E}}_{\delta_0^{\rm X}}(\cdot).
\end{align}
In Appendix~\ref{App} we prove that $\wt{\mathcal{E}}_{\rm X}$ induces a local thermalization for an appropriate choice of $\delta_i^{\rm X}$. More precisely, with the $(U\otimes U^*)$-twirling operation $\mathcal{T}$ defined in Eq.~\eqref{Eq:Twirling} we have:

\begin{alemma}\label{Prop}
For every pair $(\gamma_{\rm A},\gamma_{\rm B})$ there exists a unique vector $\{\delta_i^{\rm A};\delta_i^{\rm B}\}_{i=0}^{d-2}$ such that $(\wt{\mathcal{E}}_{\rm A}\otimes\wt{\mathcal{E}}_{\rm B})\circ\mathcal{T}$ is a local thermalization to $(\gamma_{\rm A},\gamma_{\rm B})$.
\end{alemma} 
We remark that the proof of the above lemma is constructive, hence $\wt{\mathcal{E}}_{\rm X}$ is explicitly known [Eq.~\eqref{Eq:GeneralFormula-delta}]. 
For a given pair of single party thermal states $(\gamma_{\rm A},\gamma_{\rm B})$, we then consider the following map:
\begin{align}\label{Eq:Small-epsilon-EPLT}
\wt{\mathcal{E}}^{\ep}_{(\gamma_{\rm A}, \gamma_{\rm B})}(\cdot)\coloneqq(1-\ep)(\wt{\mathcal{E}}_{\rm A}\otimes\wt{\mathcal{E}}_{\rm B})\circ\mathcal{T}(\cdot) + \ep\,\mathcal{T}(\cdot),
\end{align}
where $\ep\in[0,1]$ is a probability parameter whose value will be determined later.
By Lemma~\ref{Prop}, we let $(\wt{\mathcal{E}}_{\rm A}\otimes\wt{\mathcal{E}}_{\rm B})\circ\mathcal{T}$ locally thermalize the system ${\rm X}$ to the following state for ${\rm X=A,B}$~\cite{Hsieh2019}:
\begin{align}\label{Eq:eta}
\eta_{\rm X}^{\ep}\coloneqq\gamma_{\rm X} + \frac{\ep}{1 - \ep}\left(\gamma_{\rm X} - \frac{\id_{\rm X}}{d}\right).
\end{align}
One can then use exactly the same proof of Theorem 2 in Ref.~\cite{Hsieh2019} to show that $\wt{\mathcal{E}}^{\ep}_{(\gamma_{\rm A}, \gamma_{\rm B})}$ is a local thermalization to $(\gamma_{\rm A}, \gamma_{\rm B})$ when 
\begin{align}\label{Eq:range}
0\le\ep\le\ep_*\coloneqq dp_{\rm min},
\end{align}
where $p_{\rm min}$ is the smallest eigenvalue among $\gamma_{\rm A}$ and $\gamma_{\rm B}$.
Finally, a direct computation of fully entangled fraction defined in Eq.~\eqref{Eq:FEF} shows
\begin{align}
\mathcal{F}_{\rm max}[\wt{\mathcal{E}}^{\ep}_{(\gamma_{\rm A}, \gamma_{\rm B})}(\rho)]\ge(1-\ep)\bra{\Psi_d^+}(\wt{\mathcal{E}}_{\rm A}\otimes\wt{\mathcal{E}}_{\rm B})[\mathcal{T}(\rho)]\ket{\Psi_d^+} + \ep\bra{\Psi_d^+}\rho\ket{\Psi_d^+}.
\end{align}
Since $\mathcal{F}_{\rm max}(\rho)>\frac{1}{d}$ implies $\rho$ is entangled~\cite{Ent-RMP}, we conclude:
\begin{atheorem}
$\wt{\mathcal{E}}^{\ep_*}_{(\gamma_{\rm A}, \gamma_{\rm B})}$ is an EPLT when $p_{\rm min}>\frac{1}{d^2}$.
\end{atheorem}
This shows Eq.~\eqref{Eq:Small-epsilon-EPLT} admits EPLTs when we select the highest $\ep$ value.
It turns out that Eq.~\eqref{Eq:Small-epsilon-EPLT} can achieve EPLTs even with arbitrarily small $\ep$ value. 
We will use this property to prove the main result in Theorem~\ref{Result:LocalTher}.

\subsection{Proof of Lemma~\ref{Prop}}\label{App}
Because we will apply mathematical induction several times in the proof, it is convenient for us to adapt the following inverse energy representation.
Let $\{\ket{n}\}_{n=0}^{d-1}$ be the energy basis for the given local system Hamiltonian, and we assume the corresponding energies $E_n$ satisfies $0\le E_0\le E_1\le...\le E_{d-1}$.
Define $\tket{n}\coloneqq\ket{d-1-n}$ and $\tE{n}\coloneqq E_{d-1-n}$, which means now the ground state is $\tket{d-1}$, and the corresponding energy is $\tE{d-1}$.
In particular, we have the hierarchy $\tE{0}\ge \tE{1}\ge...\ge \tE{d-1}\ge0$.
In what follows, we also adapt the notations $\Delta^{\rm X}_{d-2}\coloneqq\{\delta_i^{\rm X}\}_{i=0}^{d-2}$ and $\Delta^{\rm AB}_{d-2}\coloneqq\{\delta_i^{\rm A};\delta_i^{\rm B}\}_{i=0}^{d-2}$, which are regarded as vectors in $[0,1]^{(d-1)}$ and $[0,1]^{2(d-1)}$, respectively.
In this line, we further define $\wt{\mathcal{E}}_{\Delta^{\rm AB}_{d-2}} = \wt{\mathcal{E}}_{\rm A}\otimes\wt{\mathcal{E}}_{\rm B}$ and $\wt{\mathcal{E}}_{\Delta^{\rm X}_{d-2}} = \wt{\mathcal{E}}_{\rm X}$, where $\wt{\mathcal{E}}_{\rm X}$ is induced by $\{\delta_i^{\rm X}\}_{i=0}^{d-2}$ via Eq.~\eqref{Eq:LocalMap}.

In this appendix we use AB to emphasize the bipartition, and we always consider equal finite local dimensions indicated as $d$; that is, the global system can be written as $\bip$.
Now we prove the following result, which has Lemma~\ref{Prop} as a direct corollary: 
\begin{alemma}\label{lemma}
Given a pair of two single party states $(\eta_{\rm A},\eta_{\rm B})$ of the form $\eta_{\rm X} = \sum_{n=0}^{d-1}\tqx{n}\tproj{n}$ with $0\le\tqx{0}\le\tqx{1}\le...\le\tqx{d-1}\le 1$ (${\rm X=A,B}$).
Then there exist a vector $\Delta^{\rm AB}_{d-2}$ whose components are given by
\begin{align}\label{Eq:Uniquely-determined}
\delta_n^{\rm X} = 1 - \frac{d\tqx{n}}{\Gamma_{n-1}^{\rm X}},
\end{align}
where $\Gamma_{n-1}^{\rm X}\coloneqq 1 + \sum_{i=0}^{n-1} \prod_{j=i}^{n-1} \delta_j$ if $n>0$ and $\Gamma_{-1}^{\rm X}\coloneqq 1$,
such that for all $\rho$ we have
\begin{align}
{\rm tr}_{\rm B}\left[\left(\wt{\mE}_{\Delta^{\rm AB}_{d-2}}\circ\mathcal{T}\right)(\rho)\right]=\eta_{\rm A}\quad;\quad{\rm tr}_{\rm A}\left[\left(\wt{\mE}_{\Delta^{\rm AB}_{d-2}}\circ\mathcal{T}\right)(\rho)\right]=\eta_{\rm B}.
\end{align}
\end{alemma}

As a remark, we note that $\Delta^{\rm AB}_{d-2}$ is {\em uniquely} determined by $\eta_{\rm X}$ due to Eq.~\eqref{Eq:Uniquely-determined}.

\begin{proof}
({\em Proof of Lemma~\ref{lemma}.}) 

Recall that $\mathcal{T}(\rho) = \rIso (p)$ for some $p$ value [Eq.~\eqref{Eq:rIso}].
We first prove the case when $p=0$.
By using the property of isotropic state, we can prove the result for arbitrary $p$ value.
Let us start with the following fact: 
\begin{afact}\label{Fact:Formula}
For the local system ${\rm X}$, we have
\begin{align}
\wt{\mathcal{E}}_{\Delta_{d-2}^{\rm X}}\left(\frac{\id}{d}\right)= \sum_{i=0}^{d-2}\frac{\Gamma_{i-1}^{\rm X}}{d}(1-\delta_i^{\rm X})\tproj{i} + \frac{\Gamma_{d-2}^{\rm X}}{d}\tproj{d-1},
\end{align}
where $\Gamma_i^{\rm X}\coloneqq 1 + \sum_{n=0}^i \prod_{j=n}^i \delta_j^{\rm X}$ and we define $\Gamma_{-1}^{\rm X}\coloneqq1$.
\end{afact}
\begin{proof}
Let us use mathematical induction to prove the following formula for all $n\in\mathbb{Z}_{d-2}$:
\begin{align}
\wt{\mathcal{E}}_{\Delta_n^{\rm X}}\left(\frac{\id}{d}\right) =\sum_{i=0}^{n} \frac{\Gamma_{i-1}^{\rm X}}{d}(1-\delta_i^{\rm X})\tproj{i} + \frac{\Gamma_{n}^{\rm X}}{d}\tproj{n+1} + \frac{1}{d}\sum_{j=n+2}^{d-1}\tproj{j}.\quad
\end{align}
First, direct computation can prove the case for $n=0,1$.
Now, let us assume the correctness of the above formula for $n$ in $\mathbb{Z}_{d-3}$ and compute the result for $n+1$:
\begin{align}
\wt{\mathcal{E}}_{\Delta_{n+1}^{\rm X}}\left(\frac{\id}{d}\right) &= \wt{\mathcal{E}}_{\delta_{n+1}^{\rm X}}\circ\wt{\mathcal{E}}_{\Delta_n^{\rm X}}\left(\frac{\id}{d}\right)\nonumber\\
&=\sum_{i=0}^{n+1} \frac{\Gamma_{i-1}^{\rm X}}{d}(1-\delta_i^{\rm X})\tproj{i} + \frac{1}{d}\left[\Gamma_n^{\rm X}\delta_{n+1}^{\rm X} + 1\right]\tproj{n+2} + \frac{1}{d}\sum_{j=n+3}^{d-1}\tproj{j}.
\end{align}\label{Eq:AppRecursion}
The result follows by observing the following recursion relation:
\begin{align}\label{Eq:GammaRecursionRelation}
\Gamma_i^{\rm X} = \Gamma_{i-1}^{\rm X}\times\delta_i^{\rm X} +1.
\end{align}
Hence, by mathematical induction, the formula works for all $n\in\mathbb{Z}_{d-2}$. 
Finally, one can apply $\wt{\mathcal{E}}_{\delta_{d-2}^{\rm X}}$ on $\wt{\mathcal{E}}_{\Delta_{d-3}^{\rm X}}\left(\frac{\id}{d}\right)$ and obtain the desired result.
\end{proof}

From the above fact, we know the final state is diagonal in the predefined energy eigenbasis.
Now we need to make sure this final state, $\wt{\mathcal{E}}_{\Delta_{d-2}^{\rm X}}\left(\frac{\id}{d}\right)$, can {\em always} be the desired state $\eta_{\rm X}$.
Intuitively, this may be achievable by tuning $\Delta_{d-2}^{\rm X}$.
Formally, we prove the following result:

\begin{afact}\label{Fact:Well-defined}
Given a single party state in the local system ${\rm X}$ of the form $\eta_{\rm X} = \sum_{n=0}^{d-1}\tqx{n}\tproj{n}$ with $0\le\tqx{0}\le\tqx{1}\le...\le\tqx{d-1}\le 1$.
Then there exists a vector $\Delta_{d-2}^{\rm X}$ such that $\tilde{\mathcal{E}}_{\Delta_{d-2}^{\rm X}}\left(\frac{\id}{d}\right) = \eta_{\rm X}$.
\end{afact}
\begin{proof}
For $\tilde{\mathcal{E}}_{\Delta_{d-2}^{\rm X}}\left(\frac{\id}{d}\right)$ to be able to describe $\eta_{\rm X}$ with some vector $\Delta_{d-2}^{\rm X}$, Fact~\ref{Fact:Formula} tells us it is sufficient to have 
\begin{align}
\max_{\delta_n^{\rm X}}\frac{\Gamma_{n-1}^{\rm X}(1-\delta_n^{\rm X})}{d} \ge \tqx{n}
\end{align}
for all $n\in\mathbb{Z}_{d-1}$ (i.e. $0\le n\le d-2$).
Note that we do not need to deal with the state $\tket{d-1}$ because normalization will do the job.
Now we observe that for any given number $n\in\mathbb{Z}_{d-1}$, we have
$
1 = \sum_{i=0}^{d-1}\tqx{i}\ge (d-n)\tqx{n} +A_n^{\rm X},
$
where $A_{n}^{\rm X}\coloneqq\sum_{i=0}^{n-1}\tqx{i}$ for $n\neq 0$ and $A_0^{\rm X}\coloneqq 0$.
This means
\begin{align}
\tqx{n}\le\frac{1 - A_{n}^{\rm X}}{d-n}.
\end{align}
Together with $\max_{\delta_n^{\rm X}}\frac{\Gamma_{n-1}^{\rm X}(1-\delta_n^{\rm X})}{d} = \frac{\Gamma_{n-1}^{\rm X}}{d}$, we will use mathematical induction to prove the following statement: 
{\em Given a number $n\in\mathbb{Z}_{d-1}$, then for all $0\le i\le n$ there exists a $\delta_i^{\rm X}$ achieving $\tqx{i} = \frac{\Gamma_{i-1}^{\rm X}(1 - \delta_i^{\rm X})}{d}$, and we have
\begin{align}
\frac{\Gamma_{n-1}^{\rm X}}{d}\ge\frac{1 - A_n^{\rm X}}{d-n}.
\end{align}
}

To begin with, we first notice that it is true for $n=0,1$.
To prove this, one can see that when $n=0$, both sides are equal to $\frac{1}{d}$.
This means one can always choose $\tqx{0} = \frac{\Gamma_{-1}^{\rm X}(1-\delta_0^{\rm X})}{d}$ by choosing a proper $\delta_0^{\rm X}$.
This proves the statement for $n=0$.

When $n=1$, recall that we have
$
\frac{\Gamma_{0}^{\rm X}}{d} = \frac{\delta_0^{\rm X} + 1}{d}.
$
Now we note that because the formula works for $n=0$, which means we can choose $\delta_0^{\rm X}$ such that $\tqx{0} = \frac{\Gamma_{-1}^{\rm X}(1-\delta_0^{\rm X})}{d} = \frac{1 - \delta_0^{\rm X}}{d}$.
Together with the fact $A_1^{\rm X} = \tqx{0}$, we have
\begin{align}
\frac{1-A_1^{\rm X}}{d-1} = \frac{d - 1 + \delta_0^{\rm X}}{(d-1)d}\le\frac{d - 1 + (d - 1)\delta_0^{\rm X}}{(d-1)d}= \frac{\Gamma_{0}^{\rm X}}{d}.
\end{align}
This means one is able to choose a proper $\delta_1^{\rm X}$ to achieve $\tqx{1} = \frac{\Gamma_0^{\rm X}(1-\delta_1^{\rm X})}{d}$.
This completes the proof of $n=1$.

Now we assume the correctness of the statement for a given $n\le d-3$, and then we try to prove the case for $n+1$.
To do so, we note that the recursion relation Eq.~\eqref{Eq:GammaRecursionRelation} implies
$
\Gamma_i^{\rm X} = \Gamma_{i-1}^{\rm X}(\delta_i^{\rm X} - 1) +\Gamma_{i-1}^{\rm X} + 1.
$
Due to the correctness of the statement, we are allowed to choose $\Gamma_{i-1}^{\rm X}(\delta_i^{\rm X} - 1) = -d\tqx{i}$ for all $0\le i \le n$.
This means
\begin{align}
\Gamma_i^{\rm X} = -d\tqx{i} + \Gamma_{i-1}^{\rm X} + 1
\end{align}
for all $0\le i \le n$.
Using this new recursion relation, one can use mathematical induction to obtain (recall that $\Gamma_{-1}^{\rm X}\coloneqq1$)
\begin{align}
\frac{\Gamma_i^{\rm X}}{d} = \frac{2 + i}{d} - \sum_{j=0}^i \tqx{j} = \frac{2 + i}{d} - A_{i+1}^{\rm X}.
\end{align}
This means
\begin{align}
\frac{\Gamma_{n}^{\rm X}}{d} -  \frac{1-A_{n+1}^{\rm X}}{d-(n+1)} =  \frac{2 + n}{d} - \frac{1}{d-n-1} + \frac{ n+2- d}{d-n-1}A_{n+1}^{\rm X}.
\end{align}
Now we recall the hierarchy $0\le\tqx{0}\le\tqx{1}\le...\le\tqx{d-1}$.
This means the following fact:
\begin{align}\label{SubFact}
A_{n+1}^{\rm X} = \sum_{i=0}^n \tqx{i}\le\frac{n+1}{d}.
\end{align}
One can prove the above inequality by contradiction.
Assume the converse, which means $\sum_{i=0}^n \tqx{i}>\frac{n+1}{d}$ and $\sum_{j=n+1}^{d-1} \tqx{j}<1 - \frac{n+1}{d} = \frac{d-n-1}{d}$.
Then there exists $0\le i\le n$ and $n+1\le j\le d-1$ such that $\tqx{i} > \frac{1}{d} > \tqx{j}$, which is a contradiction because $\tqx{j}\ge \tqx{i}\ge0$.

Because $n\le d-3$, we have $\frac{n+2-d}{d-n-1}<0$, which is the pre-factor of the term $A_{n+1}^{\rm X}$.
This means
\begin{align}
\frac{\Gamma_{n}^{\rm X}}{d} -  \frac{1-A_{n+1}^{\rm X}}{d-n-1} \ge \frac{n + 2}{d} - \frac{1}{d-n-1} + \frac{n+2  - d}{d-n-1}\times\frac{n+1}{d}= 0.
\end{align}
This proves the formula for $n+1$, which consequently implies it is always possible to choose a $\delta_{n+1}^{\rm X}$ such that $\tqx{n+1} = \frac{\Gamma_{n}^{\rm X}(1-\delta_{n+1}^{\rm X})}{d}$.
This completes the proof of the statement by using mathematical induction.

Since the statement implies it is always possible to choose a $\Delta_{d-2}^{\rm X}$ to fit $\{\tqx{i}\}_{i=0}^{d-2}$ and since the normalization condition will fix the value for the component of $\tproj{d-1}$, the proof is completed.
\end{proof}
Using Fact~\ref{Fact:Formula} and Fact~\ref{Fact:Well-defined}, we learn the following result (recall that $d$ is the common finite local dimension for both subsystems ${\rm A}$ and ${\rm B}$):
\begin{acorollary}
Given a pair of single party states $(\eta_{\rm A},\eta_{\rm B})$ of the form $\eta_{\rm X} = \sum_{n=0}^{d-1}\tqx{n}\tproj{n}$ with $0\le\tqx{0}\le\tqx{1}\le...\le\tqx{d-1}\le 1$ ({\rm X=A,B}), then there exists a vector $\Delta_{d-2}^{\rm AB}$ such that the channel $\wt{\mE}_{\Delta^{\rm AB}_{d-2}}$ achieves
\begin{align}
\wt{\mE}_{\Delta^{\rm AB}_{d-2}}\left(\frac{\id}{d^2}\right) = \eta_{\rm A}\otimes\eta_{\rm B}.
\end{align}
\end{acorollary}
This describes the behavior when the input is a maximally mixed state.
Now, it remains to show the {\em same} output can occur when the input state is an isotropic state given by Eq.~\eqref{Eq:rIso}.
This can be done by the following relation between partial trace and local channel when acting on separable states:
\begin{afact}\label{Fact:Commutative}
Given a separable state $\rho=\sum_i f_i \rho^{\rm A}_i \otimes\rho^{\rm B}_i$ and two single party channels $\mathcal{E}_{\rm X}$ acting on the ${\rm X}$ system.
Then
\begin{align}
{\rm tr}_{\rm B}\left[(\mathcal{E}_{\rm A}\otimes\mathcal{E}_{\rm B})(\rho)\right] = \mathcal{E}_{\rm A}[{\rm tr}_{\rm B}(\rho)]\quad;\quad{\rm tr}_{\rm A}\left[(\mathcal{E}_{\rm A}\otimes\mathcal{E}_{\rm B})(\rho)\right] = \mathcal{E}_{\rm B}[{\rm tr}_{\rm A}(\rho)].
\end{align}
\end{afact}
\begin{proof}
Due to separability, one have
\begin{align}
{\rm tr}_{\rm B}\left[(\mathcal{E}_{\rm A}\otimes\mathcal{E}_{\rm B})(\rho)\right] &= {\rm tr}_{\rm B}\left[\sum_i f_i \mathcal{E}_{\rm A}(\rho^{\rm A}_i) \otimes \mathcal{E}_{\rm B}(\rho^{\rm B}_i)\right]\nonumber\\
& = \sum_i f_i \mathcal{E}_{\rm A}(\rho^{\rm A}_i) {\rm tr}\left[\mathcal{E}_{\rm B}(\rho^{\rm B}_i)\right]\nonumber\\
& = \mathcal{E}_{\rm A}\left(\sum_i f_i \rho^{\rm A}_i\right)\nonumber\\
& = \mathcal{E}_{\rm A}\left[{\rm tr}_{\rm B}(\rho)\right].
\end{align}
Similar calculation proves the other case.
\end{proof}
Since both $(\wt{\mathcal{E}}_{\Delta_{d-2}^{\rm A}}\otimes\mathcal{I}_{\rm B})$ and $(\mathcal{I}_{\rm A}\otimes\wt{\mathcal{E}}_{\Delta_{d-2}^{\rm B}})$ map an isotropic state to a separable state (this can be seen by the fact that they will map $\ket{\Psi_d^+}$ to a separable state), the above fact means
\begin{align}
{\rm tr}_{\rm A}\left\{\wt{\mE}_{\Delta_{d-2}^{\rm AB}}[\rIso(p)]\right\}= {\rm tr}_{\rm A}\left\{\right(\mathcal{I}_{\rm A}\otimes\wt{\mathcal{E}}_{\Delta_{d-2}^{\rm B}})[\rIso(p)]\} =\wt{\mathcal{E}}_{\Delta_{d-2}^{\rm B}}\left(\frac{\id}{d}\right) 
\end{align}
for all $p$.
Similar result can be shown for Bob's local system by replacing ${\rm A}$ and ${\rm B}$.
In particular, this means
\begin{align}
{\rm tr}_{\rm B}\left\{\wt{\mE}_{\Delta_{d-2}^{\rm AB}}[\rIso(p)]\right\} = {\rm tr}_{\rm B}\left\{\wt{\mE}_{\Delta_{d-2}^{\rm AB}}[\rIso(0)]\right\} = {\rm tr}_{\rm B}\left[\wt{\mE}_{\Delta_{d-2}^{\rm AB}}\left(\frac{\id}{d^2}\right)\right] = \eta_{\rm A}
\end{align}
for all $p$.
Similar argument proves
\begin{align}
{\rm tr}_{\rm A}\left\{\wt{\mE}_{\Delta_{d-2}^{\rm AB}}[\rIso(p)]\right\} = \eta_{\rm B}
\end{align}
Finally, because $\mathcal{T}(\rho)$ will be an isotropic state for {\em any} state $\rho$, the result follows.
({\em End of Proof of Lemma~\ref{lemma}}.)
\end{proof}

\subsection{Remarks}
Here we make some remarks.
First, note that Fact~\ref{Fact:Well-defined} can apply on arbitrary single party thermal state.
As another remark, we note that for a given $\eta_{\rm X}=\sum_{n=0}^{d-1}\tqx{n}\tproj{n}$, there is a {\em uniquely} determined vector $\Delta_{d-2}^{\rm X}$ which can realize it.
To find this vector $\Delta_{d-2}^{\rm X}$, one can start from $\delta_0^{\rm X}$, which is given by
\begin{align}
\delta_0^{\rm X} = 1 - d\tqx{0}.
\end{align}
After determining $\delta_0^{\rm X}$, one can determine $\delta_1^{\rm X}$, which is given by
\begin{align}
\delta_1^{\rm X} = 1 - \frac{d\tqx{1}}{\Gamma_0^{\rm X}} = 1 - \frac{d\tqx{1}}{2 - d\tqx{0}}.
\end{align}
In general, one can determine $\delta_n^{\rm X}$ by the following formula:
\begin{align}\label{Eq:GeneralFormula-delta}
\delta_n^{\rm X} = 1 - \frac{d\tqx{n}}{\Gamma_{n-1}^{\rm X}},
\end{align}
this is because after knowing $\delta_i^{\rm X}$ for $0\le i\le n-1$, one can directly compute $\Gamma_{n-1}^{\rm X}$.

\section{Proof of Theorem~\ref{Result:LocalTher}}\label{App:Proof_Result:LocalTher}
To prove Theorem~\ref{Result:LocalTher}, first we prove Eq.~\eqref{Thm:distance} in Appendix~\ref{App:LongDistance}.
As the next step in Appendix~\ref{App:Prop3} we prove a lemma, which is a preliminary result for the proof of Eq.~\eqref{Result:LocalTherEq} given in Appendix~\ref{App:ShortDistance}.

In the proof of Eq.~\eqref{Thm:distance}, we will use the EPLT candidate constructed in Ref.~\cite{Hsieh2019}, which is given by:
\begin{align}\label{Eq:TheMap}
\mE_{(\gamma_{\rm A},\gamma_{\rm B})}^{\ep}\coloneqq (1 - \ep)\Phi_{\eta_{\rm A}^{\ep}\otimes\eta_{\rm B}^{\ep}}\circ\mathcal{T} + \ep\mathcal{T},
\end{align}
where the $(U\otimes U^*)$-twirling $\mathcal{T}$ is defined in Eq.~\eqref{Eq:Twirling}, $\Phi_\rho(\cdot) = \rho$ is the constant map, and $\eta_{\rm X}^{\ep}$ is defined in Eq.~\eqref{Eq:eta}.
$\mE_{(\gamma_{\rm A},\gamma_{\rm B})}^{\ep}$ is proved to be an EPLT to $(\gamma_{\rm A},\gamma_{\rm B})$ for all full-rank $\gamma_{\rm A}$ and $\gamma_{\rm B}$~\cite{Hsieh2019}.
As a remark, we note that Eq.~\eqref{Eq:TheMap} can be interpreted as the twirling operation $\mathcal{T}$ followed by a partial thermalization channel $(\cdot)\mapsto(1 - \ep)\eta_{\rm A}^{\ep}\otimes\eta_{\rm B}^{\ep} + \ep(\cdot)$, which can be thought as a generalized depolarizing channel with finite local temperatures (captured by the thermal states $\eta_{\rm A}^{\ep}$ and $\eta_{\rm B}^{\ep}$).

\subsection{Proof of Eq.~(\ref{Thm:distance})}\label{App:LongDistance}
\begin{proof}
We compute the lower bound for the map $\mE_{(\gamma_{\rm A},\gamma_{\rm B})}^{\ep}$ defined in Eq.~\eqref{Eq:TheMap} (for a channel $\mE$ we write $\norm{\mE}_\infty\coloneqq\sup_\rho\norm{\mE(\rho)}_\infty$):
\begin{align}
\inf_{\Lambda\in\mathcal{O}_{\rm FE}^N}\norm{\mathcal{E}^{\ep}_{(\gamma_{\rm A},\gamma_{\rm B})} - \Lambda}_\diamond&\ge\inf_{\Lambda\in\mathcal{O}_{\rm FE}^N}\norm{\mathcal{E}^{\ep}_{(\gamma_{\rm A},\gamma_{\rm B})} - \Lambda}_1\nonumber\\
&=\inf_{\Lambda\in\mathcal{O}_{\rm FE}^N}\norm{\left[(1 - \ep)\Phi_{\eta_{\rm A}^{\ep}\otimes\eta_{\rm B}^{\ep}}\circ\mathcal{T} + \ep\mathcal{T}\right] - \Lambda}_1\nonumber\\
&=\inf_{\Lambda\in\mathcal{O}_{\rm FE}^N}\norm{(1-\ep)\left(\Phi_{\eta_{\rm A}^{\ep}\otimes\eta_{\rm B}^{\ep}}\circ\mathcal{T} - \Lambda\right) + \ep\left(\mathcal{T} - \Lambda\right)}_1\nonumber\\
&\ge\inf_{\Lambda\in\mathcal{O}_{\rm FE}^N}\left|(1-\ep)\norm{\Phi_{\eta_{\rm A}^{\ep}\otimes\eta_{\rm B}^{\ep}}\circ\mathcal{T} - \Lambda}_1 - \ep\norm{\mathcal{T} - \Lambda}_1\right|,\nonumber \\
&\ge\inf_{\Lambda\in\mathcal{O}_{\rm FE}^N}\left[\ep\norm{\mathcal{T} - \Lambda}_1 - (1-\ep)\norm{\Phi_{\eta_{\rm A}^{\ep}\otimes\eta_{\rm B}^{\ep}}\circ\mathcal{T} - \Lambda}_1\right],
\end{align}
where the fourth line follows from the inverse triangle inequality of the trace norm. 
Using the estimate $\norm{\Phi_{\eta_{\rm A}^{\ep}\otimes\eta_{\rm B}^{\ep}}\circ\mathcal{T} - \Lambda}_1\le\norm{\Phi_{\eta_{\rm A}^{\ep}\otimes\eta_{\rm B}^{\ep}}\circ\mathcal{T}}_1+\norm{\Lambda}_1=2$, we have
\begin{align}
\inf_{\Lambda\in\mathcal{O}_{\rm FE}^N}\norm{\mathcal{E}^{\ep}_{(\gamma_{\rm A},\gamma_{\rm B})} - \Lambda}_\diamond\ge\ep\inf_{\Lambda\in\mathcal{O}_{\rm FE}^N}\norm{\mathcal{T} - \Lambda}_1 - 2(1-\ep).
\end{align}
Now we bound $\inf_{\Lambda\in\mathcal{O}_{\rm FE}^N}\norm{\mathcal{T} - \Lambda}_1$. 
Denoting by $\rho_{\rm b}$ an arbitrary state which is not free entangled and by $\rIso$ an arbitrary isotropic state [defined in Eq.~\eqref{Eq:rIso}], we have
\begin{align}\label{Eq:App_T_lowerbound}
\inf_{\Lambda\in\mathcal{O}_{\rm FE}^N}\norm{\mathcal{T} - \Lambda}_1&\coloneqq\inf_{\Lambda\in\mathcal{O}_{\rm FE}^N}\sup_{\rho}\norm{\mathcal{T}(\rho) - \Lambda(\rho)}_1\nonumber\\
&\ge\inf_{\rho_{\rm b}}\sup_{\rIso}\norm{\rIso - \rho_{\rm b}}_1\nonumber\\
&\ge\inf_{\rho_{\rm b}}\sup_{\rIso}\left|\bra{\Psi_d^+}\rIso\ket{\Psi_d^+} - \bra{\Psi_d^+}\rho_{\rm b}\ket{\Psi_d^+}\right|\nonumber\\
&=\inf_{\rho_{\rm b}}\left|1 - \bra{\Psi_d^+}\rho_{\rm b}\ket{\Psi_d^+}\right| \nonumber\\
&=1 - \frac{1}{d},
\end{align}
where we note that $\norm{\cdot}_1\coloneqq{\rm tr}|\cdot|\ge \norm{\cdot}_\infty \coloneqq \sup_{\ket{\psi}}|\bra{\psi}(\cdot)\ket{\psi}|$ and the last equality is due to the sufficient condition $\bra{\Psi_d^+}\rho\ket{\Psi_d^+}>\frac{1}{d}$ of distillability~\cite{Horodecki1998} for a quantum state $\rho$.
Hence,
\begin{align}
\inf_{\Lambda\in\mathcal{O}_{\rm FE}^N}\norm{\mathcal{E}^{\ep}_{(\gamma_{\rm A},\gamma_{\rm B})} - \Lambda}_\diamond\ge\ep\left(1 - \frac{1}{d}\right) - 2(1-\ep).
\end{align}
The strongest bound is achieved by taking the largest $\ep$ allowed by Eq.~\eqref{Eq:range}, giving
\begin{align}
\inf_{\Lambda\in\mathcal{O}_{\rm FE}^N}\norm{\mathcal{E}^{\ep}_{(\gamma_{\rm A},\gamma_{\rm B})} - \Lambda}_\diamond\ge (3d-1) p_{\rm min} -2.
\end{align}
The proof is completed.
\end{proof}

\subsection{Preliminary for the Proof of Eq~(\ref{Result:LocalTherEq})}\label{App:Prop3}
As the first step, we show the following lemma:
\begin{alemma}\label{Prop3}
For every $\ep\in\left(0,\frac{1}{2}\right]$, there exists $\tau_{\ep}\in(0,+\infty)$ such that if $\min_{\rm X=A,B}\tau_{\rm X}>\tau_{\ep}$ then $\wt{\mathcal{E}}^{\ep}_{(\gamma_{\rm A},\gamma_{\rm B})}$ is a local thermalization to $(\gamma_{\rm A},\gamma_{\rm B})$ with $\mathcal{F}_{\rm max}[\wt{\mathcal{E}}^{\ep}_{(\gamma_{\rm A}, \gamma_{\rm B})}(\rho)] > \frac{1}{d}$ for some $\rho$.
\end{alemma}
\begin{proof}
First of all, we try to argue that $\bra{\Psi_d^+}\wt{\mE}_{\Delta^{\rm AB}_{d-2}}(\proj{\Psi_d^+})\ket{\Psi_d^+}$ can be arbitrarily close to $\frac{1}{d}$, even though $\wt{\mE}_{\Delta^{\rm AB}_{d-2}}(\proj{\Psi_d^+})$ is a separable state.
Write $\max_{\rm X=A,B}\norm{\gamma_{\rm X}- \frac{\id}{d}}_\infty<\delta_0$ for a given small positive value $\delta_0$, which implies $\max_{\rm X=A,B}\norm{\eta_{\rm X}^{\ep}-\frac{\id}{d}}_\infty<2\delta_0$ for $\ep\in\left(0,\frac{1}{2}\right]$.
By Lemma~\ref{Prop} and Appendix~\ref{App}, since the vector $\Delta^{\rm AB}_{d-2}$ is uniquely determined by $(\eta_{\rm A}^{\ep},\eta_{\rm B}^{\ep})$, $\max_{\rm X=A,B}\norm{\eta_{\rm X}^{\ep}-\frac{\id}{d}}_\infty<2\delta_0$ will imply the existence of a small value $\delta_1$ such that $\norm{\Delta^{\rm AB}_{d-2}}<\delta_1$ (one can see this by the structure of $\wt{\mE}_{\Delta^{\rm X}_{d-2}}$ given in Appendix~\ref{App}).
Hence, the continuity implies the existence of a small positive value $\delta = \delta(\delta_0)$ such that $\bra{\Psi_d^+}\wt{\mE}_{\Delta^{\rm AB}_{d-2}}(\proj{\Psi_d^+})\ket{\Psi_d^+} \ge \frac{1}{d}-\delta$, and $\delta$ can be as small as we want by choosing a proper $\delta_0$.

Now we note the following property of normalized temperature: $\gamma_{\rm X}\to\frac{\id}{d}$ if $\tau_{\rm X}\to\infty$; in other words, for a given value $\Delta$, there exists a normalized temperature threshold $\tau_\Delta$ such that $\max_{\rm X=A,B}\norm{\gamma_{\rm X}-\frac{\id}{d}}_\infty<\Delta$ if $\min_{\rm X=A,B}\tau_{\rm X}>\tau_\Delta$. 

Together with this property of normalized temperature, for a given $k\in\mathbb{N}$, there exists a $\Delta_k$ such that $\bra{\Psi_d^+}\wt{\mE}_{\Delta^{\rm AB}_{d-2}}(\proj{\Psi_d^+})\ket{\Psi_d^+} \ge \frac{1}{d}-\frac{1}{k}$ if $\min_{\rm X=A,B}\tau_{\rm X}>\Delta_k$.
Finally, for a given $\ep\in\left[0,dp_{\rm min}\right]$ we have the following estimate if $\min_{\rm X=A,B}\tau_{\rm X}>\Delta_k$:
\begin{align}
\mathcal{F}_{\rm max}[\wt{\mathcal{E}}^{\ep}_{(\gamma_{\rm A}, \gamma_{\rm B})}(\proj{\Psi_d^+})]\ge(1-\ep)\left(\frac{1}{d} - \frac{1}{k}\right) + \ep,
\end{align}
which is strictly larger than $\frac{1}{d}$ if $\ep>\frac{d}{k(d-1)+d}$.
Then for the given $\ep$, there exits a $k_{\ep}\coloneqq1+\left[\frac{(1-\ep)d}{\ep(d-1)}\right]$, where $[\cdot]$ is the Gauss' notion, such that $\mathcal{F}_{\rm max}[\wt{\mathcal{E}}^{\ep}_{(\gamma_{\rm A}, \gamma_{\rm B})}(\proj{\Psi_d^+})]>\frac{1}{d}$, thereby being free entangled, if $\min_{\rm X=A,B}\tau_{\rm X}>\tau_{\ep}\coloneqq\Delta_{k_{\ep}}$.
\end{proof}

\subsection{Proof of Eq~(\ref{Result:LocalTherEq})}\label{App:ShortDistance}
\begin{proof}
We show that $\wt{\mathcal{E}}^{\ep}_{(\gamma_{\rm A}, \gamma_{\rm B})}$ given in Eq.~\eqref{Eq:Small-epsilon-EPLT} can be arbitrarily close to the set $\mathcal{O}_{\rm E}^N$ (here E denotes entanglement) while preserving free entanglement for certain entangled input states. 
For any given $\delta>0$, there exists an \mbox{$\ep\in(0,1]$} small enough such that $\ep\times \norm{\mathcal{T} - \wt{\mathcal{E}}_{\rm A}\otimes\wt{\mathcal{E}}_{\rm B}\circ\mathcal{T}}_\diamond<\delta$. 
Lemma~\ref{Prop3} implies there exists $\tau_{\ep}\in(0,+\infty)$ such that for every pair $(\gamma_{\rm A},\gamma_{\rm B})$ with $\min_{\rm X=A,B}\tau_{\rm X}>\tau_{\ep}$, $\wt{\mathcal{E}}^{\ep}_{(\gamma_{\rm A},\gamma_{\rm B})}$ is an EPLT to $(\gamma_{\rm A},\gamma_{\rm B})$ that can preserve free entanglement and achieves
\begin{align}
\bar{P}_{\norm{\cdot}_1}\left[\wt{\mathcal{E}}^{\ep}_{(\gamma_{\rm A},\gamma_{\rm B})}\right]&\coloneqq\inf_{\Lambda_{\rm S}\in\mathcal{O}_{\rm E}^N}\sup_{{\rm A};\rho_{\rm SA}}\norm{\left[\wt{\mathcal{E}}^{\ep}_{(\gamma_{\rm A},\gamma_{\rm B})}\otimes\mathcal{I}_{\rm A}\right](\rho_{\rm SA}) - (\Lambda_{\rm S}\otimes\mathcal{I}_{\rm A})(\rho_{\rm SA})}_1\nonumber\\
&=\inf_{\Lambda_{\rm S}\in\mathcal{O}_{\rm E}^N}\norm{\wt{\mathcal{E}}^{\ep}_{(\gamma_{\rm A},\gamma_{\rm B})} - \Lambda_{\rm S}}_\diamond\nonumber\\
&\le\norm{\wt{\mathcal{E}}^{\ep}_{(\gamma_{\rm A},\gamma_{\rm B})} - \wt{\mathcal{E}}_{\rm A}\otimes\wt{\mathcal{E}}_{\rm B}\circ\mathcal{T}}_\diamond\nonumber\\
& = \ep\norm{\mathcal{T} - \wt{\mathcal{E}}_{\rm A}\otimes\wt{\mathcal{E}}_{\rm B}\circ\mathcal{T}}_\diamond\nonumber\\
&<\delta,
\end{align}
where $\sup_{{\rm A};\rho_{\rm SA}}$ is optimizing over all the ancillary system ${\rm A}$ and states $\rho_{\rm SA}$ on the system ${\rm SA}$.
By redefining $\tau_{\ep}$ to be the $\tau_\delta$ given in the statement of the theorem, the proof is completed.
\end{proof}


\newpage
\begin{thebibliography}{99}

\bibitem{Ent-RMP} R, Horodecki, P, Horodecki, M, Horodecki, and K, Horodecki, {\em Quantum entanglement}, \href{https://doi.org/10.1103/RevModPhys.81.865}{Rev. Mod. Phys. {\bf81}, 865 (2009).}

\bibitem{Peres1996} A. Peres, {\em Separability criterion for density matrices}, \href{https://doi.org/10.1103/PhysRevLett.77.1413}{Phys. Rev. Lett. {\bf 77}, 1413 (1996).}
	
\bibitem{Horodecki1996} M. Horodecki, P. Horodecki, and R. Horodecki, {\em Separability of mixed states: Necessary and sufficient conditions}, \href{https://doi.org/10.1016/S0375-9601(96)00706-2}{Phys. Lett. A {\bf 223}, 1 (1996).}

\bibitem{Bell} J. S. Bell, {\em On the Einstein Podolsky Rosen paradox}, \href{https://doi.org/10.1103/PhysicsPhysiqueFizika.1.195}{Physics Physique Fizika {\bf1}, 195 (1964).}

\bibitem{Bell-RMP} N. Brunner, D. Cavalcanti, S. Pironio, V. Scarani, and S. Wehner, {\em Bell nonlocality}, \href{https://doi.org/10.1103/RevModPhys.86.419}{Rev. Mod. Phys. {\bf86}, 419 (2014).}

\bibitem{Wiseman2007} H. M. Wiseman, S. J. Jones, and A. C. Doherty, {\em Steering, entanglement, nonlocality, and the Einstein-Podolsky-Rosen paradox}, \href{https://doi.org/10.1103/PhysRevLett.98.140402}{Phys. Rev. Lett. {\bf98}, 140402 (2007).}

\bibitem{Jones2007} S. J. Jones, H. M. Wiseman, and A. C. Doherty, {\em Entanglement, Einstein-Podolsky-Rosen correlations, Bell nonlocality, and steering}, \href{https://doi.org/10.1103/PhysRevA.76.052116}{Phys. Rev. A {\bf76}, 052116 (2007).}

\bibitem{steering-review} D. Cavalcanti and P. Skrzypczyk, {\em Quantum steering: A review with focus on semidefinite programming}, \href{https://doi.org/10.1088/1361-6633/80/2/024001}{Rep. Prog. Phys. {\bf80}, 024001 (2017).}

\bibitem{Steering-RMP} R. Uola, A. C. S. Costa, H. C. Nguyen, and O. G$\ddot{\rm u}$hne, {\em Quantum steering}, \href{https://doi.org/10.1103/RevModPhys.92.015001}{Rev. Mod. Phys. {\bf92}, 015001 (2020).}

\bibitem{Bennett1993} C. H. Bennett, G. Brassard, C. Cr\'epeau, R. Jozsa, A. Peres, and W. K. Wootters, {\em Teleporting an unknown quantum state via dual classical and Einstein-Podolsky-Rosen channels}, \href{https://doi.org/10.1103/PhysRevLett.70.1895}{Phys. Rev. Lett. {\bf70}, 1895 (1993).}

\bibitem{Horodecki1999-2} M. Horodecki, P. Horodecki, and R. Horodecki, {\em General teleportation channel, singlet fraction, and quasidistillation}, \href{https://doi.org/10.1103/PhysRevA.60.1888}{Phys. Rev. A {\bf60}, 1888 (1999).}

\bibitem{Jennings2010} D. Jennings and T. Rudolph, {\em Entanglement and the thermodynamic arrow of time
}, \href{https://doi.org/10.1103/PhysRevE.81.061130}{Phys. Rev. E {\bf 81}, 061130 (2010).}

\bibitem{QCI-book} M. A. Nielsen and I. L. Chuang, {\em Quantum Computation and Quantum Information}, (Cambridge University Press, Cambridge, UK, 2000).

\bibitem{Vedral1997} V. Vedral, M. B. Plenio, M. A. Rippin, and P. L. Knight, {\em Quantifying entanglement}, \href{https://doi.org/10.1103/PhysRevLett.78.2275}{Phys. Rev. Lett. {\bf 78}, 2275 (1997).}

\bibitem{Coherence-RMP} A. Streltsov, G. Adesso, and M. B. Plenio, {\em Colloquium: Quantum coherence as a resource}, \href{https://doi.org/10.1103/RevModPhys.89.041003}{Rev. Mod. Phys. 89, 041003 (2017).}

\bibitem{Baumgratz2014} T. Baumgratz, M. Cramer, and M. B. Plenio, {\em Quantifying coherence}, \href{https://doi.org/10.1103/PhysRevLett.113.140401}{Phys. Rev. Lett. {\bf113}, 140401 (2014).}

\bibitem{Wolfe2019} E. Wolfe, D. Schmid, A. B. Sainz, R. Kunjwal, and R. W. Spekkens, {\em Quantifying Bell: The resource theory of nonclassicality of common-cause boxes}, \href{https://arxiv.org/abs/1903.06311}{arXiv:1903.06311.}

\bibitem{Skrzypczyk2014} P. Skrzypczyk, M. Navascu\'es, and D. Cavalcanti, {\em Quantifying Einstein-Podolsky-Rosen steering}, \href{https://doi.org/10.1103/PhysRevLett.112.180404}{Phys. Rev. Lett. {\bf112}, 180404 (2014).}

\bibitem{Piani2015} M. Piani and J. Watrous, {\em Necessary and sufficient quantum information characterization of Einstein-Podolsky-Rosen steering}, \href{https://doi.org/10.1103/PhysRevLett.114.060404}{Phys. Rev. Lett. {\bf114}, 060404 (2015).}

\bibitem{Gallego2015} R. Gallego and L. Aolita, {\em Resource theory of steering}, \href{https://doi.org/10.1103/PhysRevX.5.041008}{Phys. Rev. X {\bf5}, 041008 (2015).}

\bibitem{Gour2008} G. Gour and R. W. Spekkens, {\em The resource theory of quantum reference frames: Manipulations and monotones}, \href{https://doi.org/10.1088/1367-2630/10/3/033023}{New J. Phys. {\bf10}, 033023 (2008).}

\bibitem{Marvian2016} I. Marvian and R. W. Spekkens, {\em How to quantify coherence: Distinguishing speakable and unspeakable notions}, \href{https://doi.org/10.1103/PhysRevA.94.052324}{Phys. Rev. A {\bf94}, 052324 (2016).}

\bibitem{Brandao2013} F. G. S. L. Brand$\tilde{\rm a}$o, M. Horodecki, J. Oppenheim, J. M. Renes, and R. W. Spekkens, {\em Resource theory of quantum states out of thermal equilibrium}, \href{https://doi.org/10.1103/PhysRevLett.111.250404}{Phys. Rev. Lett. {\bf111}, 250404 (2013).}

\bibitem{Horodecki2013} M. Horodecki and J. Oppenheim, {\em Fundamental limitations for quantum and nanoscale thermodynamics}, \href{https://doi.org/10.1038/ncomms3059}{Nat. Commun. {\bf 4}, 2059 (2013).}

\bibitem{Brandao2015} F. G. S. L. Brand$\tilde{\rm a}$o, M. Horodecki, N. Ng, J. Oppenheim, and S. Wehner, {\em The second laws of quantum thermodynamics}, \href{https://doi.org/10.1073/pnas.1411728112}{Proc. Natl. Acad. Sci. U.S.A. {\bf112}, 3275 (2015).}

\bibitem{Lostaglio2018} M. Lostaglio, {\em An introductory review of the resource theory approach to thermodynamics}, \href{https://doi.org/10.1088/1361-6633/ab46e5}{Rep. Prog. Phys. {\bf82}, 114001 (2019).}

\bibitem{Narasimhachar2019} V. Narasimhachar, S. Assad, F. C. Binder, J. Thompson, B. Yadin, and M. Gu, {\em Thermodynamic resources in continuous-variable quantum systems}, \href{https://arxiv.org/abs/1909.07364}{arXiv:1909.07364.}

\bibitem{Serafini2019} A. Serafini, M. Lostaglio, S. Longden, U. Shackerley-Bennett, C.-Y. Hsieh, and G. Adesso, {\em Gaussian thermal operations and the limits of algorithmic cooling}, \href{https://doi.org/10.1103/PhysRevLett.124.010602}{Phys. Rev. Lett. {\bf124}, 010602 (2020).}

\bibitem{Horodecki2013-2} M. Horodecki and J. Oppenheim, {\em(Quantumness in the context of) resource theories}, \href{https://doi.org/10.1142/S0217979213450197}{Int. J. Mod. Phys. B {\bf27}, 1345019 (2013).}

\bibitem{Brandao2015-2} F. G. S. L. Brand$\tilde{\rm a}$o and G. Gour, {\em Reversible framework for quantum resource theories}, \href{https://doi.org/10.1103/PhysRevLett.115.070503}{Phys. Rev. Lett. {\bf115}, 070503 (2015).}

\bibitem{del_Rio2015} L. del Rio, L. Kraemer, and R. Renner, {\em Resource theories of knowledge}, \href{https://arxiv.org/abs/1511.08818}{arXiv:1511.08818.}

\bibitem{Coecke2016} B. Coecke, T. Fritz, and R. W. Spekkens, {\em A mathematical theory of resources}, \href{https://doi.org/10.1016/j.ic.2016.02.008}{Inf. Comput. {\bf 250}, 59 (2016).}

\bibitem{Gour2017} G. Gour, {\em Quantum resource theories in the single-shot regime}, \href{https://doi.org/10.1103/PhysRevA.95.062314}{Phys. Rev. A {\bf95}, 062314 (2017).}

\bibitem{Liu2017} Z.-W. Liu, X. Hu, and S. Lloyd, {\em Resource destroying maps}, \href{https://doi.org/10.1103/PhysRevLett.118.060502}{Phys. Rev. Lett. {\bf118}, 060502 (2017).}

\bibitem{Anshu2018} A. Anshu, M.-H. Hsieh, and R. Jain, {\em Quantifying resources in general resource theory with catalysts}, \href{https://doi.org/10.1103/PhysRevLett.121.190504}{Phys. Rev. Lett. {\bf121}, 190504 (2018).}

\bibitem{RT-RMP} E. Chitambar and G. Gour, {\em Quantum resource theories}, \href{https://doi.org/10.1103/RevModPhys.91.025001}{Rev. Mod. Phys. {\bf91}, 025001 (2019).}

\bibitem{Lami2018} L. Lami, B. Regula, X. Wang, R. Nichols, A. Winter, and G. Adesso, {\em Gaussian quantum resource theories}, \href{https://doi.org/10.1103/PhysRevA.98.022335}{Phys. Rev. A {\bf98}, 022335 (2018).}

\bibitem{Regula2018} B. Regula, {\em Convex geometry of quantum resource quantification}, \href{https://doi.org/10.1088/1751-8121/aa9100}{J. Phys. A: Math. Theor. {\bf51}, 045303 (2018).}

\bibitem{Liu2019} Z.-W. Liu, K. Bu, and R. Takagi, {\em One-shot operational quantum resource theory}, \href{https://doi.org/10.1103/PhysRevLett.123.020401}{Phys. Rev. Lett. {\bf123}, 020401 (2019).}

\bibitem{Fang2019} K. Fang and Z.-W. Liu, {\em No-go theorems for quantum resource purification}. \href{https://arxiv.org/abs/1909.02540}{arXiv:1909.02540.}

\bibitem{Takagi2019}R. Takagi and B. Regula, {\em General resource theories in quantum mechanics and beyond: Operational characterization via discrimination tasks}, \href{https://doi.org/10.1103/PhysRevX.9.031053}{Phys. Rev. X {\bf9}, 031053 (2019).}

\bibitem{Takagi2019-2} R. Takagi, B. Regula, K. Bu, Z.-W. Liu, and G. Adesso, {\em Operational advantage of quantum resources in subchannel discrimination}, \href{https://doi.org/10.1103/PhysRevLett.122.140402}{Phys. Rev. Lett. {\bf122}, 140402 (2019).}

\bibitem{Bu2018} L. Li, K. Bu and Z.-W. Liu, {\em Quantifying the resource content of quantum channels: An operational approach}, \href{https://doi.org/10.1103/PhysRevA.101.022335}{Phys. Rev. A {\bf101}, 022335 (2020).}

\bibitem{Hsieh2017} J.-H. Hsieh, S.-H. Chen, and C.-M. Li, {\em Quantifying quantum-mechanical processes}, \href{https://doi.org/10.1038/s41598-017-13604-9}{Scientific Reports {\bf7}, 13588 (2017).}

\bibitem{Kuo2018} C.-C. Kuo, S.-H. Chen, W.-T. Lee, H.-M. Chen, H. Lu, and C.-M. Li, {\em Quantum process capability}, \href{https://doi.org/10.1038/s41598-019-56751-x}{Scientific Reports {\bf9}, 20316 (2019).}

\bibitem{Dana2017} K. B. Dana, M. G. D\'iaz, M. Mejatty, and A. Winter, {\em Resource theory of coherence: Beyond states}, \href{https://doi.org/10.1103/PhysRevA.95.062327}{Phys. Rev. A {\bf95}, 062327 (2017).}

\bibitem{Pirandola2017} S. Pirandola, R. Laurenza, C. Ottaviani, and L. Banchi, {\em Fundamental limits of repeaterless quantum communications}, \href{https://doi.org/10.1038/ncomms15043}{Nat. Commun. {\bf8}, 15043 (2017).}

\bibitem{Diaz2018} M. G. D\'iaz, K. Fang, X. Wang, M. Rosati, M. Skotiniotis, J. Calsamiglia, and A. Winter, {\em Using and reusing coherence to realize quantum processes}, \href{https://doi.org/10.22331/q-2018-10-19-100}{Quantum {\bf2}, 100 (2018).}

\bibitem{Rosset2018} D. Rosset, F. Buscemi, and Y.-C. Liang, {\em A resource theory of quantum memories and their faithful verification with minimal assumptions}, \href{https://doi.org/10.1103/PhysRevX.8.021033}{Phys. Rev. X {\bf8}, 021033 (2018).}

\bibitem{Wilde2018} M. M. Wilde, {\em Entanglement cost and quantum channel simulation}, \href{https://doi.org/10.1103/PhysRevA.98.042338}{Phys. Rev. A {\bf98}, 042338 (2018).}

\bibitem{Zhuang2018} Q. Zhuang, P. W. Shor, and J. H. Shapiro, {\em Resource theory of non-Gaussian operations}, \href{https://doi.org/10.1103/PhysRevA.97.052317}{Phys. Rev. A {\bf97}, 052317 (2018).}

\bibitem{Bauml2019} S. B$\ddot{\rm a}$uml, S. Das, X. Wang, and M. M. Wilde, {\em Resource theory of entanglement for bipartite quantum channels}, \href{https://arxiv.org/abs/1907.04181}{arXiv:1907.04181.}

\bibitem{Seddon2019} J. R. Seddon and E. Campbell, {\em Quantifying magic for multi-qubit operations}, \href{https://doi.org/10.1098/rspa.2019.0251}{Proc. R. Soc. A {\bf475}, 20190251 (2019).}

\bibitem{LiuWinter2019} Z.-W. Liu and A. Winter, {\em Resource theories of quantum channels and the universal role of resource erasure}, \href{https://arxiv.org/abs/1904.04201}{arXiv:1904.04201.}

\bibitem{LiuYuan2019} Y. Liu and X. Yuan, {\em Operational resource theory of quantum channels}, \href{https://doi.org/10.1103/PhysRevResearch.2.012035}{Phys. Rev. Research {\bf2}, 012035(R) (2020).}

\bibitem{Gour2019-3} G. Gour, {\em Comparison of quantum channels by superchannels}, \href{https://doi.org/10.1109/TIT.2019.2907989}{IEEE Trans. Inf. Theory {\bf65}, 5880 (2019).}

\bibitem{Gour2019} G. Gour and A. Winter, {\em How to quantify a dynamical resource?} \href{https://doi.org/10.1103/PhysRevLett.123.150401}{Phys. Rev. Lett. {\bf123}, 150401 (2019).}

\bibitem{Gour2019-2} G. Gour and C. M. Scandolo, {\em The entanglement of a bipartite channel}, \href{https://arxiv.org/abs/1907.02552}{arXiv:1907.02552.}

\bibitem{Takagi2019-3} R. Takagi, K. Wang, and M. Hayashi, {\em Application of a resource theory of channels to communication scenarios}, \href{https://arxiv.org/abs/1910.01125}{arXiv:1910.01125v1.}

\bibitem{Theurer2019} T. Theurer, D. Egloff, L. Zhang, and M. B. Plenio, {\em Quantifying operations with an application to coherence}, \href{https://doi.org/10.1103/PhysRevLett.122.190405}{Phys. Rev. Lett. {\bf122}, 190405 (2019).}

\bibitem{Wang2019} X. Wang and M. M. Wilde, {\em Resource theory of asymmetric distinguishability for quantum channels}, \href{https://doi.org/10.1103/PhysRevResearch.1.033169}{Phys. Rev. Research {\bf1}, 033169 (2019).}

\bibitem{Berk2019} G. D. Berk, A. J. P. Garner, B. Yadin, K. Modi, and F. A. Pollock, {\em Resource theories of multi-time processes:
A window into quantum non-Markovianity}, \href{https://arxiv.org/abs/1907.07003}{arXiv:1907.07003.}

\bibitem{Hsieh2019} C.-Y. Hsieh, M. Lostaglio, and A. Ac\'in, {\em Entanglement preserving local thermalization}, \href{https://doi.org/10.1103/PhysRevResearch.2.013379}{Phys. Rev. Research {\bf2}, 013379 (2020).}

\bibitem{Moravcíkova2010} L. Morav${\check{\rm c}}$\'ikov\'a and M. Ziman, {\em Entanglement-annihilating and entanglement-breaking channels}, \href{https://doi.org/10.1088/1751-8113/43/27/275306}{J Phys. A: Math. Theor. {\bf43}, 275306 (2010).}

\bibitem{Horodecki1998} M. Horodecki, P. Horodecki, and R. Horodecki, {\em Mixed-state entanglement and distillation: Is there a ``bound'' entanglement in nature?}, \href{https://doi.org/10.1103/PhysRevLett.80.5239}{Phys. Rev. Lett. {\bf80}, 5239 (1998).}

\bibitem{Chiribella2008} G. Chiribella, G. M. D'Ariano, and P. Perinotti, {\em Transforming quantum operations: Quantum supermaps}, \href{https://doi.org/10.1209/0295-5075/83/30004}{EPL (Europhysics Letters) {\bf83}, 30004 (2008).}

\bibitem{Chiribella2008-2} G. Chiribella, G. M. D'Ariano, and P. Perinotti, {\em Quantum circuit architecture}, \href{https://doi.org/10.1103/PhysRevLett.101.060401}{Phys. Rev. Lett. {\bf101}, 060401 (2008).}

\bibitem{Palazuelos2012} C. Palazuelos, {\em Superactivation of quantum nonlocality}, \href{https://doi.org/10.1103/PhysRevLett.109.190401}{Phys. Rev. Lett. {\bf109}, 190401 (2012).}

\bibitem{Hsieh2016} C.-Y. Hsieh, Y.-C. Liang, and R.-K. Lee, {\em Quantum steerability: Characterization, quantification, superactivation and unbounded amplification}, \href{https://doi.org/10.1103/PhysRevA.94.062120}{Phys. Rev. A {\bf94}, 062120 (2016).}

\bibitem{Quintino2016}M. T. Quintino, M. Huber, and N. Brunner, {\em Superactivation of quantum steering}, \href{https://doi.org/10.1103/PhysRevA.94.062123}{Phys. Rev. A {\bf94}, 062123 (2016).}

\bibitem{Masanes2008} Ll. Masanes, Y.-C. Liang, and A. C. Doherty, {\em All bipartite entangled states display some hidden nonlocality}, \href{https://doi.org/10.1103/PhysRevLett.100.090403}{Phys. Rev. Lett. {\bf100}, 090403 (2008).}

\bibitem{Liang2012} Y.-C. Liang, Ll. Masanes, and D. Rosset, {\em All entangled states display some hidden nonlocality}, \href{https://doi.org/10.1103/PhysRevA.86.052115}{Phys. Rev. A {\bf86}, 052115 (2012).}

\bibitem{Horodecki2003} M. Horodecki, P. W. Shor, and M. B. Ruskai, {\em Entanglement breaking channels}, \href{https://doi.org/10.1142/S0129055X03001709}{Rev. Math. Phys. {\bf15}, 629 (2003).}

\bibitem{Datta2009} N. Datta, {\em Min- and max-relative entropies and a new entanglement monotone}, \href{https://doi.org/10.1109/TIT.2009.2018325}{IEEE Trans. Inf. Theory {\bf55}, 2816 (2009).}

\bibitem{Saxena2019} G. Saxena, E. Chitambar,and G. Gour, {\em Dynamical resource theory of quantum coherence}, \href{https://arxiv.org/abs/1910.00708}{arXiv:1910.00708.}

\bibitem{Sparaciari2019} C. Sparaciari, M. Goihl, P. Boes, J. Eisert, and N. Ng, {\em Bounding the resources for thermalizing many-body localized systems}, \href{https://arxiv.org/abs/1912.04920}{arXiv:1912.04920.}

\bibitem{Cavalcanti2013} D. Cavalcanti, A. Acin, N. Brunner, and T. Vertesi, {\em All quantum states useful for teleportation are nonlocal resources}, \href{https://doi.org/10.1103/PhysRevA.87.042104}{Phys. Rev. A {\bf87}, 042104 (2013).}

\bibitem{Albeverio2002} S. Albeverio, S.-M. Fei, and W.-L. Yang, {\em Optimal teleportation based on bell measurements}, \href{https://doi.org/10.1103/PhysRevA.66.012301}{Phys. Rev. A {\bf66}, 012301 (2002).}

\bibitem{Zhao2010} M.-J. Zhao, Z.-G. Li, S.-M. Fei, and Z.-X. Wang, {\em A note on fully entangled fraction}, \href{https://doi.org/10.1088/1751-8113/43/27/275203}{J. Phys. A: Math. Theor. {\bf 43}, 275203 (2010).}
	
\bibitem{Hsieh2018E} C.-Y. Hsieh and R.-K. Lee, {\em Work extraction and fully entangled fraction}, \href{https://doi.org/10.1103/PhysRevA.96.012107}{Phys. Rev. A {\bf96}, 012107 (2017).}

\bibitem{Liang2019} Y.-C. Liang, Y.-H. Yeh, P. E. M. F. Mendon\c{c}a, R. Y. Teh, M. D. Reid, and P. D. Drummond, {\em Quantum fidelity measures for mixed states}, \href{https://doi.org/10.1088/1361-6633/ab1ca4}{Rep. Prog. Phys. {\bf82}, 076001 (2019).}

\bibitem{Horodecki1999} M. Horodecki and P. Horodecki, {\em Reduction criterion of separability and limits for a class of distillation protocols}, \href{https://doi.org/10.1103/PhysRevA.59.4206}{Phys. Rev. A {\bf59}, 4206 (1999).}

\bibitem{Bennett1996} C. H. Bennett, G. Brassard, S. Popescu, B. Schumacher, J. A. Smolin, and W. K. Wootters, {\em Purification of noisy entanglement and faithful teleportation via noisy channels}, \href{https://doi.org/10.1103/PhysRevLett.76.722}{Phys. Rev. Lett. {\bf76}, 722 (1996).}

\bibitem{Almeida2007} M. L. Almeida, S. Pironio, J. Barrett, G. T\'oth, and A. Ac\'in, {\em Noise robustness of the nonlocality of entangled quantum states}, \href{https://doi.org/10.1103/PhysRevLett.99.040403}{Phys. Rev. Lett. {\bf99}, 040403 (2007).}



\end{thebibliography}
\end{document}